\documentclass[10pt]{article}

\pdfoutput=1

\usepackage[latin1]{inputenc}
\usepackage{mystyle}
\usepackage{dsfont}
\usepackage{url}
\usepackage{notations}
\graphicspath{{./figures/}}
\newcommand{\myauthor}[3]{#1\protect\footnote{#2, \protect\url{#3}}}

\title{Support Recovery for Sparse \\ Deconvolution of Positive Measures}

\author{\myauthor{Quentin Denoyelle}{CEREMADE, Univ. Paris-Dauphine}{denoyelle@ceremade.dauphine.fr}, \quad
		\myauthor{Vincent Duval}{INRIA Rocquencourt, MOKAPLAN}{vincent.duval@inria.fr} \quad and \quad
		\myauthor{Gabriel Peyr{\'e}}{CNRS \& CEREMADE, Univ. Paris-Dauphine}{gabriel.peyre@ceremade.dauphine.fr}
		}

\date{\today}

\begin{document}

\maketitle

 

\begin{abstract}
We study sparse spikes deconvolution over the space of Radon measures on $\RR$ or $\TT$ when the input measure is a finite sum of positive Dirac masses using the BLASSO convex program. We focus on the recovery properties of the support and the amplitudes of the initial measure in the presence of noise as a function of the minimum separation $t$ of the input measure (the minimum distance between two spikes). We show that when ${w}/\lambda$, ${w}/t^{2N-1}$ and $\lambda/t^{2N-1}$ are small enough (where $\lambda$ is the regularization parameter, $w$ the noise and $N$ the number of spikes), which corresponds roughly to a sufficient signal-to-noise ratio and a noise level small enough with respect to the minimum separation, there exists a unique solution to the BLASSO program with exactly the same number of spikes as the original measure. 
We show that the amplitudes and positions of the spikes of the solution both converge toward those of the input measure when the noise and the regularization parameter drops to zero faster than $t^{2N-1}$.
\end{abstract}


\section{Introduction}

\subsection{Super-Resolution and Sparse Spikes Deconvolution}

Super-resolution consists in retrieving fine scale details of a possibly noisy signal from coarse scale information. The  importance of recovering the high frequencies of a signal comes from the fact that there is often a physical blur in the acquisition process, such as diffraction in optical systems, wave reflection in seismic imaging or spikes recording from neuronal activity. 

In resolution theory, the two-point resolution criterion defines the ability of a system to resolve two points of equal 
intensities. As a point source produces a diffraction pattern which is centered about the geometrical image point of the point source, it is often admitted that two points are resolved by a system if the central maximum of the intensity diffraction of one point source coincides with the first zero of the intensity diffraction pattern of the other 
point. This defines a distance that only depends on the system and which is called the Rayleigh Length. In the case of the ideal low-pass filter (meaning that the input signal is convolved with the Dirichlet kernel, see~\eqref{eq-idealLPF} for the exact definition) with cutoff frequency $f_c$, the Rayleigh Length is $1/f_c$. We refer to~\cite{den-resolution1997} for more details about resolution theory. Super-resolution in signal processing thus consists in developing techniques that enable to retrieve information below the Rayleigh  Length.


Let us introduce in a more formal way the problem which will be the core of this article. Let $\Pos$ be the real line $\RR$ or the 1-D torus $\TT=\RR/\ZZ$ and $\radon$ the Banach space of bounded Radon measures on $\Pos$, which can be seen as the topological dual of the space $\ContX$ where $\ContX$ is either the space of continuous functions on $\RR$ that vanish at infinity when $\Pos=\RR$ or the space of continuous functions on $\TT$ when $\Pos=\TT$. We consider a given integral operator $\Phi:\radon\to\Obs$, where $\Obs$ is a separable Hilbert space, whose kernel $\phi$ is supposed to be a smooth function (see Definition~\ref{def-kernel} for the technical assumptions made on $\phi$), \ie 
\begin{align}\label{def-Phi}
	\forall m\in\radon,\quad\Phi m=\int_\Pos \phi(\pos)\d m(\pos).
\end{align}
$\Phi$ represents the acquisition operator and can for instance account for a blur in the measurements. In the special case of $\phi(x)=\tilde{\phi}(\cdot-x)$, $\Phi$ is a convolution operator. We denote by $\meastO=\sum_{i=1}^{N} \ampOi \dirac{\postOi}$ our input sparse spikes train where the $\ampOi\in\RR_+^*$ are the amplitudes of the Dirac masses at positions $\postOi\in\Pos$. Let $\obs=\Phi\meastO$ be the noiseless observation. The parameter $t>0$ controls the minimum separation distance between the spikes, and this paper aims at studying the recovery of $\meastO$ from $\obs+w$ (where $w\in\Obs$ is some noise) when $t$ is small. 

\subsection{From the LASSO to the BLASSO}

\paragraph{LASSO.}

$\lun$ regularization techniques were first introduced in geophysics (see \cite{claerbout-robust1973,levy-reconstruction1981,santosa-linear1986}) for seismic prospecting. Indeed, the density changes in the underground can be modeled as a sparse spikes train. $\lun$ reconstruction property provides solutions with few nonzero coefficients and can be solved efficiently with convex optimization methods. Donoho theoretically studied and justified these techniques in~\cite{donoho-superresolution1992}.
In statistics, the $\lun$ norm is used in the Lasso method~\cite{tibshirani-regression1994} which consists in minimizing a quadratic error subject to a $\lun$ penalization. As the authors remarked it, it retains both the features of subset selection (by setting to zero some coefficients, thanks to the property of the $\lun$ norm to favor sparse solutions) and ridge regression (by shrinking the other coefficients).
In signal processing, the basis pursuit method~\cite{chen-atomic1998} uses the $\lun$ norm to decompose signals into overcomplete dictionaries.

\paragraph{BLASSO.}

Following recent works (see for instance \cite{bhaskar-atomic2011,bredies-inverse2013,candes-towards2013,deCastro-exact2012,duval-exact2013}), the sparse deconvolution method that we consider in this article operates over a continuous domain, i.e. without resorting to some sort of discretization on a grid. The inverse problem is solved over the space of Radon measures which is a non-reflexive Banach space. This continuous ``grid free'' setting makes the mathematical analysis easier and allows us to make precise statement about the location of the recovered spikes locations. 

The technique that we study in this paper consists in solving a convex optimization problem that uses the total variation norm which is the equivalent of the $\lun$ norm for measures. The $\lun$ norm is known to be particularly well fitted for the recovery of sparse signals. Thus the use of the total variation norm favors the emergence of spikes in the solution.

The total variation norm is defined by
\eq{
  	\forall m \in \radon, \quad
	\normTVX{m}	\eqdef	\usup{\psi \in \ContX} \enscond{ \int_\Pos \psi \d m }{ \normLi{\psi} \leq 1 }.
}
In particular,
\eq{
	\normTVX{\meastO}=\normu{\ampO},
}
which shows in a way that the total variation norm generalizes the $\lun$ norm to the continuous setting of measures (i.e. no discretization grid is required).

The first method that we are interested in is the classical basis pursuit, defined originally in~\cite{chen-atomic1998} in a finite dimensional setting, and written here over the space of Radon measures
\begin{align}\label{eq-blasso-noiseless}
  \umin{m \in \radon} 
  	\enscond{ \normTVX{m} }{\Phi m=\obs} \tag{$\bpursuit$}.
\end{align}
This is the problem studied in~\cite{candes-towards2013}, in the case where $\Phi$ is an ideal low-pass filter on the torus (i.e. $\Pos=\TT$).

When the signal is noisy, \ie when we observe $\obs+w$ instead, with $w\in \Obs$, we may rather consider the problem 
\begin{align}\label{eq-blasso-noisy}
	\umin{m \in \radon} \frac{1}{2} \normObs{\Phi m - (\obs+w)}^2 + \lambda \normTVX{m} \tag{$\blasso$}.
\end{align}
Here $\lambda>0$ is a parameter that should adapted to the noise level $\normObs{w}$. This problem is coined ``BLASSO'' in~\cite{deCastro-exact2012}.

While this is not the focus of this article, we note that there exist algorithms to solve the infinite dimensional convex problems~\eqref{eq-blasso-noiseless} and~\eqref{eq-blasso-noisy}, see for instance~\cite{bredies-inverse2013,candes-towards2013}.

\subsection{Previous Works}
\label{sec-previous-works}

\paragraph{LASSO/BLASSO performance analysis.}

In order to quantify the recovery performance of the methods $\bpursuit$ and $\blasso$, the following two questions arise:
\begin{enumerate}
 	\item Does the solutions of $\bpursuit$ recover the input measure $\meastO$ ?
 	\item How close is the solution of $\blasso$ to the solution of $\bpursuit$ ?
\end{enumerate}

When the amplitudes of the spikes are arbitrary complex numbers, the answers to the above questions require a large enough minimum separation distance $\De(\postO)$ between the spikes where 
\begin{align}\label{min-sep-dist}
	\De(\postO) \eqdef \umin{i \neq j} \dX(\postOi,t z_{0,j}).
\end{align}
where $\dX$ is either the canonical distance on $\RR$ \ie
\begin{align}\label{dist-R}
	\forall x,y\in\RR, \quad \dX(x,y)=|x-y|,
\end{align}
when $\Pos=\RR$, or the canonical induced distance on $\TT$ \ie 
\begin{align}\label{dist-T}
	\forall x,y\in\RR, \quad \dX(x+\ZZ,y+\ZZ)=\min_{k\in\ZZ} |x-y+k|,
\end{align}
when $\Pos=\TT$.
The first question is addressed in~\cite{candes-towards2013} where the authors showed, in the case of $\Phi$ being the ideal low-pass filter on the torus (see~\eqref{eq-idealLPF}), i.e. when $\Obs=\TT$, that $\meastO$ is the unique solution of 
$\bpursuit$ provided that $\De(\postO)\geq \frac{C}{f_c}$ where $C>0$ is a universal constant and $f_c$ the cutoff frequency of the ideal low-pass filter. In the same paper, it is shown that 
$C\leq 2$ when $\ampO\in\CC^N$ and $C\leq 1.87$ when $\ampO\in \RR^N$. In~\cite{duval-exact2013}, the authors show that necessarily $C\geq \frac{1}{2}$.

The second question receives partial answers in~\cite{azais-spike2014,bredies-inverse2013,candes-superresolution2013,fernandez-support2013}. In~\cite{bredies-inverse2013}, it is shown that if the solution of $\bpursuit$ is unique then the measures recovered by $\blasso$ converge in the 
weak-* sense to the solution of $\bpursuit$ when $\lambda \rightarrow 0$ and $\normObs{w}/\lambda \rightarrow 0$. In~\cite{candes-superresolution2013}, the authors measure the reconstruction error using the $\Ldeux$ norm of an ideal low-pass filtered version of the recovered measures. 
In~\cite{azais-spike2014,fernandez-support2013}, error bounds are given on the locations of the recovered spikes with respect to those of the input measure $\meastO$.
However, those works provide little information about the structure of the measures recovered by $\blasso$. That point is addressed in~\cite{duval-exact2013} where the authors 
show that under the \emph{Non Degenerate Source Condition} (see Section~\ref{propdef-vanishing} for more details), there exists a unique solution to $\blasso$ with the exact same number of spikes as the original measure provided that $\lambda$ and $\normObs{w}/\lambda$ are small enough. Moreover in that regime, this solution converges to the original measure when the noise drops to zero.

\paragraph{LASSO/BLASSO for positive spikes.}

For positive spikes (i.e. $\ampOi>0$), the picture is radically different, since exact recovery of $\meastO$ without noise (i.e. $(w,\la)=(0,0)$) holds for all $t>0$, see for instance~\cite{deCastro-exact2012}. Stability constants however explode as $t \rightarrow 0$. 
A recent work~\cite{candes-stable2014} shows however that stable recovery is obtained if the signal-to-noise ratio grows faster than $O(1/t^{2N})$, closely matching optimal lower bounds of $O(1/t^{2N-1})$ obtained by combinatorial methods, as also proved recently~\cite{demanet-recoverability2014}. Our main contribution is to show that the same $O(1/t^{2N-1})$ signal-to-noise scaling in fact guarantees a perfect support recovery of the spikes under a certain non-degeneracy condition on the filter. This extends, for positive measures, the initial results of~\cite{duval-exact2013} by providing an asymptotic analysis when $t \rightarrow 0$.

\paragraph{MUSIC and related methods.}

There is a large body of literature in signal processing on spectral methods to perform spikes location from low frequency measurements. One of the most popular methods is MUSIC (for MUltiple Signal Classification)~\cite{schmidt-multiple1986} and we refer to~\cite{krim-two1996} for an overview of its use in signal processing for line spectral estimation. 
In the noiseless case, exact reconstruction of the initial signal is guaranteed as long as there are enough observations compared to the number of distinct frequencies~\cite{liao-music2014}. Stability to noise is known to hold under a minimum separation distance similar to the one of the BLASSO~\cite{liao-music2014}. 
However, on sharp contrast with the behavior of the BLASSO, numerical observations (see for instance~\cite{condat-cadzow2013}), as well as a recent work of Demanet and Nguyen, show that this stability continues to hold regardless of the sign of the amplitudes $\ampOi$, as soon as the signal-to-noise ratio scales like $O(1/t^{2N-1})$. 
Note that this matches (when $w$ is a Gaussian white noise) the Cramer-Rao lower bound achievable by any unbiased estimator~\cite{vetterli-sparse2008}. 
This class of methods are thus more efficient than BLASSO for arbitrary measures, but they are restricted to operators $\Phi$ that are convolution with a low-pass filter, which is not the case of our analysis for the~BLASSO.

\subsection{Contributions}
\label{intro-contributions}

\paragraph{Main results.}

From these previous works, one can ask whether exact support estimation by BLASSO for positive spikes is achievable when $t$ tends to 0.
Our main result, Theorem~\ref{thm-main}, shows that this is indeed the case. It states, under some non-degeneracy condition on $\Phi$, that there exists a unique solution to $\blasso$ with the exact same number of spikes as the original measure provided that $\normObs{w}/\lambda$, 
$\normObs{w}/t^{2N-1}$ and $\la/t^{2N-1}$ are small enough. 
Moreover we give an upper bound, in that regime, on the error of the recovered measure with respect to the initial measure.
As a by-product, we show that the amplitudes and positions of the spikes of the solution both converge towards those of the input measure when the noise and the regularization parameter tend to zero faster than $t^{2N-1}$.

\paragraph{Extensions.}\label{extensions}

%
We consider in this article the case where all the spikes locations $t z_0$ cluster near zero. Following for instance~\cite{candes-stable2014}, it is possible to consider a more general model with several cluster points, where the sign of the Diracs is the same around each of these points. Our analysis, although more difficult to perform, could be extended to this setting, at the price of modifying the definition of $\etaW$ (see Definition~\ref{def-nvanishing}) to account for several cluster points. 

Lastly, if the kernel $\Phi$ under consideration has some specific scale $\si$ (such as the standard deviation of a Gaussian kernel, or $\si=1/f_c$ for the ideal low-pass filter in the case of the deconvolution on the torus), then it is possible to state our contribution by replacing $t$ by the dimensionless quantity $\text{SRF} \eqdef t/\si$ (called ``super-resolution factor'' in~\cite{candes-stable2014}). It is then possible to extend our proof so show that the signal-to-noise ratio should obey the scaling $1/\text{SRF}^{2N-1}$.

\paragraph{Roadmap.}

The exact statement of Theorem~\ref{thm-main} (our main contribution), and in particular the definition of the non-degeneracy condition, requires some more background, which is the subject of Section~\ref{sec-exact-support-recov}.
The proof of this result can be found in Sections~\ref{sec-prelim},~\ref{sec-tfi} and~\ref{sec-etaL}.
It relies on an independent study of the asymptotic behavior of quantities depending on the operator $\PhitzO$ (such as its pseudo-inverse) when $t$ tends to zero. This takes place in Section \ref{sec-prelim}. Note that a sketch of proof of the main result can be found in Section~\ref{sec-main-contrib}.

\subsection{Notations}

\paragraph{Measures.}

We consider $\Pos=\RR$ or $\Pos=\TT$ as the space of Dirac masses positions. $\Pos$ equipped with the distance $\dX$ (see~\eqref{dist-R} and~\eqref{dist-T}) is a locally compact metric space (compact in the case $\Pos=\TT$). We denote by $\radon$ the space of bounded Radon measures on $\Pos$. It is the topological dual of the Banach space $\ContX$ (endowed with $\normLi{\cdot}$) of continuous functions defined on $\Pos$, that furthermore are imposed to vanish at infinity in the case $\Pos=\RR$. The two problems that we study, \ie the Basis Pursuit for measures $\bpursuit$ and the BLASSO $\blasso$, are two convex optimization problems on the space $\radon$.

Let us consider $\ampO\in(\RR_+^*)^N$ and $\poszO\in\RR^N$. When $\Pos=\TT$, we make the assumption that $\poszO\in(-\frac{1}{4},\frac{1}{4})^N$. We define
\begin{align}\label{min-sep-dist-ini}
	\De_0 \eqdef 
	\De(\poszO),
\end{align}
where $\De(\poszO)$ is introduced in~\eqref{min-sep-dist}.
We denote by $\bball{\bar{x}}{r}$ (resp. $\ball{\bar{x}}{r}$) the closed (resp. open) $\ell^\infty$ ball in $\RR^N$ with center $\bar{x}$ and radius $r$. We define neighborhoods of respectively $\poszO$ and $\ampO$ as
\begin{align}\label{vois-ini}
  \Bo \eqdef\bball{\poszO}{\frac{\De_0}{4}}
  \qandq
  \Ba \eqdef\bball{\ampO}{\frac{\min_i (\ampOi)}{4}}.
\end{align}
Note that when $\Pos=\TT$, if $\posz\in\Bo$ then $\posz\in (-\frac{1}{2},\frac{1}{2})^N$ and for all $t\in(0,1]$, $\post\in (-\frac{1}{2},\frac{1}{2})^N$. As a result, depending on the context, we consider $\postO$ and $\post$ as elements of $\TT^N$ by identifying them with their equivalent classes.

Then, the initial measure we want to recover is of the form
\eq{
	\meastO\eqdef\sum_{i=1}^N \ampOi\dirac{\postOi}.
}
Note that for example, in the case $\Pos=\TT$, when we write $\dirac{\postOi}$, $\postOi$ is considered as an element of $\TT$.

We use the parameter $t\in(0,1]$ to make the spikes locations tend to $0$ in $\Pos$, so that the minimum separation distance of $\meastO$ : $\De(\postO)\to0$ when $t\to0$.

\paragraph{Kernels.}
The admissible kernels $\phi:\Pos\to\Obs$ defining the integral operators $\Phi$ (modeling the blur of acquisition) are functions, defined on $\Pos$ and taking values in a separable Hilbert space $\Obs$, satisfying some regularity properties listed in the following Definition.

\begin{defn}[Admissible kernels]\label{def-kernel}
We denote by $\kernel{k}$, the set of admissible kernels of order $k$. A function $\phi:\Pos\to\Obs$ belongs to $\kernel{k}$ if $\phi\in\Cder{k}(\Pos,\Obs)$. When $\Pos=\RR$, $\phi$ must also satisfies the following requirements:
\begin{itemize}
\item For all $p\in\Obs$, $\dotObs{\phi(\pos)}{p}\to0$ when $|\pos |\to +\infty$.
\item For all $0\leq i\leq k$, $\underset{\pos\in\Obs}{\sup} \normObs{\phi^{(i)}(x)} <+\infty$.
\end{itemize}
\end{defn}


\paragraph{Linear operators.}

We consider a linear operator $\Phi: \radon\rightarrow \Obs$ of the form
\begin{align}\label{eq-defnPhi}
	\forall m\in\radon, \quad
	\Phi m &\eqdef \int_\Pos \varphi(\pos)\d m(\pos),
\end{align}
where $\varphi: \Pos\rightarrow \Obs$ belongs to $\kernel{0}$. $\Phi$ is weak-* to weak continuous.
Its adjoint  $\Phi^*: \Obs\rightarrow \ContX$ is given by
\begin{align*}
	\forall p\in\Obs, \quad
	\forall \pos\in\Pos, \quad 
	(\Phi^*p)(\pos) = \dotObs{\phi(\pos)}{p}.
\end{align*}

A typical example is a convolution operator, where $\varphi(\pos)= \tilde{\varphi}(\cdot-x)$, for some continuous function $\tilde{\phi}:\Pos\rightarrow \RR$, $\tilde{\phi}\in \Ldeux(\Pos)\cap \ContX$, so that $(\Phi m)(\varV) = \int_\Pos \tilde{\varphi}(\varV-\pos)\d m(\pos)$. 
In particular, in the case $\Pos=\RR$, the Gaussian filter is defined as
\begin{align}\label{eq-gaussianFilter}
  \forall x\in\RR,\quad  \tilde{\varphi}(x)=\tilde{\phi}_G (x)&\eqdef e^{-x^2/2},
\end{align}
where $\tilde{\phi}_G$ is the Gaussian kernel.
In the case $\Pos=\TT$, a typical example of convolution operator is the \textit{ideal low pass filter} which is defined as
\begin{align}\label{eq-idealLPF}
  \forall \pos\in\TT,\quad  \tilde{\phi}(\pos)=\tilde{\phi}_D (\pos)&\eqdef \sum_{k=-f_c}^{f_c} e^{2\imath\pi k\pos },
\end{align}
for some $f_c\in\NN^*$ called the cutoff frequency. $\tilde{\phi}_D$ is the Dirichlet kernel with cutoff frequency $f_c$.
\begin{rem}
This last example is equivalently obtained (as considered for instance in~\cite{candes-towards2013}) by using
$\Obs=\CC^{2f_c+1}$ (endowed with its canonical inner product) in place of $\Obs=L^2(\TT)$, and defining $\phi$ as $\phi(\pos)=(e^{2i\pi k \pos})_{-f_c\leq k \leq f_c}$. Note that, to simplify the notation, we consider in this paper real Hilbert spaces $\Hh$, but our analysis readily extends to complex Hilbert spaces.
\end{rem}

Given general $\phi$ and $\Phi$ as in~\eqref{eq-defnPhi}, and given $\bar{\pos}=(\pos_1,\ldots,\pos_N)\in\Pos^N$, we denote by $\Phi_{\bar{\pos}}:\RR^N\rightarrow \Obs$ the linear operator such that 
\eq{
	\forall a\in \RR^N, \quad	
	\Phi_{\bar{\pos}}(a) \eqdef \sum_{i=1}^{N} a_i \phi(\pos_i),
}
and by $\Ga_{\bar{\pos}}:(\RR^{N}\times\RR^N)\rightarrow \Obs$ the linear operator defined by
\eq{
	\Ga_{\bar{\pos}}\begin{pmatrix}a\\b\end{pmatrix} \eqdef 
	\sum_{i=1}^{N}\left( a_i \phi(\pos_i) + b_i \varphi'(\pos_i)\right).
}

For $\phi \in \kernel{k}$, operators involving the derivatives are defined similarly,
 \eq{
   \forall 0\leq i\leq k, \quad (\Phi_{\bar{\pos}})^{(k)} : \amp \in \RR^N  \longmapsto (\Phi_{\bar{\pos}})^{(k)}(\amp) = \sum_{i=1}^{N} \amp_i \varphi^{(i)}(\pos_i).
 }

We occasionally write $\Phi_{\bar{\pos}}'$ (resp. $\Phi_{\bar{\pos}}''$) for $(\Phi_{\bar{\pos}})^{(1)}$ (resp. for $(\Phi_{\bar{\pos}})^{(2)}$), and we adopt the following matricial notation 
\eq{
	\Phi_{\bar{\pos}}=(\phi(\pos_1) \ \ldots \ \phi(\pos_{N})) \qandq \Ga_{\bar{\pos}}=(\Phi_{\bar{\pos}} \ \Phi_{\bar{\pos}}'),
}
where the ``columns'' of those ``matrices'' are elements of $\Obs$. In particular, the adjoint operator $\Phi_{\bar{\pos}}^*:\Obs\rightarrow \RR^N$ is given by
\eq{
  \forall p\in\Obs,\quad  \Phi_{\bar{\pos}}^*p=\left((\Phi^*p)(\pos_i)\right)_{1\leq i\leq N}.
}

We denote by $\phiD{k} \in \Obs$ the $k^{th}$ derivative of $\phi$ at $0$, \textit{i.e.}
\eql{\label{eq-defn-phider}
\phiD{k} \eqdef \phi^{(k)}(0).
}
In particular, $\phiD{0} = \phi(0)$.

Given $k \in \NN$, we define
\eql{\label{eq-defn-Fk}
	\Fk \eqdef \begin{pmatrix} \phiD{0} & \phiD{1} & \ldots & \phiD{k} \end{pmatrix}.
}
If $\Fk: \RR^{k+1}\rightarrow \Obs$ has full column rank, we define its pseudo-inverse as $\Fk^+\eqdef (\Fk^*\Fk)^{-1}\Fk^*$. Similarly, we denote $\Gamma_{\bar{\pos}}^+ \eqdef (\Gamma_{\bar{\pos}}^*\Gamma_{\bar{\pos}})^{-1}\Gamma_{\bar{\pos}}^*$  provided $\Gamma_{\bar{\pos}}$ has full column rank.

\paragraph{Linear algebra.}

For $\posz \in \RR^N$, we let
\eql{\label{eq-defn-Hx}
	\Hz \eqdef\begin{pmatrix}
  1  & \ldots & 1 & 0 &\ldots & 0\\
  \posz_1 & \ldots & \posz_N & 1 & \ldots & 1\\
  \vdots & & \vdots & \vdots & & \vdots\\
  \frac{(\posz_1)^{2N-1}}{(2N-1)!} & \ldots &   \frac{(\posz_N)^{2N-1}}{(2N-1)!}  &   
  \frac{(\posz_1)^{2N-2}}{(2N-2)!}  & \ldots &   \frac{(\posz_N)^{2N-2}}{(2N-2)!}  
\end{pmatrix}
\in \RR^{2N \times 2N},
}
so that $\Hz^{*,-1}$ is the matrix of the Hermite interpolation at points $\posz_1,\ldots \posz_N$ when $\RR_{2N-1}[X]$ is equiped with the basis $\left(1,X,\ldots, \frac{X^{2N-1}}{(2N-1)!}\right)$.

For each $N \in \NN$, we define
\begin{align}
  \dirac{N} &\eqdef (1,0,\ldots,0)^T \in \RR^N,\label{eq-dirac}\\
  \bun_N &\eqdef (1,1,\ldots, 1)^T \in \RR^N. \label{eq-bun}
\end{align}

We use the $\ell^{\infty}$ norm, $\normLiVec{\cdot}$, for vectors of $\RR^N$ or $\RR^{2N}$, whereas the notation $\norm{\cdot}$ refers to an operator norm (on matrices, or bounded linear operators). $\normObs{\cdot}$ is the norm on $\Obs$ associated to the inner product $\dotObs{\cdot}{\cdot}$. $\normLi{\cdot}$ denotes the $\Linf$ norm for functions defined on $\Pos$.


\section{Asymptotic Analysis of Support Recovery}
\label{sec-exact-support-recov}

This section exposes our main contribution (Theorem~\ref{thm-main}). The hypotheses of this result require an injectivity condition and a non-degeneracy condition, that are explained in Sections~\ref{sec-inject} and~\ref{sec-vanish-precertif}.

\subsection{Injectivity Hypotheses}
\label{sec-inject}

We introduce here an injectivity hypothesis which ensures the invertibility of $\Phitz^* \Phitz$ and $\Gatz^* \Gatz$ for $t>0$ small enough.

In the case of the ideal low-pass filter (defined in~\eqref{eq-idealLPF}), 
 $\Gax$ has full column rank provided that $\bar{\pos}=(\pos_1,\ldots, \pos_N) \in \Pos^N$ has pairwise distinct coordinates (see~\cite{duval-exact2013}, Section 3.6, Proposition 6). That property is not true for a general operator $\Phi$.
However, in this paper we focus on sums of Dirac masses that are clustered around the point $0\in\Pos$, \ie $\bar{\pos}=t\posz$ for $t>0$ and $\posz\in \RR^N$ with pairwise distinct components. The following assumption, which is crucial to our analysis, shall ensure that $\Gatz$ has full rank at least for small $t$.  

\begin{defn}\label{main-hypothesis}
  	Let $\varphi : \Pos\to\Obs$. For all $k\in\NN$, we say that the hypothesis $\injk$ holds if and only if 
	\eql{\tag{$\injk$}
		\phi\in\kernel{k}
		\text{ and }
		(\phiD{0}, \ldots,\phiD{k}) \text{ are linearly independent in }
		\Obs.
	}
\end{defn}
See Definition~\ref{def-kernel} for the definition of the space $\kernel{k}$ and Equation~\eqref{eq-defn-phider} for the definition of $\phiD{k}$.


If $\injk$ holds, then $\Fk^* \Fk$ is a symmetric positive \emph{definite} matrix, where $\Fk$ is defined in~\eqref{eq-defn-Fk}.

To exemplify the meaning of this injectivity hypothesis, Proposition~\ref{hyp-fourier} below considers the case $\Pos=\TT$ with $\Phi$ a convolution operator. 

\begin{prop}\label{hyp-fourier}
  Let $\tilde{\phi} \in \Cder{k}(\TT,\RR)$ (where $\phi(\pos)=\tilde{\phi}(\cdot-\pos)$), then $\injk$ holds if and only if $\phiD{0}$ has at least $k+1$ non-zeros Fourier coefficients.
  In particular if $\Phi$ is the ideal low-pass filter with cutoff frequency $f_c\in\NN^*$, $\injk$ holds if and only if $k\leq 2 f_c$.
\end{prop}

The proof of this proposition is given in Section~\ref{sec-proof-hyp-fourier}.



As we shall see in Section~\ref{sec-prelim-dl}, the conditions $\injn$ and $\injdn$ imply respectively the invertibility of $\Phitz^* \Phitz$ and $\Gatz^* \Gatz$, provided that $t$ is small enough. According to Proposition~\ref{hyp-fourier}, in the special case of an ideal low-pass filter, these conditions holds if and only if $f_c$ is large enough with respect to the number $N$ of spikes.

\subsection{Vanishing Derivatives Precertificate}
\label{sec-vanish-precertif}

Following~\cite{duval-exact2013} (Section 4.1, Definition 6), we introduce below the so called ``vanishing derivatives pre-certificate'' $\etaV$, which is a function defined on $\Pos$ that interpolates the spikes positions and signs (here $+1$). Note that $\etaV$ can be computed in closed form by solving the linear system~\eqref{eq-vanish-closed-form}.


\begin{prop}[Vanishing derivatives precertificate, \cite{duval-exact2013}]\label{propdef-vanishing}
 	If $\GatzO$ has full column rank, there is a unique solution to the problem
  \begin{align*}
    \inf\enscond{\normObs{p}}{\forall i=1,\ldots,N, \; (\Phi^*p)(t z_{0,i})=1, (\Phi^*p)'(t z_{0,i})=0}.
  \end{align*}
  Its solution $\pV$ is given by 
  \begin{align}\label{eq-vanish-closed-form}
   	\pV=\GatzO^{+,*} \begin{pmatrix} \bun_N \\ 0\end{pmatrix},
  \end{align}
and we define the vanishing derivatives precertificate as $\etaV\eqdef\Phi^*\pV$.
\end{prop}

As shown in~\cite{duval-exact2013} (see Section 4.2), $\etaV$ governs the support recovery properties of the BLASSO $\blasso$. More precisely, if $\etaV$ is \textit{non-degenerate}, \textit{i.e.} 
\eql{\label{eq-etaV-nondegen}
\left\{\begin{split}
	\foralls \pos\in \Pos \setminus \{t z_{0,1},\ldots, t z_{0,N}\}, \quad \abs{\etaV(\pos)}&<1,\\
	 \foralls i\in\{1,\ldots,N\}, \quad
	\etaV''( t \posz_{0,i} )&\neq 0,
\end{split}\right.
}
then there exists a low noise regime in which the BLASSO recovers exactly the correct number $N$ of spikes, and the error on the locations and amplitudes is proportional to the noise level.

The constants involved in the main result of~\cite{duval-exact2013} (Theorem 2) depend on the value of $t>0$. The goal of this paper is to precisely determine this dependency, and to show that this support recovery result extends to the setting where $t \rightarrow 0$, provided that $(\la,w)$ obey some precise scaling with respect to $t$.


Since our focus is actually on the support recovery properties of the BLASSO when $t\to 0$, it is natural to look at the limit of $\pV$ as $t\to 0$. This leads us to the $(2N-1)$-vanishing derivatives precertificate defined below.

\begin{prop}[$(2N-1)$-vanishing derivatives precertificate]\label{def-nvanishing}
	If $\injdn$ holds, there is a unique solution to the problem
  \begin{align*}
    \inf\enscond{\normObs{p}}{(\Phi^*p)(0)=1, (\Phi^*p)'(0)=0, \ldots , (\Phi^*p)^{(2N-1)}(0)=0}.
  \end{align*}
  We denote by $\pW$ its solution, given by 
  \begin{align}\label{def-pW}
    \pW=\Fdn^{+,*}\dirac{2N}   
  \end{align}
and we define the $(2N-1)$-vanishing derivatives precertificate as $\etaW\eqdef\Phi^*\pW$.
\end{prop}

The following Proposition, which is a direct consequence of Lemma~\ref{lem-dl-ftz} in the next section, shows that indeed $\etaV$ converges toward $\etaW$. 

\begin{prop}
  If $\injdn$ holds and $\phi\in\kernel{K}$ (for $K\geq 2N-1$), then for $t>0$ small enough $\GatzO$ has full column rank. Moreover
  \begin{align*}
    \lim_{t\to 0^+} \pV &= \pW \mbox{ strongly in $\Obs$},\\
    \lim_{t\to 0^+} \etaV^{(k)} &= \etaW^{(k)} \mbox{ in the sense of the uniform convergence on $\Pos$,}
  \end{align*}
  for all $0\leq k\leq K$.
  \label{prop-cvetav}
\end{prop}

Figure~\ref{certif-convergence} shows graphically this convergence of $\etaV$ toward $\etaW$ in the case of the deconvolution problem over the $1$-D torus with the Dirichlet kernel.
Figure~\ref{etaW-dirichlet-fig} and~\ref{etaW-gaussian} show $\etaW$ for several values of $N$. Notice how it becomes flatter at $0$ as $N$ increases. This implies that $\etaV$ for small $t$ gets closer to degeneracy as $N$ increases. This is reflected in our main contribution (Theorem~\ref{thm-main}) where the signal-to-noise ratio is required to scale with $t^{2N-1}$. 

\begin{figure*}
	\centering    
	\begin{tabular}{@{\hspace{1mm}}c@{\hspace{2mm}}c@{\hspace{2mm}}c@{\hspace{1mm}}}
      \includegraphics[width=0.31\linewidth]{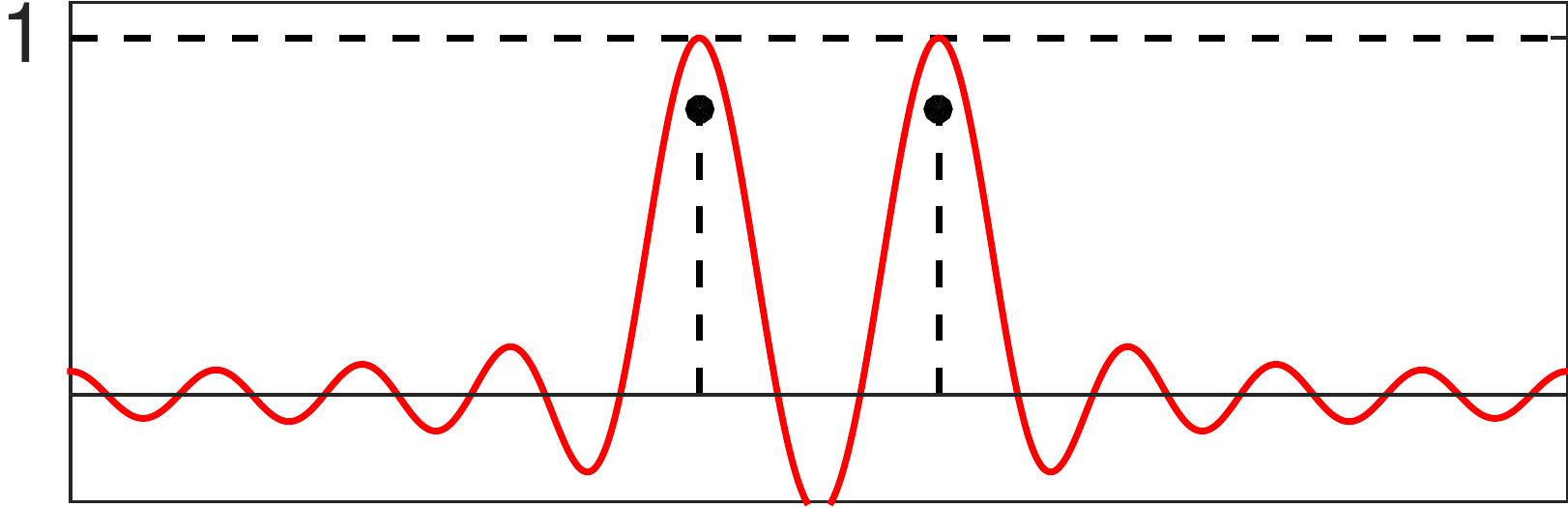} &
      \includegraphics[width=0.31\linewidth]{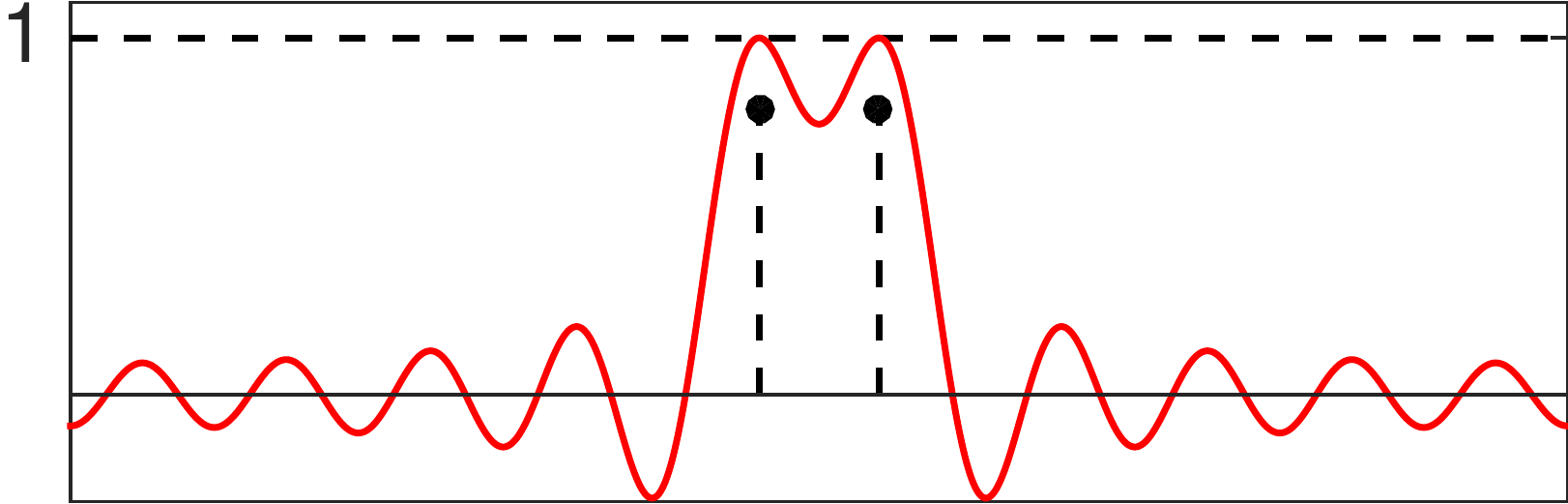} &
      \includegraphics[width=0.31\linewidth]{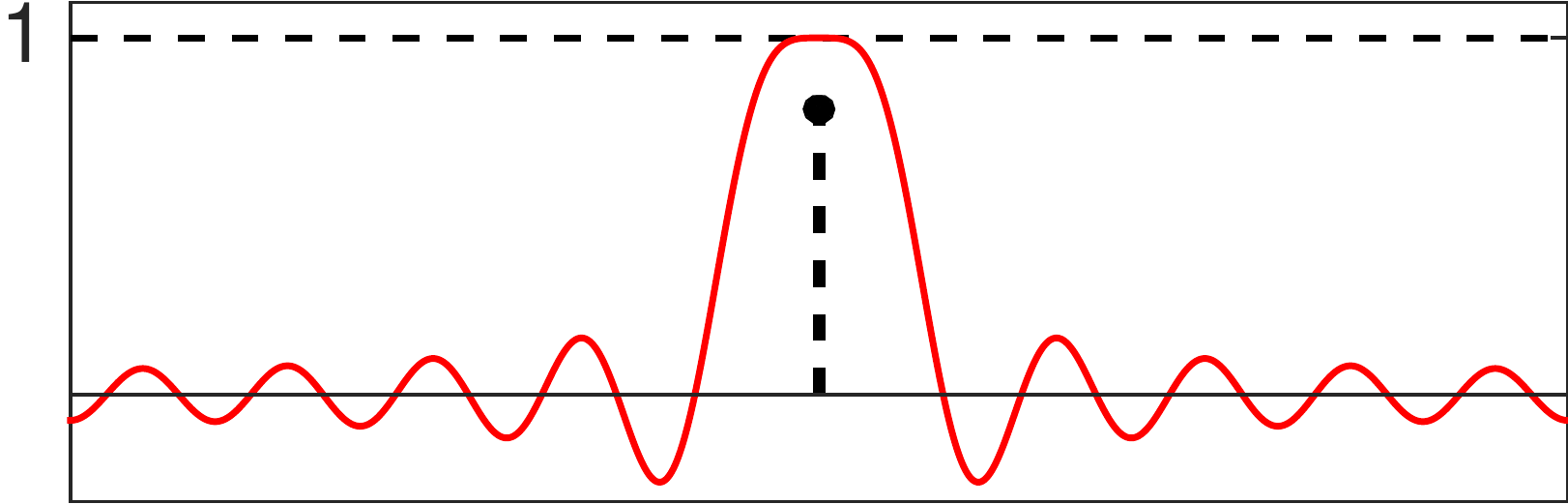} \\
      $t=0.4$ & $t=0.2$ & $t=0.01$ 
   	\end{tabular}
\caption{\label{certif-convergence} 
	\textit{Top row:} $\etaV$ for several values of $t$, showing the convergence toward $\etaW$. 
  The operator $\Phi$ is an ideal low-pass filter with a cutoff frequency $f_c=10$ (see~\eqref{eq-idealLPF}).
}
\end{figure*}

\begin{figure*}
	\centering    
	\begin{tabular}{@{\hspace{1mm}}c@{\hspace{2mm}}c@{\hspace{2mm}}c@{\hspace{1mm}}}
      \includegraphics[width=0.31\linewidth]{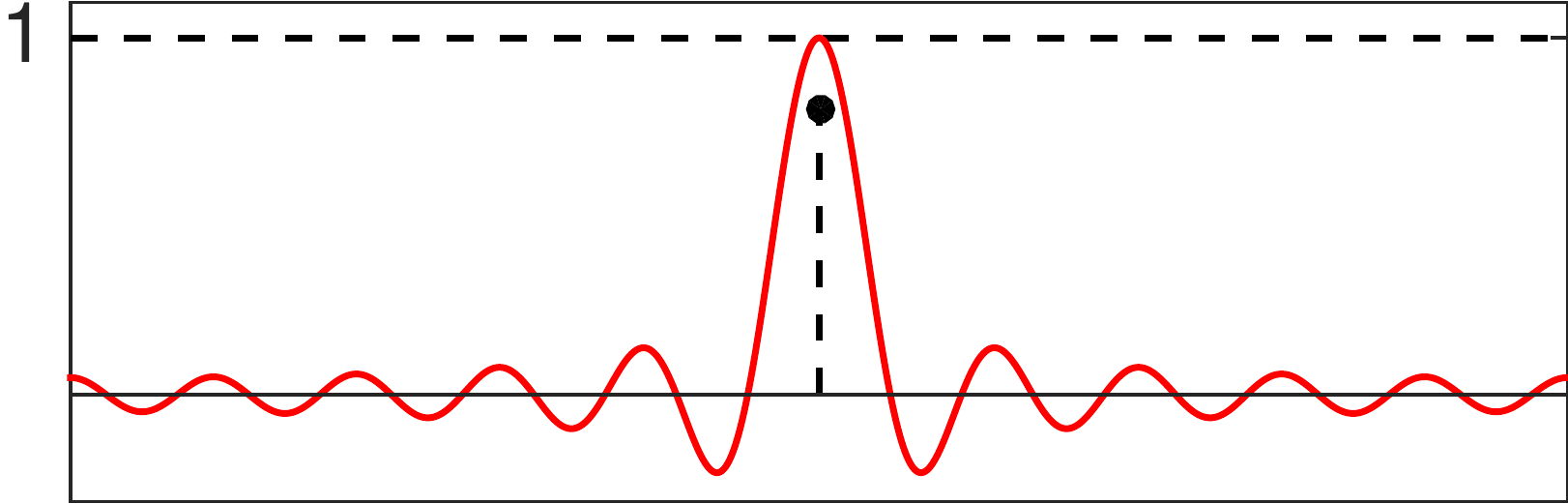}&
      \includegraphics[width=0.31\linewidth]{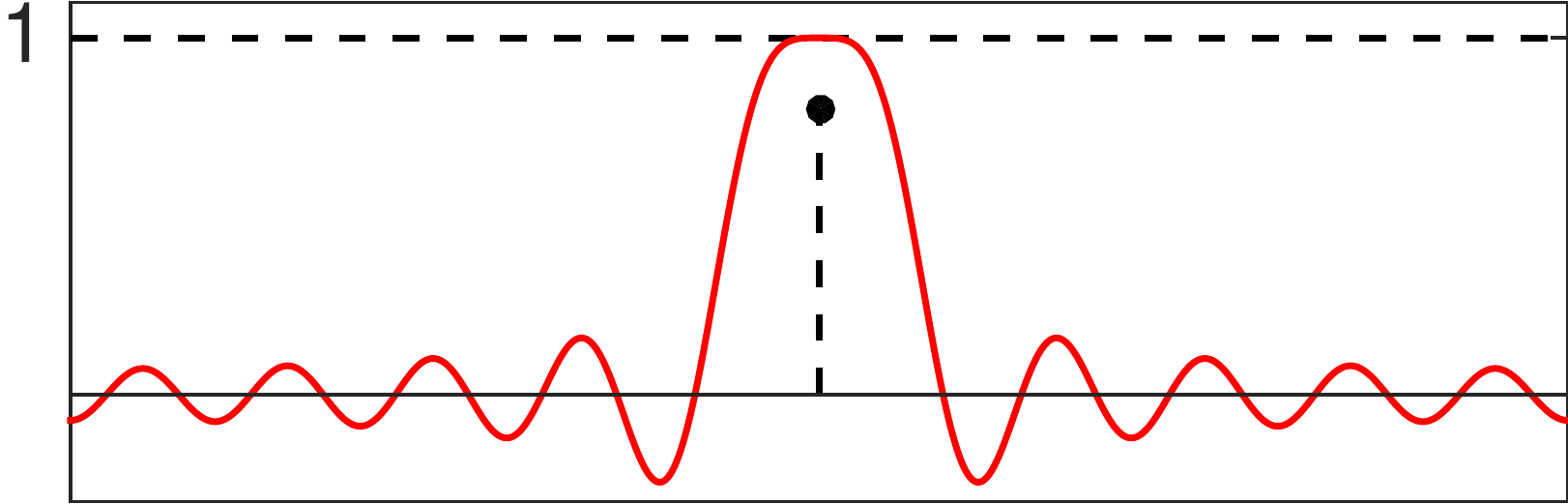}&
      \includegraphics[width=0.31\linewidth]{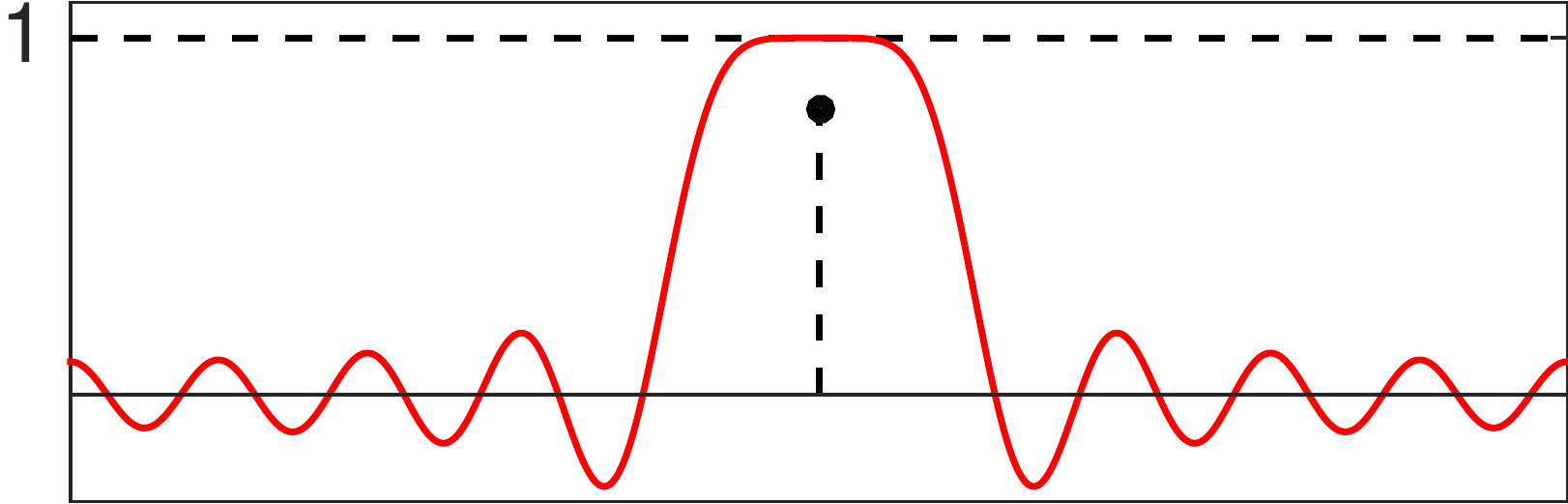}\\
      $N=1$ ($\etaV=\etaW$) & $N=2$ & $N=3$
   	\end{tabular}
\caption{\label{etaW-dirichlet-fig} 
	$\etaW$ for several values of $N$. The operator $\Phi$ is an ideal low-pass filter with a cutoff frequency $f_c=10$ (see~\eqref{eq-idealLPF})
}
\end{figure*}

The behavior of $\etaV$ is therefore governed by specific properties of $\etaW$ for small values of $t>0$. 
In particular, as stated by Theorem~\ref{if-etaw-nondegen} below, the non-degeneracy of $\etaW$ (as defined in Definition~\ref{def-etaw-nondegen} below) implies the non-degeneracy of $\etaV$ (as defined in~\eqref{eq-etaV-nondegen}). 

\begin{defn}\label{def-etaw-nondegen}
 	Assume that $\injdn$ holds and $\phi\in\kernel{2N}$. We say that $\etaW$ is $(2N-1)$-non-degenerate if $\etaW^{(2N)}(0)\neq 0$ and for all $\pos\in \Pos\setminus\{0\}$, $|\etaW(\pos)| <1$.
\end{defn}

\begin{thm}\label{if-etaw-nondegen}
Suppose that $\etaW$ is $(2N-1)$-non-degenerate (Definition~\ref{def-etaw-nondegen}). Let $\RW>0$. Then, there exist $\constW>0$, $\tW>0$ such that for all $t\in (0,\tW)$, all $\posz\in\RR^N$ with pairwise distinct coordinates and $\normLiVec{z}\leq\RW$, and all $\eta \in \precert{2N}$ satisfying for $1\leq i \leq N$, $\eta(\post_i)=1$ and $\eta'(\post_i)=0$,
\begin{align*}
 &\left(\forall \ell\in \{0,\ldots 2N\},\quad \normLi{\eta^{(\ell)}-\etaW^{(\ell)}}\leq \constW\right) \\ 
\Longrightarrow 	&\left(\forall \pos\in \Pos\setminus\bigcup_i \{\post_i\},
	\quad 
	|\eta(\pos)|<1 \qandq \forall 1\leq i \leq N, \ \eta''(\post_i)<0\right).
\end{align*}
\end{thm}

The proof of this theorem can be found in Section~\ref{proof-if-etaw-nondegen}. 

\begin{rem}
  An important consequence of Theorem~\ref{if-etaw-nondegen} is that if $\etaW$ is $(2N-1)$-non-degenerate, then, by Proposition~\ref{prop-cvetav}, $\etaV$ is also non-degenerate for $t>0$ small enough. The function $\etaV$ is then the minimal norm certificate for the measure $\meastO$, and the \textit{Non-Degenerate Source Condition} (see~\cite{duval-exact2013}) holds. As a result, for fixed small $t>0$, the BLASSO admits a unique solution in a certain low noise regime corresponding to a large enough signal to noise ratio, with exactly the same number of spikes as the original measure $\meastO$. For more details on that matter, see~\cite{duval-exact2013}.
\end{rem}


A natural question is whether $\etaW$ is indeed $(2N-1)$-non-degenerate. Proposition~\ref{etaw-locnondegen} below (proved in Section~\ref{sec-proof-etaw-locnondegen}) gives a partial answer in the case of the deconvolution over $\Pos$.

\begin{prop}\label{etaw-locnondegen}
  Assume that $\Phi$ is a convolution operator (\ie for all $\pos\in\Pos$, $\phi(\pos)=\tilde{\phi}(\cdot-\pos)$ and $\Obs=L^2(\Pos)$) and $\injdnu$ holds. Suppose also, only in the case $\Pos=\RR$, that for all $0\leq i\leq 2N-1$, $\tilde{\phi}^{(i)}(\pos)\to 0$ when $|\pos|\to+\infty$. Then $\etaW^{(2N)}(0)<0$.
\end{prop}

\begin{rem}
  Thanks to Proposition~\ref{etaw-locnondegen} and the first part of the proof of Theorem~\ref{if-etaw-nondegen} (which is given in Section~\ref{proof-if-etaw-nondegen}), note that the following is true: there exist $\constW>0$, $\tW>0$ such that for all $t\in (0,\tW)$, $\posz\in\RR^N$ with pairwise distinct coordinates and $\normLiVec{\posz}\leq\RW$, there exists $r^+>0$ with $r^+>\underset{1\leq i\leq N}{\max} \tW\posz_i$ and $r^-<0$ with $r^-<\underset{1\leq i\leq N}{\min} \tW\posz_i$ such that for all $\eta \in \precert{2N}$ satisfying for all $1\leq i\leq N$, $\eta(\post_i)=1$ and $\eta'(\post_i)=0$:
\begin{align*}
  &\left(\forall \ell\in \{0,\ldots, 2N\}, \ \normLi{\eta^{(\ell)}-\etaW^{(\ell)}}\leq \constW \right)\\
  &\Longrightarrow \left(\forall \pos\in (r^-,r^+)\setminus\bigcup_i \{\post_i\}, \ |\eta(\pos)|<1 \qandq \forall i\in\{1, \ldots, N\}, \ \eta''(\post_i)<0   \right).
\end{align*}
\end{rem}

Whether the other condition in the definition of the $(2N-1)$-non-degeneracy of $\etaW$ holds (i.e. whether $|\etaW|<1$ on $\Pos\setminus\{0\}$) should be checked on a case-by-case basis. Since $\etaW$ depends only on the kernel $\phi$ and can be computed by simply inverting a linear system, it is easy to check numerically if $\etaW$ is $(2N-1)$-non-degenerate.
Proposition~\ref{etaW-gaussian-nondegen} shows that in the special case of $\Pos=\RR$ and the Gaussian kernel, $\etaW$ can be computed in closed form and is indeed $(2N-1)$-non-degenerate.


\begin{prop}\label{etaW-gaussian-nondegen}
Assume that $\Pos=\RR$, $\Obs=L^2(\RR)$ and $\Phi$ is a convolution operator, \ie for all $\pos\in\RR$, $\phi(\pos)=\tilde{\phi}(\cdot-\pos)$, where $\tilde{\phi}:\pos\in\RR\mapsto e^{-\pos^2/2}$ is a Gaussian. Then the associated $(2N-1)$-vanishing derivatives precertificate is 
\begin{align}\label{etaW-closedform-gaussian}
	\foralls \pos\in\RR, \quad
	\etaW(\pos)=e^{-\frac{\pos^2}{4}}\sum_{k=1}^N \frac{\pos^{2k}}{2^{2k}k!}.
\end{align}
In particular, $\etaW$ is $(2N-1)$-non-degenerate.
\end{prop}
The proof of this result can be found in the Appendix~\ref{sec-proof-gaussian-nondegen}. If we denote by $\eta_{W,\sigma}$, the $(2N-1)$ vanishing derivatives precertificate associated to the filter $\phi_\sigma:\pos\in\RR\mapsto e^{-\frac{\pos^2}{2\sigma^2}}$, then $\eta_{W,\sigma}=\eta_{W,1}(\frac{\cdot}{\sigma})$. That is why we only consider the case of $\sigma=1$ in Proposition~\ref{etaW-gaussian-nondegen}. Figure~\ref{etaW-gaussian} shows $\etaW$ for the Gaussian filter with an increasing number $N$ of spikes.

For the deconvolution over the $1$-D torus, we observed numerically (as illustrated in Figure~\ref{etaW-dirichlet-fig}) that $\etaW$ is $(2N-1)$-non degenerate for the ideal low pass filter for any value $N$ such that $N\leq f_c$ (Figure~\ref{etaW-fc-incr} represents the complementary case where $N$ is fixed but $f_c$ increases). However, for some filters, the associated $\etaW$ might be degenerate. This is illustrated in Figure~\ref{etaW-comp-incr} where $\etaW$ is illustrated for several filters with increasing complexity \ie we consider low pass filters with a fixed cutoff frequency, with increasing extreme Fourier coefficients (starting with a slowly varying filter). Remark that the last two $\etaW$ (in red) are degenerate, as they correspond to the filters with the higher complexity (the Fourier coefficients increase the most with the frequency).

%
%

\begin{figure}[htbp]
\centering
 \begin{tabular}{@{\hspace{-1mm}}c@{\hspace{-3mm}}c@{\hspace{-3mm}}c@{\hspace{-3mm}}c@{\hspace{-1mm}}}
\includegraphics[width=0.26\linewidth]{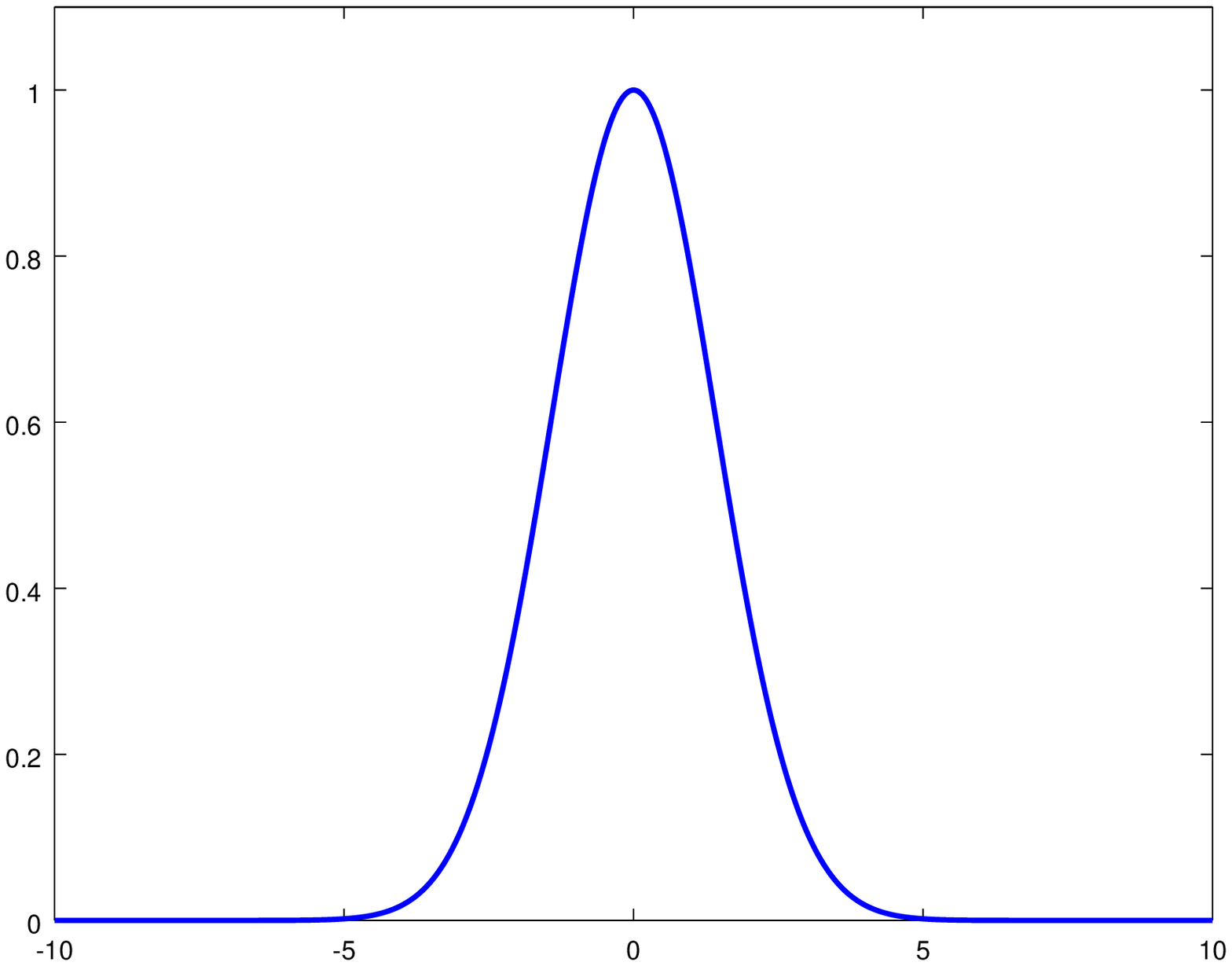} &
 \includegraphics[width=0.26\linewidth]{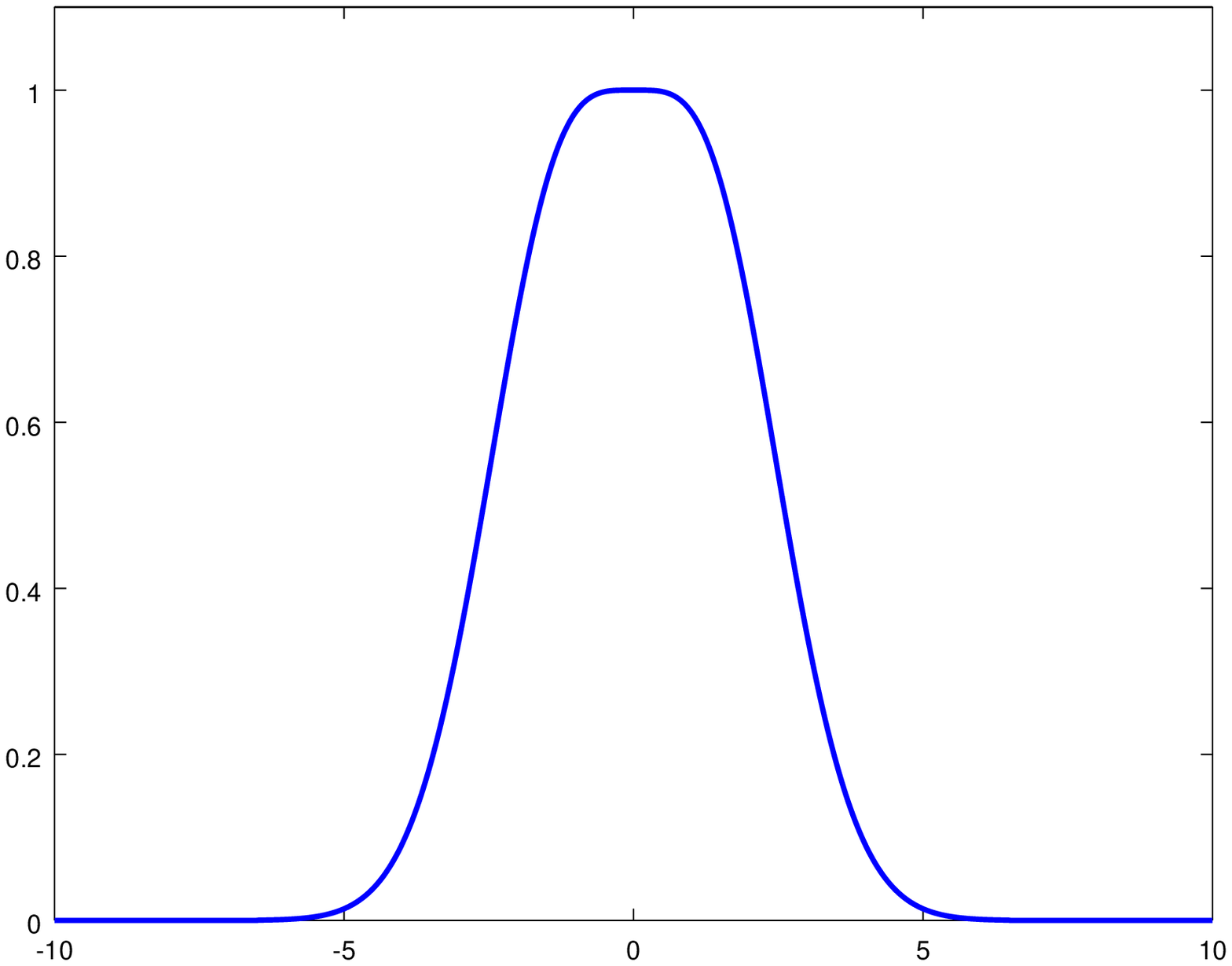} &
 \includegraphics[width=0.26\linewidth]{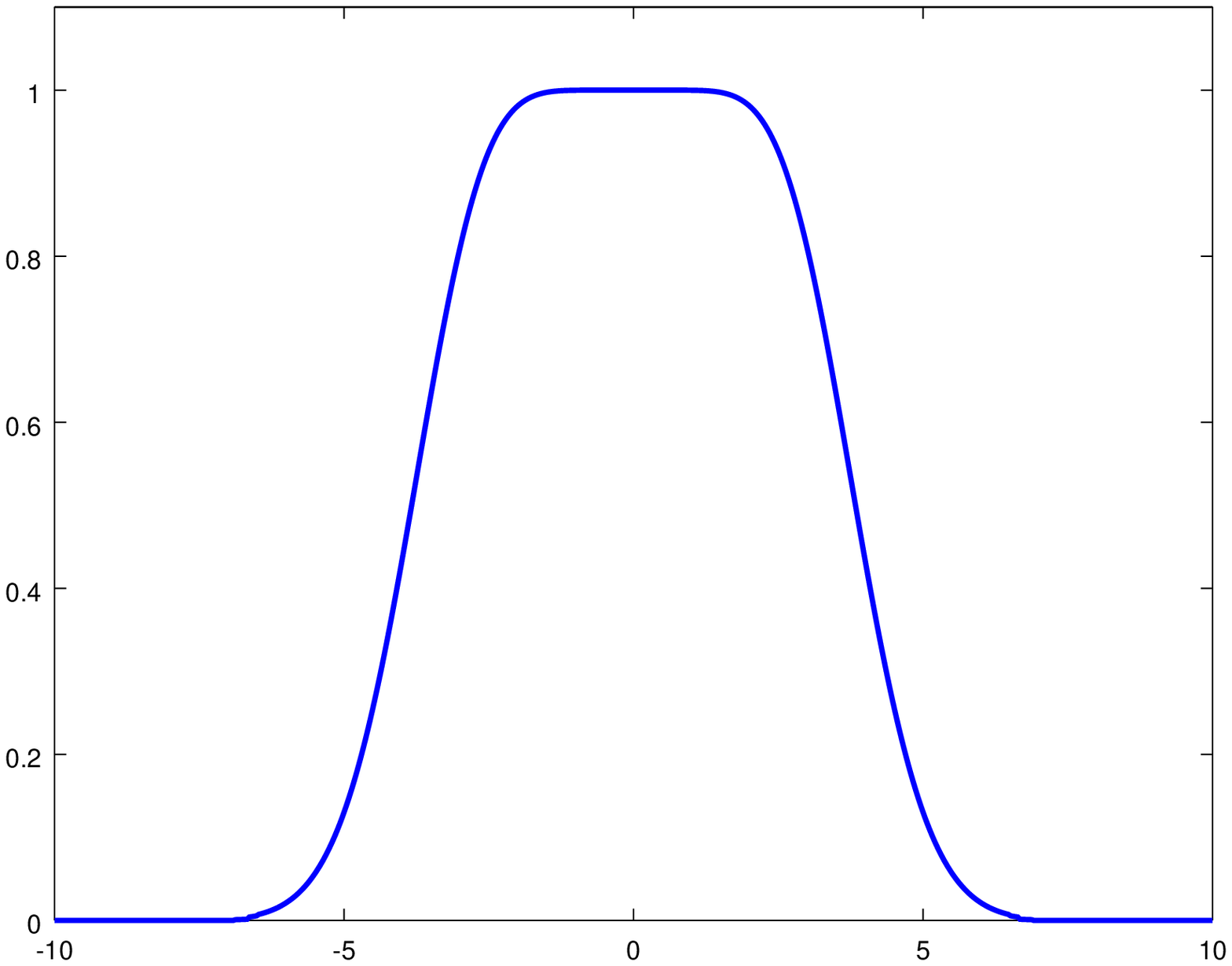} &
 \includegraphics[width=0.26\linewidth]{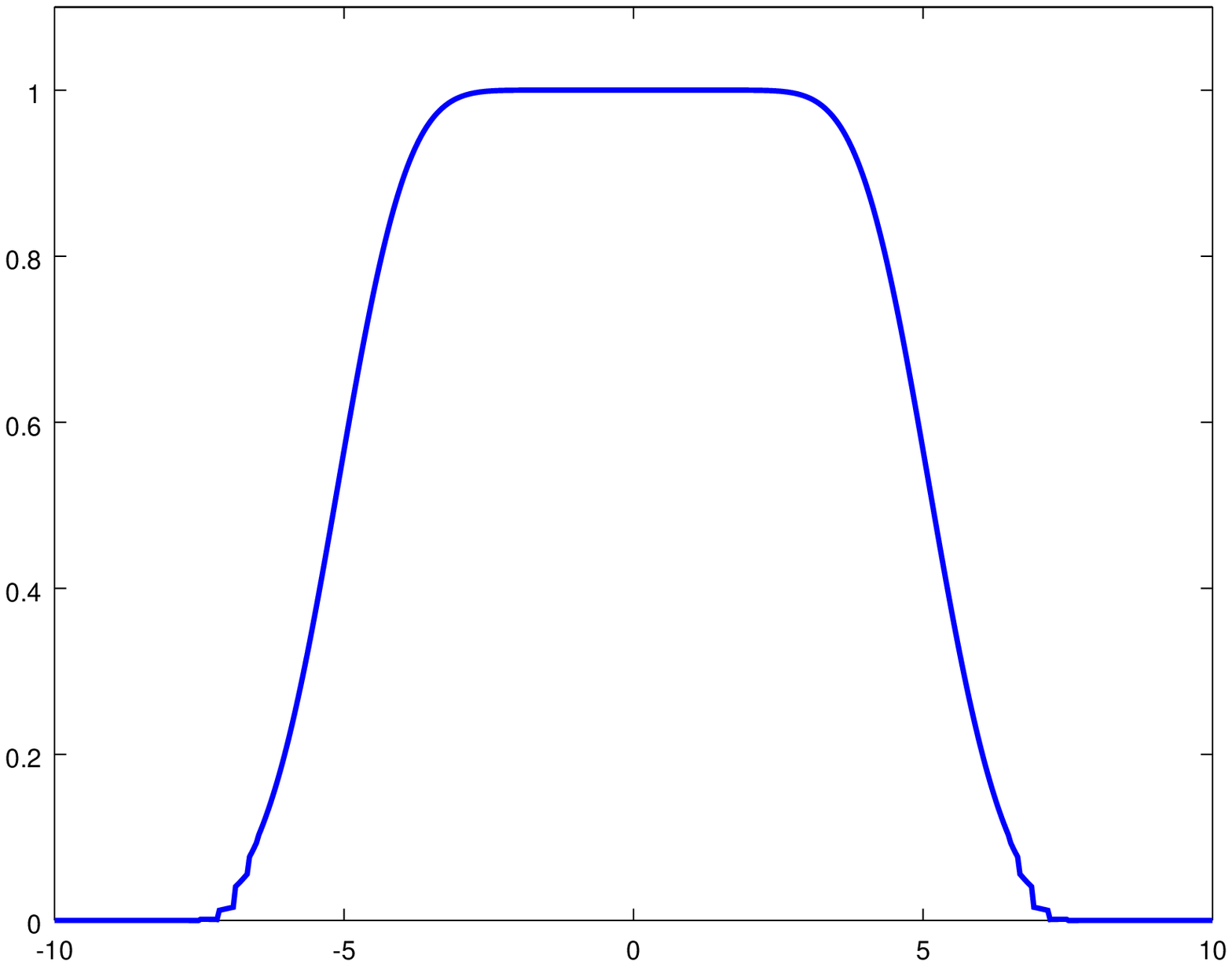} 
  \vspace{-0.4cm}\\
 $N=1$ & $N=2$ & $N=4$ & $N=7$
 \end{tabular}
   \vspace{-0.4cm}
 \caption{\label{etaW-gaussian}$\etaW$ for the Gaussian filter ($x\in\RR$, $\phi(x)=e^{-\frac{x^2}{2\sigma^2}}$) for several numbers of spikes and $\sigma=1$. $\etaW$ is $(2N-1)$-non-degenerate. It gets flatter at $0$ when the number of spikes increases.}
\end{figure}

\begin{figure}[htbp]
 \centering
 \begin{tabular}{@{\hspace{-1mm}}c@{\hspace{-3mm}}c@{\hspace{-3mm}}c@{\hspace{-3mm}}c@{\hspace{-1mm}}}
 \includegraphics[width=0.26\linewidth]{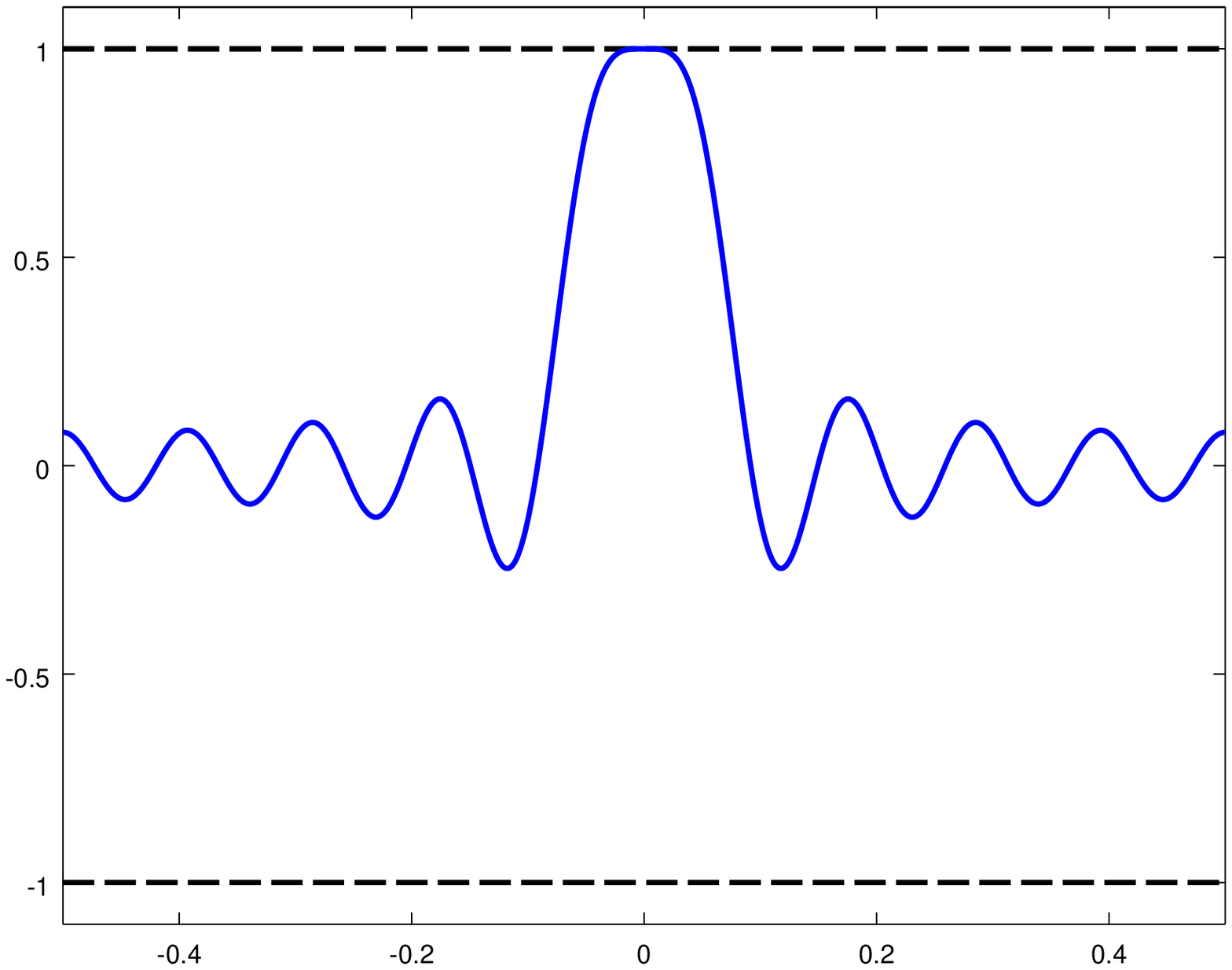} &
 \includegraphics[width=0.26\linewidth]{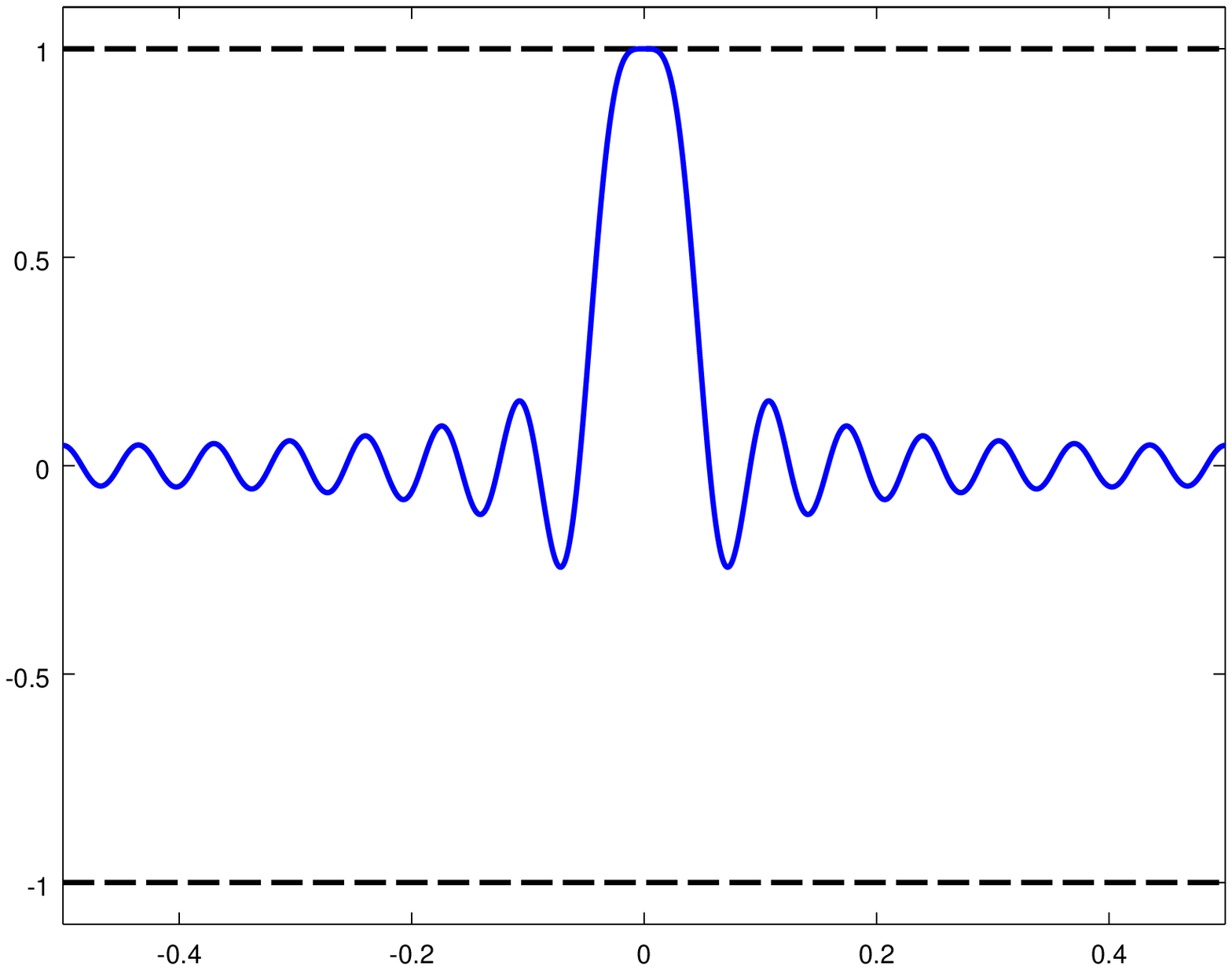} &
 \includegraphics[width=0.26\linewidth]{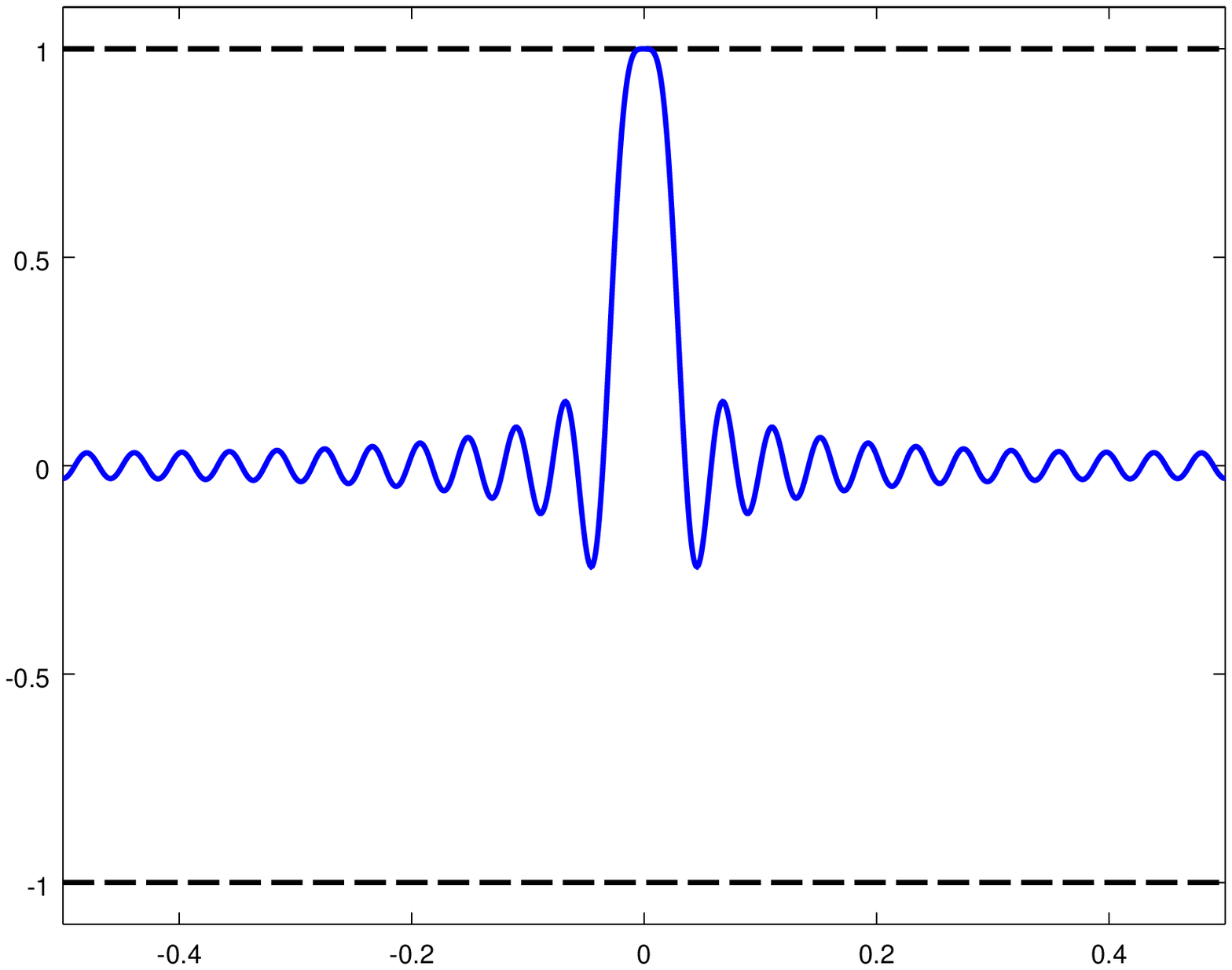} &
 \includegraphics[width=0.26\linewidth]{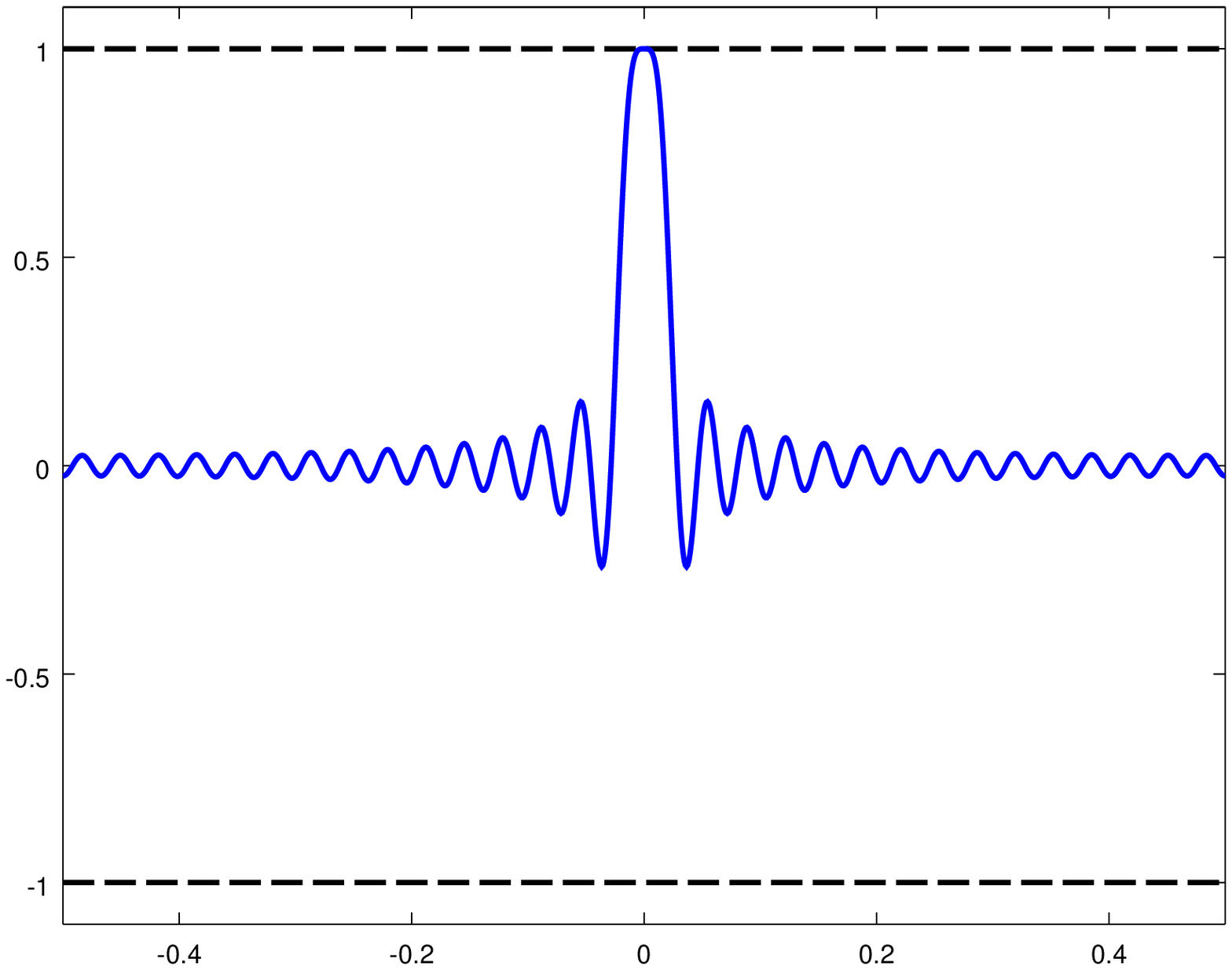}
 \vspace{-0.4cm}\\
 $f_c=9$ & $f_c=15$ & $f_c=24$ & $f_c=30$  \vspace{-0.2cm}
 \end{tabular}
 \caption{\label{etaW-fc-incr}$\etaW$ for the ideal low pass filter as $f_c$ increases, for $N=2$.}
\end{figure}

\begin{figure}[htbp]
 \centering
 \begin{tabular}{@{\hspace{-1mm}}c@{\hspace{-3mm}}c@{\hspace{-3mm}}c@{\hspace{-3mm}}c@{\hspace{-3mm}}c@{\hspace{-1mm}}}
 \includegraphics[width=0.22\linewidth]{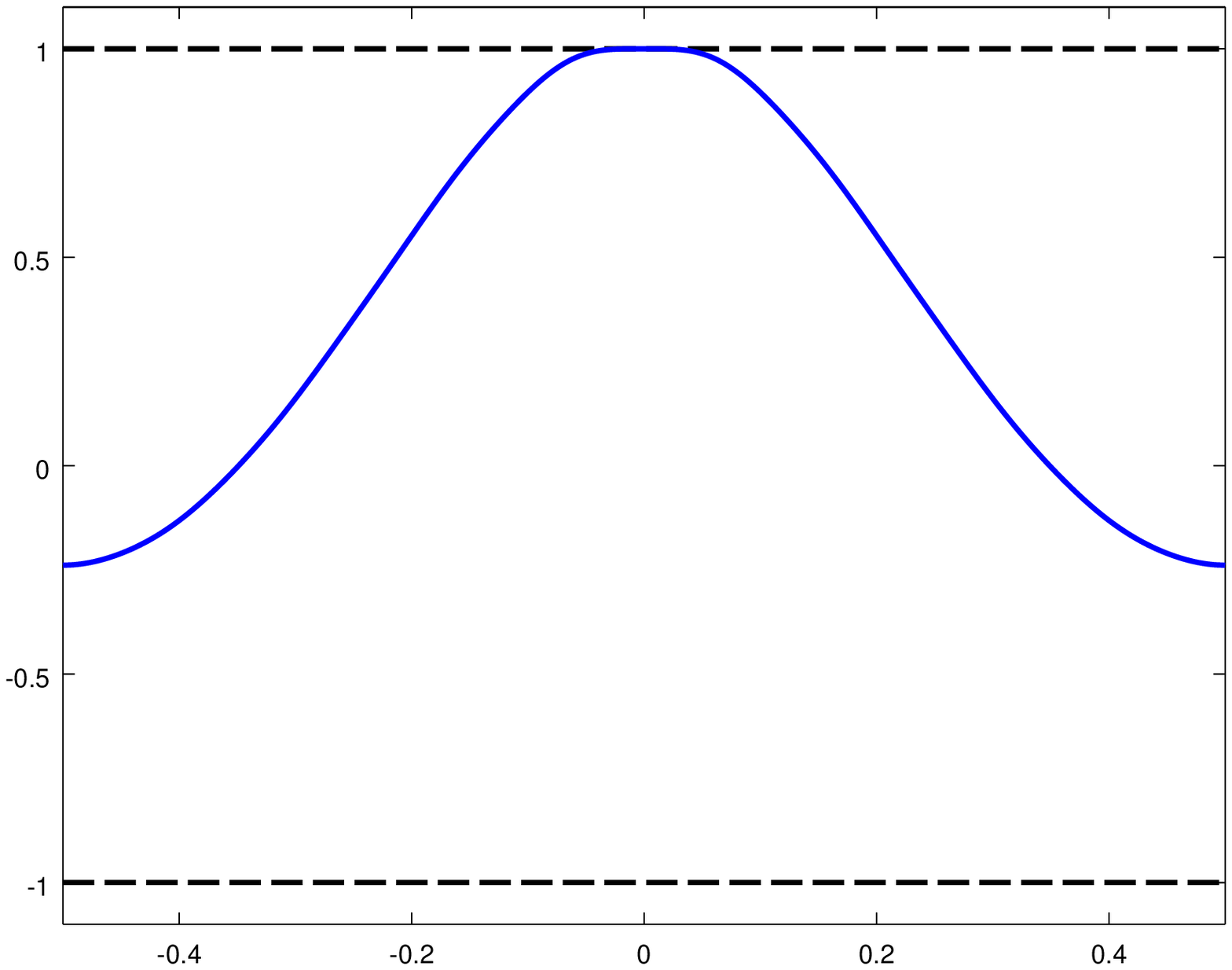} &
 \includegraphics[width=0.22\linewidth]{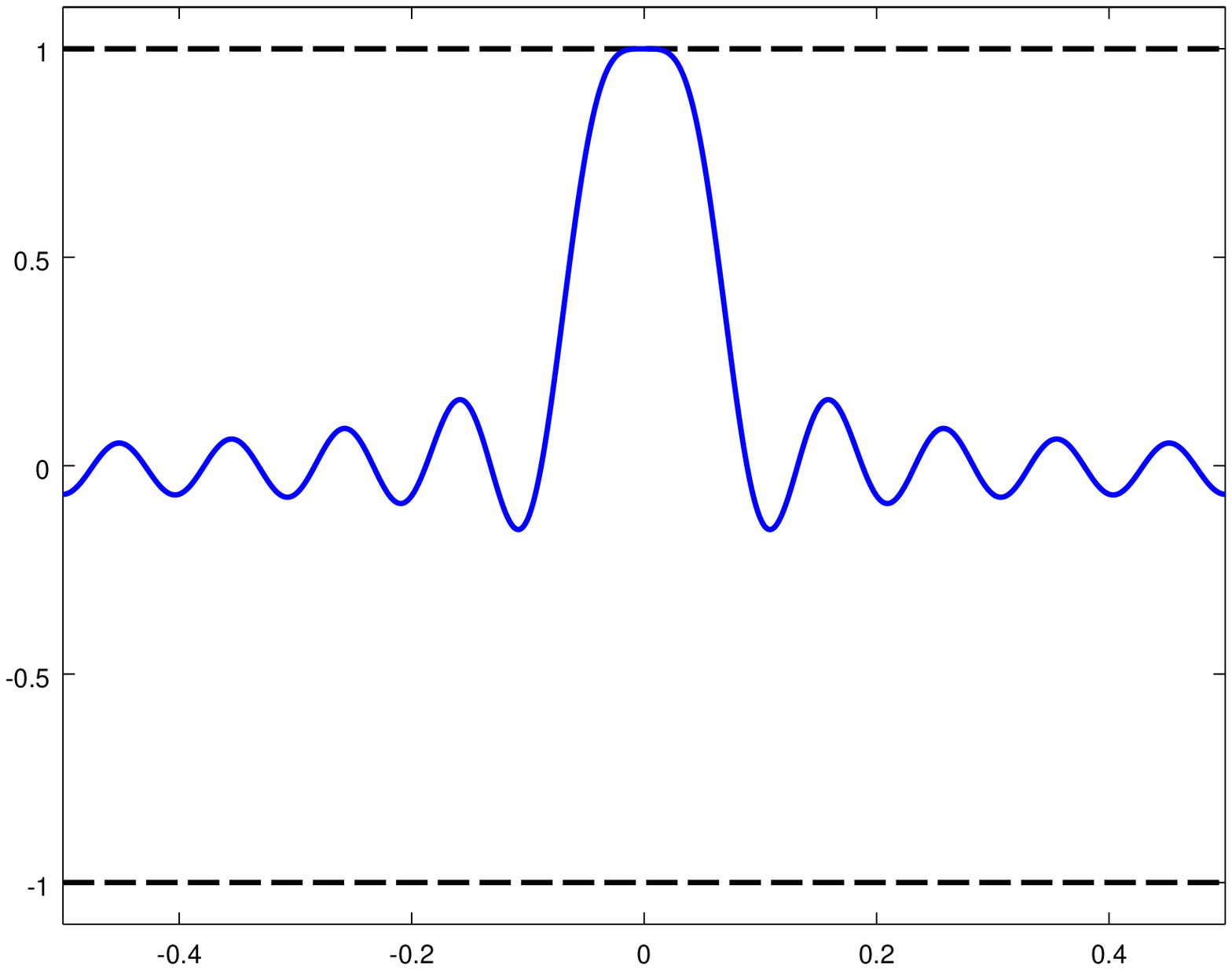} &
 \includegraphics[width=0.22\linewidth]{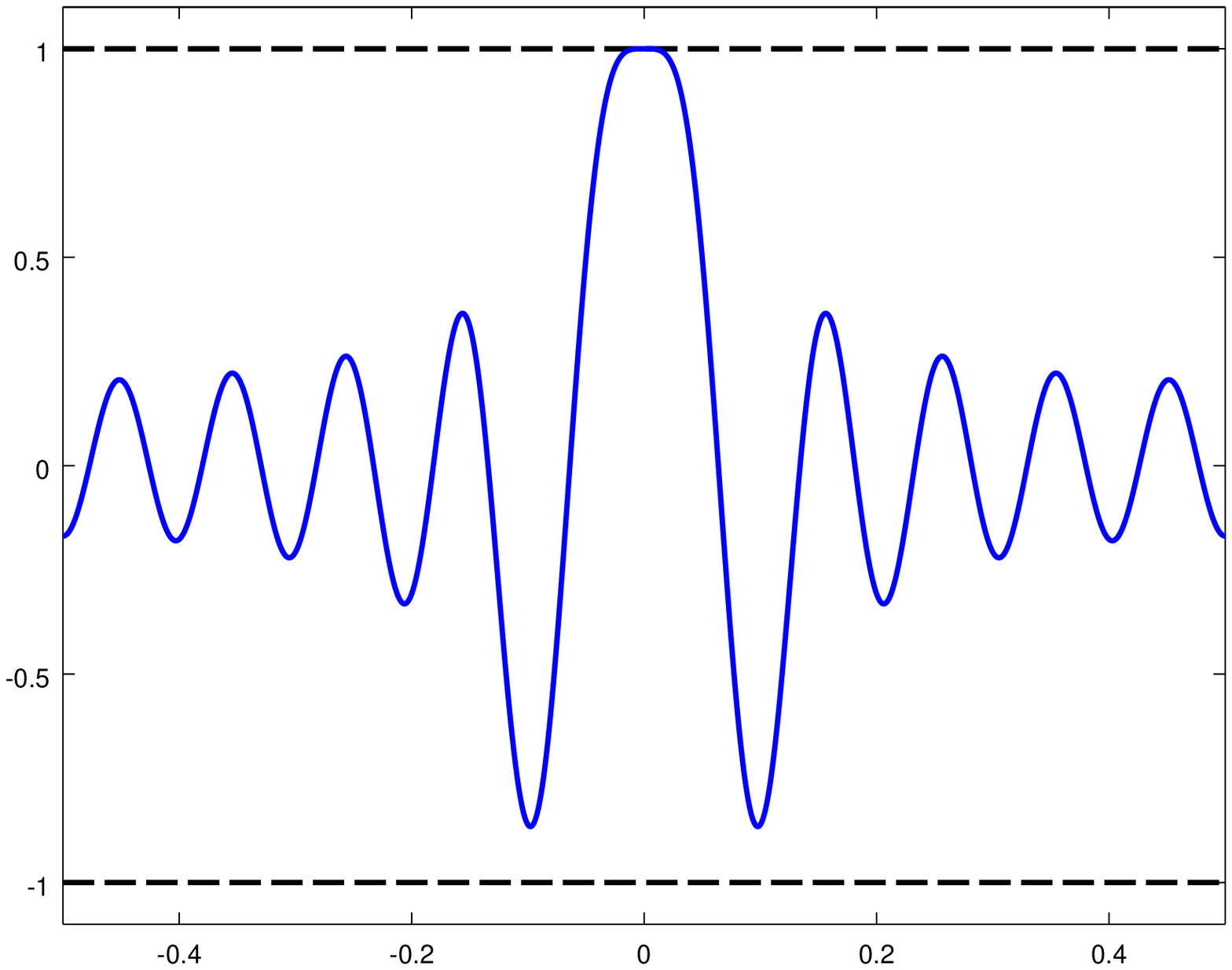} &
 \includegraphics[width=0.22\linewidth]{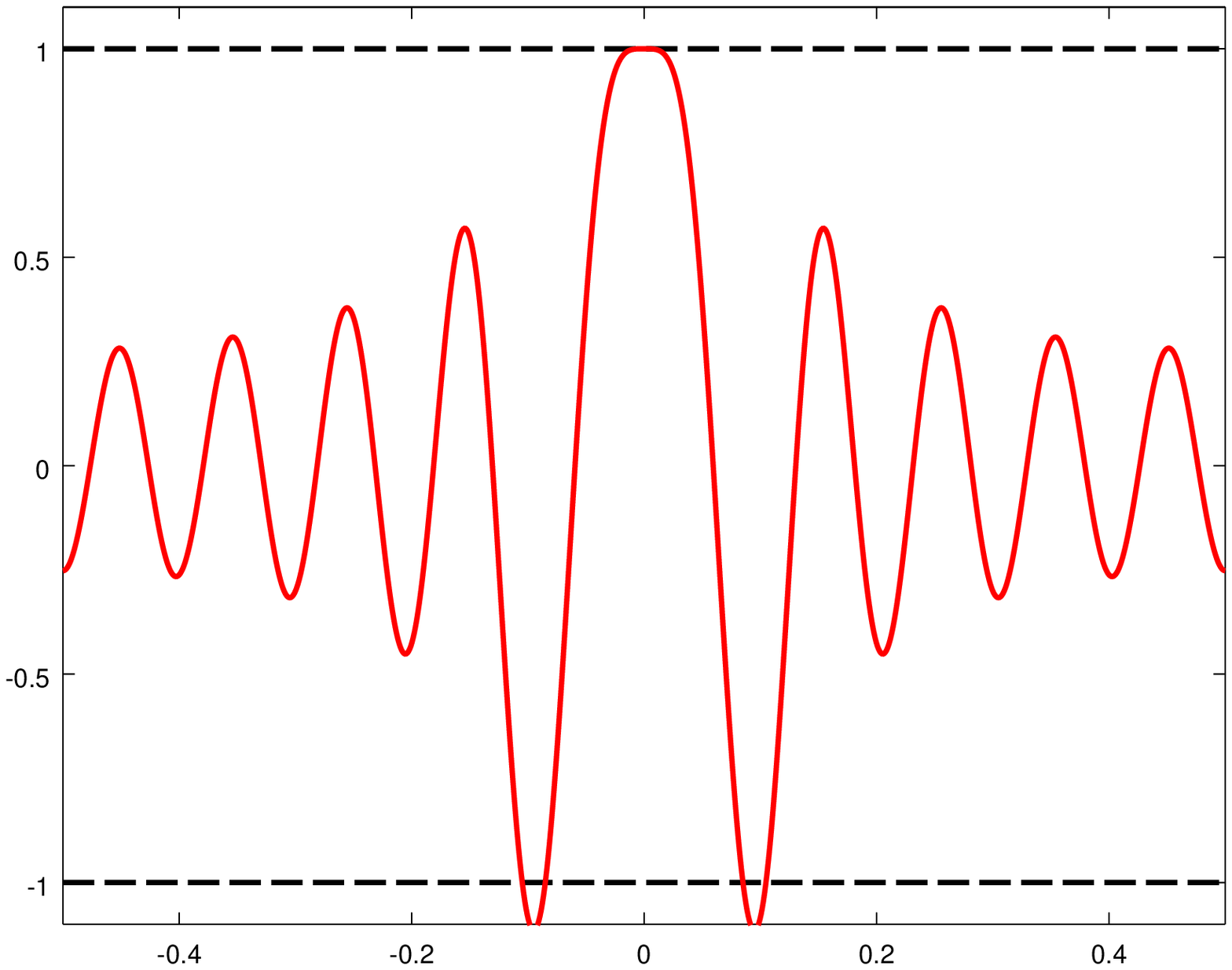} &
 \includegraphics[width=0.22\linewidth]{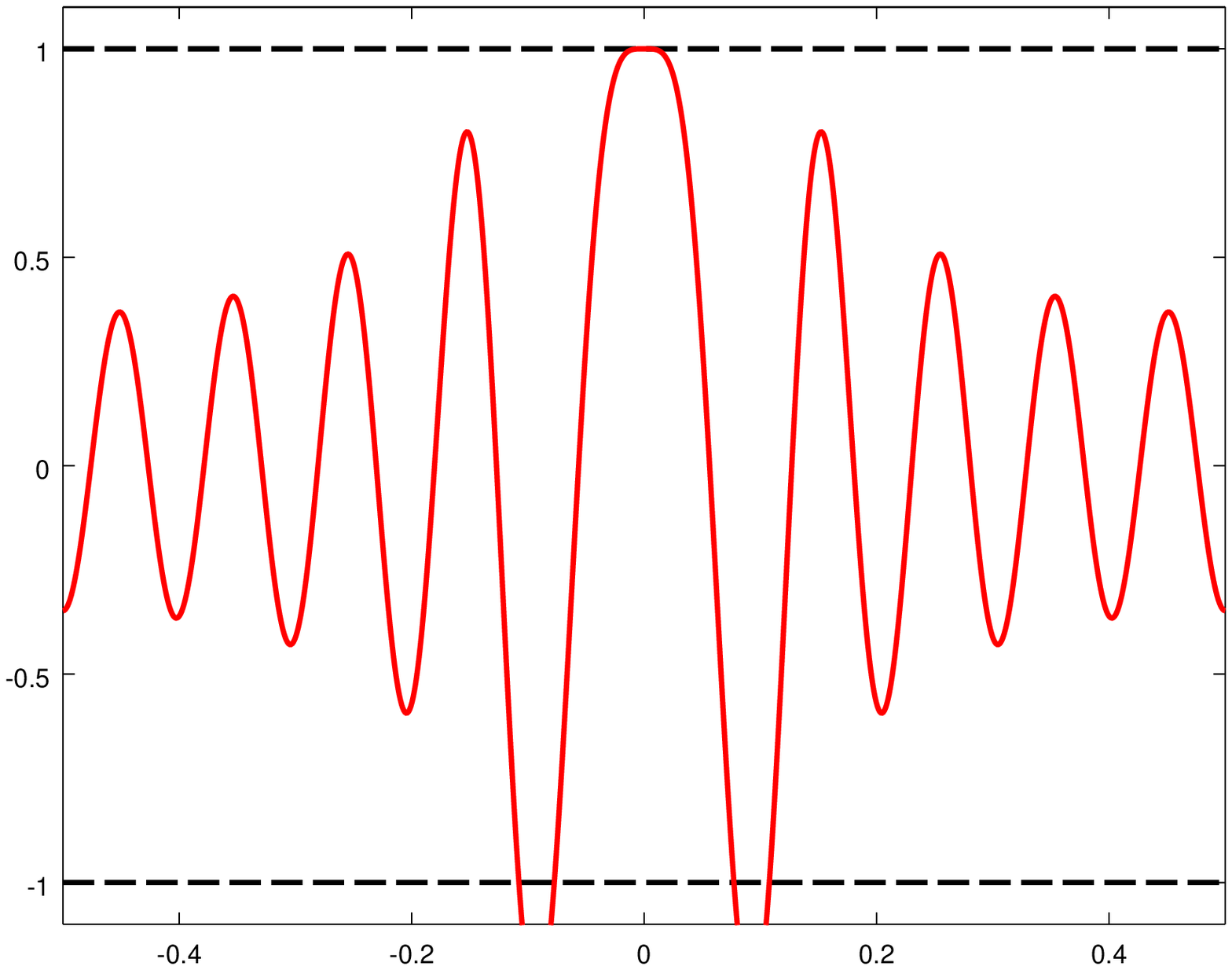}
 \vspace{-0.4cm}\\
 \includegraphics[width=0.22\linewidth]{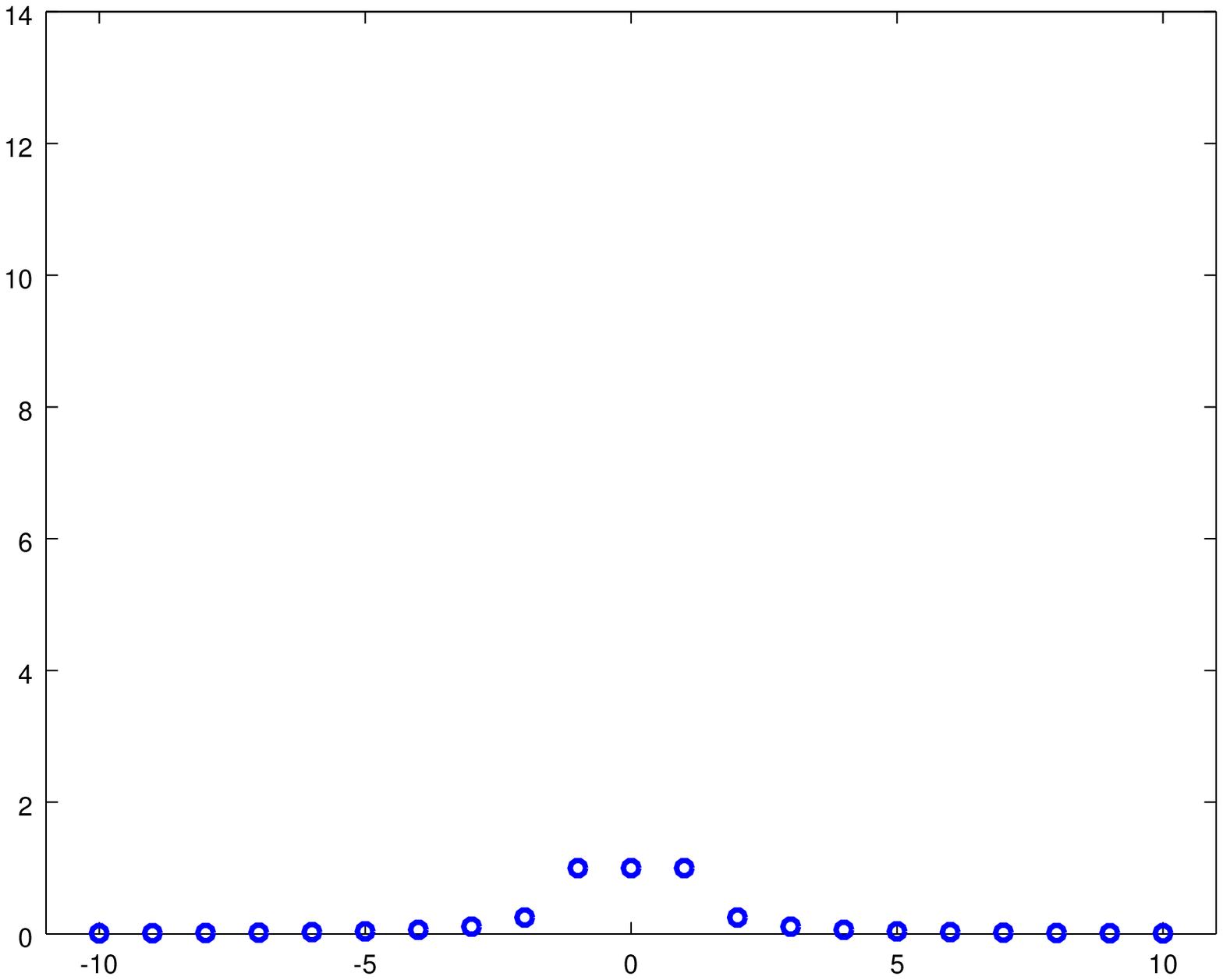} &
 \includegraphics[width=0.22\linewidth]{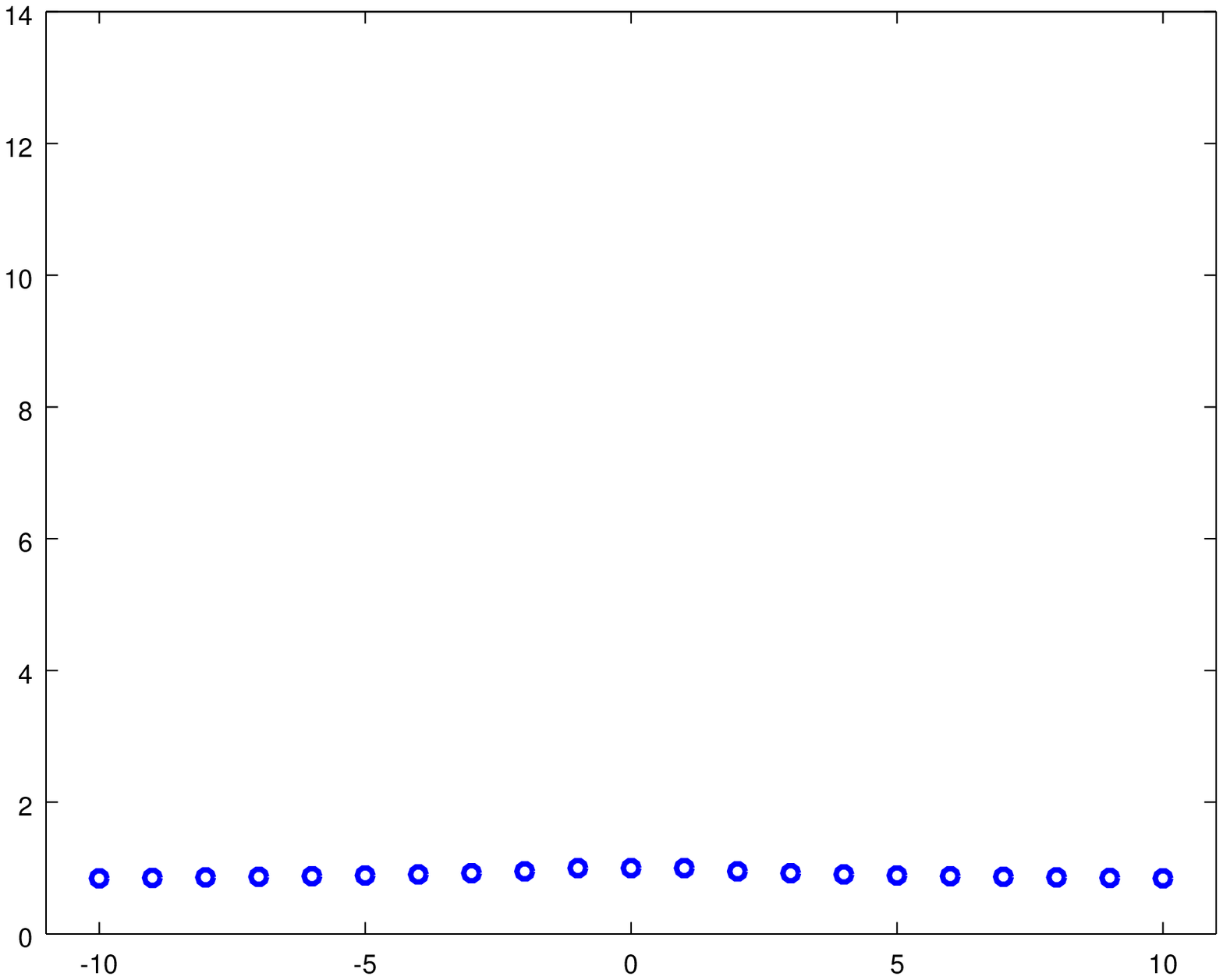} &
 \includegraphics[width=0.22\linewidth]{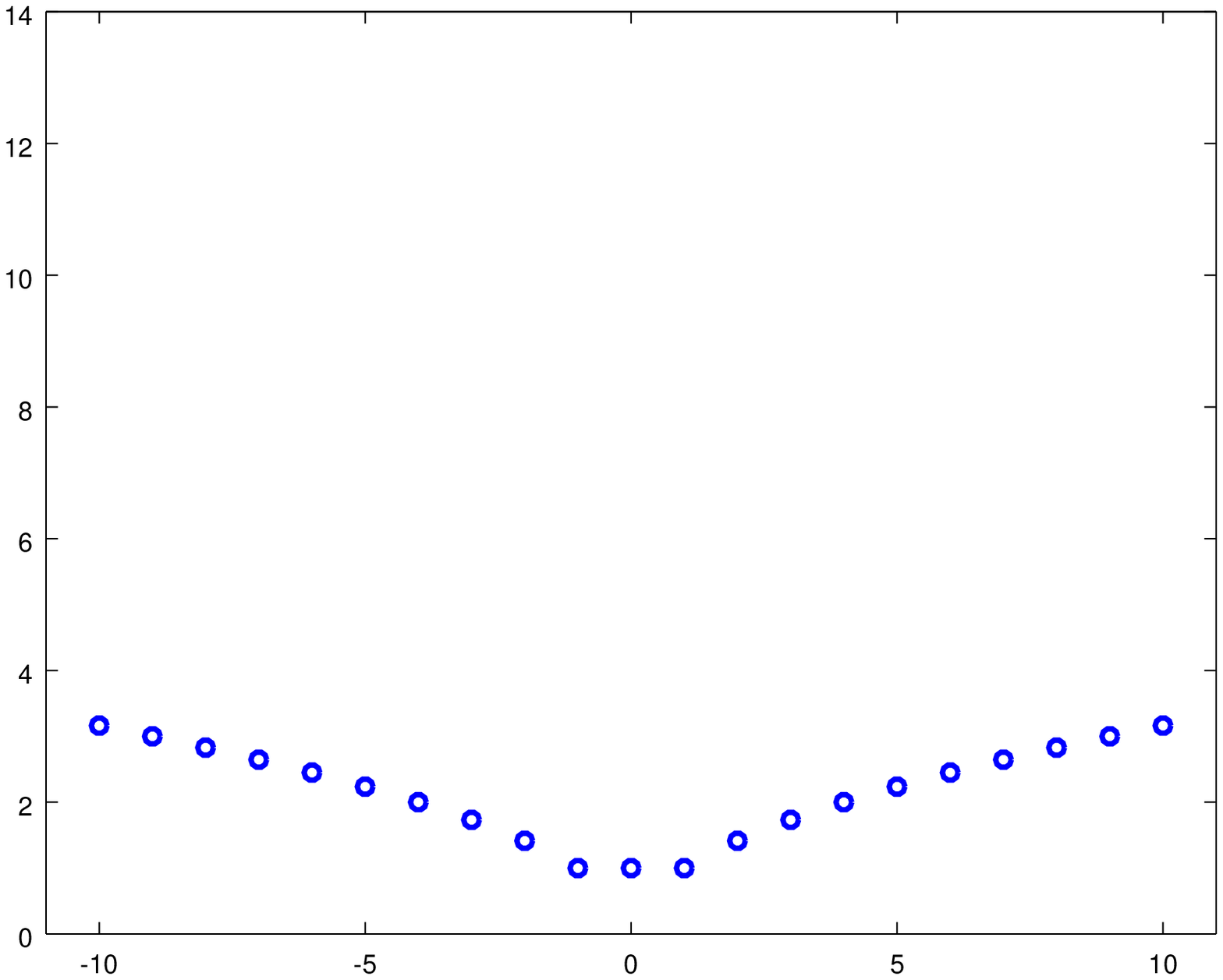} &
 \includegraphics[width=0.22\linewidth]{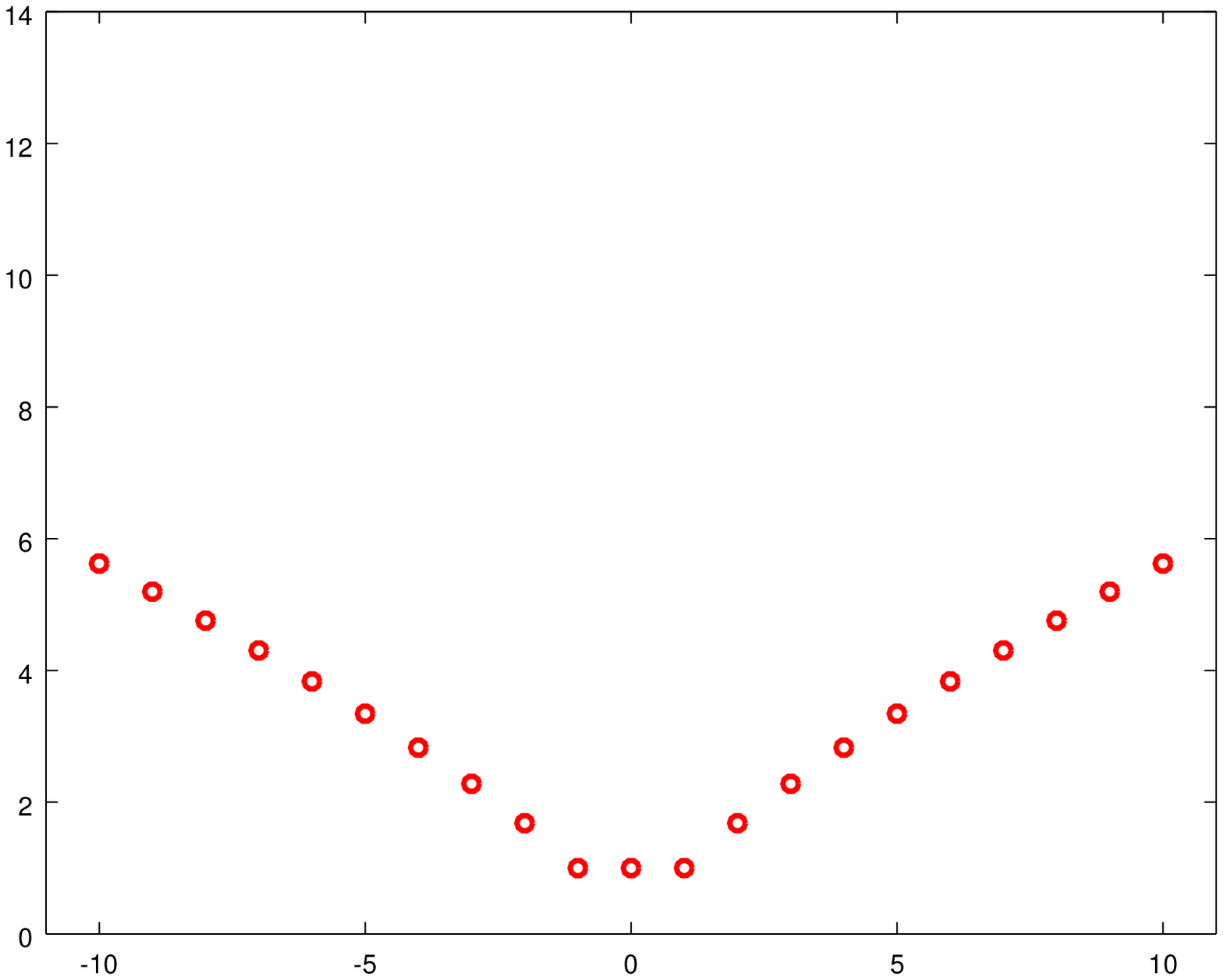} &
 \includegraphics[width=0.22\linewidth]{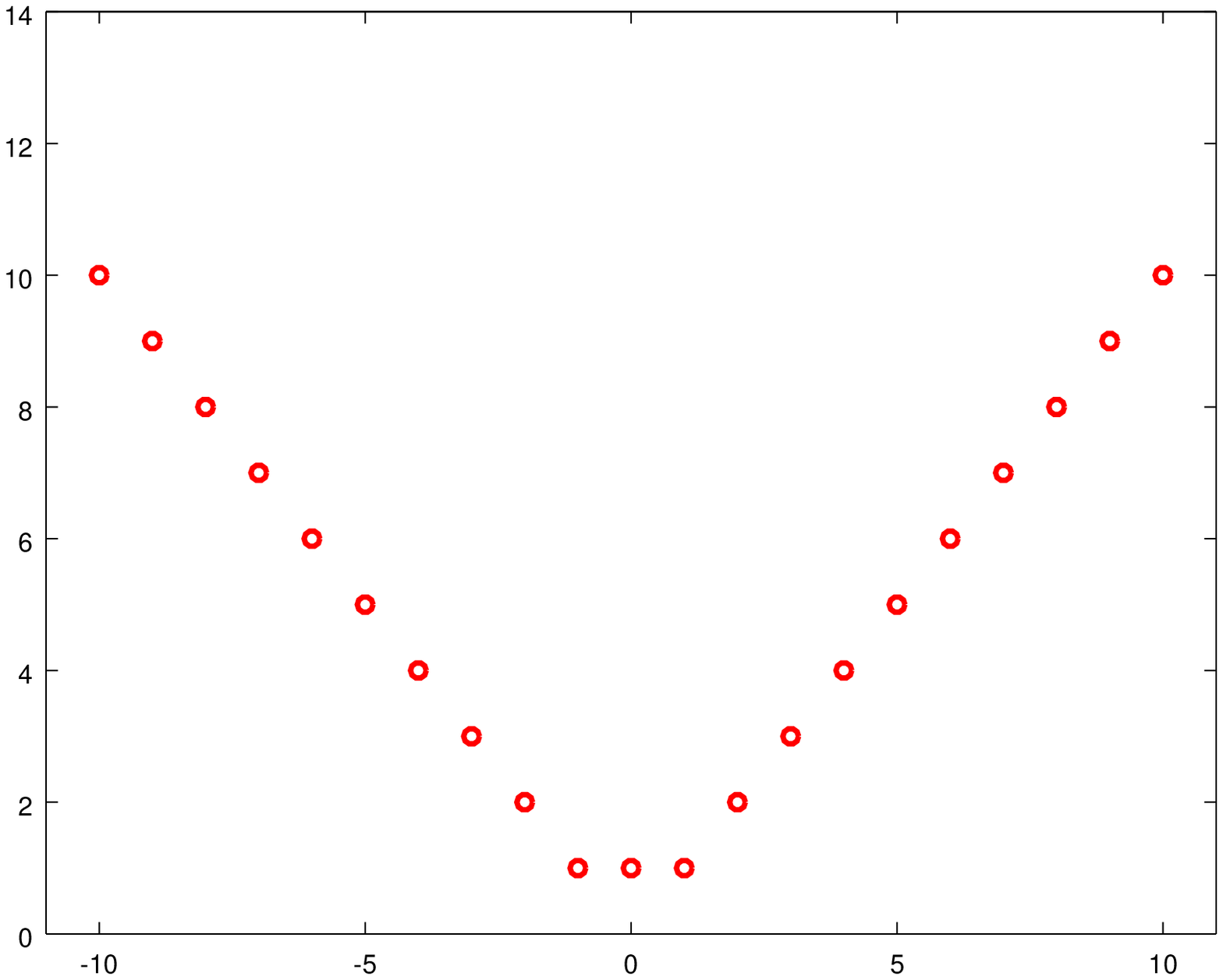}
 \vspace{-0.4cm}\\
 \end{tabular}
 \caption{\label{etaW-comp-incr}$\etaW$ for a low pass filter for $f_c=10$ with increasing high frequency content. 
 \textit{First row: }$\etaW$. \textit{Second row: }associated Fourier coefficients of the filter. 
 The curve showing $\etaW$ is in blue when it is $(2N-1)$-non-degenerate and in red when it is degenerate.}
\end{figure}

\subsection{Main Contribution}
\label{sec-main-contrib}

We now state our main contribution.

%

\begin{thm}\label{thm-main}
	Suppose that $\varphi\in\kernel{2N+1}$ and that $\etaW$ is $(2N-1)$-non-degenerate. 
	Then there exist constants $(t_1,\tW,C,\constR,M)$ (depending only on $\phi$, $\ampO$ and $\poszO$) such that 
	for all $0<t<\min\pa{ t_1,\tW}$, 
	for all $(\la,w)\in \ball{0}{\constR t^{2N-1}}$ with $\normObs{\frac{w}{\la}}\leq C$, 
\begin{itemize}
  \item the problem $\blasso$ admits a unique solution,
  \item that solution has exactly $N$ spikes, and it is of the form $\meast$, with $(\amp,\posz)=\gs(\la,w)$ (where $\gs$ is a $\Cder{2N}$ function defined on $\ball{0}{\constR t^{2N-1}}\subset\RR\times\Obs$),
  \item the following inequality holds
    \eq{
		\normLiVec{(\amp,\posz)-(\ampO,\poszO)}\leq\constdg\pa{\frac{|\la|}{t^{2N-1}}+\frac{\normObs{w}}{t^{2N-1}}}.
	}
\end{itemize}
\end{thm}

Note that the value of constants involved can be found in the proof of the theorem, more precisely, $(t_1,C)$ are defined in~\eqref{eq-cst-nondeg}, $\constR$ is defined in Proposition~\ref{prop-tfi-ext}, $\tW$ is defined in Theorem~\ref{if-etaw-nondegen} and $M$ is given in Corollary~\ref{maj-dg}.

The proof of Theorem~\ref{thm-main} uses results spanning Sections~\ref{sec-prelim}, \ref{sec-tfi} and \ref{sec-etaL}. Below, we give a sketch of proof to guide the reader through the remaining of the paper. The elements of the proof are divided in three main steps.


\paragraph{Step 1 (Section~\ref{sec-first-order-opt-cond} and~\ref{sec-tfi-tfi}).} We start with the first order optimality equation that any solution $\meast$, for fixed $t>0$, of $\blasso$ must satisfy \ie
 \begin{align*}
 	\Gatz^*\pa{ \Phitz\amp-\PhitzO\ampO-w}+\la \begin{pmatrix}\bun_N\\ 0 \end{pmatrix}=0.
 \end{align*}
 It is obtained by applying Fermat's rule to the problem $\blasso$. Since the parameters $(\amp,\posz,\la,w)=(\ampO,\poszO,0_\RR,0_{\Obs})$ are a solution of the equation, the idea is to parametrize, in a neighborhood of $(\la,w)=(0_\RR,0_{\Obs})$, the amplitudes and positions $(\amp,\posz)$ in terms of $(\la,w)$ by applying the implicit function theorem so that $(\amp,\posz,\la,w)$ is a solution of the first order equation. This process is detailed in Section~\ref{sec-first-order-opt-cond} and~\ref{sec-tfi-tfi}. The rest of the proof consists in proving that the measure $\meast$ is the unique solution of the problem $\blasso$. But before we have to deal with the domain of existence of the above parametrization.
 
\paragraph{Step 2 (Section~\ref{sec-extension-tfi}).} The implicit function theorem only provides the existence of a neighborhood in $(\la,w)$ of $(0_\RR,0_{\Obs})$ where the parametrization holds,  but we do not know how it size depends on the parameter $t$. This issue is important because one of our aims is to determine the constraints on $t$ (corresponding roughly to the minimum distance between the spikes of the original measure), on the noise level and on the regularization parameter $\la$ so that the recovery of the support is possible. Section~\ref{sec-extension-tfi} is devoted to show that the parametrization, which writes $(\amp,\posz)=\gs(\la,w)$ (see Equation~\eqref{eq-fimpext} for the definition of the implicit function $\gs$), of the solutions of the first order optimality equation holds in a neighborhood of $(0_{\RR},0_{\Obs})$ and of size proportional to $t^{2N-1}$. This result corresponds to Proposition~\ref{prop-tfi-ext}. The proof uses an upper bound of $\d\gs$ which is stated in Corollary~\ref{maj-dg}. Proposition~\ref{prop-tfi-ext} relies on asymptotic expansions (of $\Gatz$ for example), when $t\to 0$, gathered and used in the proof of Lemma~\ref{maj-df}. Section~\ref{sec-prelim} is devoted to these asymptotic expansions and it may be skipped at first reading.
 
\paragraph{Step 3 (Section~\ref{sec-etaL}).} Up to now, we have constructed a candidate solution $\meast$ (composed of $N$ spikes) where $(\amp,\posz)=\gs(\la,w)$ is built by parametrizing the solutions of the first order optimality equation. Moreover this parametrization holds for all $(\la,w)$ in a ball of radius proportional to $t^{2N-1}$. It remains to prove that $\meast$ is indeed the unique solution of $\blasso$. To prove that $\meast$ is a solution, it is equivalent to check that
 \eq{
 	0\in \partial\pa{m\mapsto \frac{1}{2}\normObs{\Phi m -\obs-w}^2+\la\normTVX{m}}(\meast),
}
which reformulates into $\etaL\eqdef\frac{1}{\la}\Phi^*(\obs+w-\Phi\meast)\in\partial\normTVX{\meast}$. This is done by first showing the convergence of $\etaL$ towards $\etaW$ when $(t,\la,w)\to 0$ in a well chosen domain, see Proposition~\ref{etaL-nondegen}, and then using Theorem~\ref{if-etaw-nondegen} and the fact that $\etaW$ is ensured to be $(2N-1)$-non-degenerate (which is one of the hypotheses of Theorem~\ref{thm-main}) to get the non-degeneracy of $\etaL$ and the conclusion. 
 
\paragraph{Putting all together.}

After this sketch, we now give the detailed proof. It uses Proposition~\ref{prop-tfi-ext} (parametrization of the solution of the first order optimality equation on a ball, for the parameter $(\la,w)$, of radius proportional to $t^{2N-1}$), Proposition~\ref{etaL-nondegen} (convergence of $\etaL$ towards $\etaW$), Theorem~\ref{if-etaw-nondegen} (use of the $(2N-1)$-non-degeneracy of $\etaW$ to transfer it to $\etaL$), and Proposition~\ref{first-order-dev-sol} (upper bound on the error of $\meast$ with respect to $\meastO$).

\begin{proof}[Proof of Theorem~\ref{thm-main}]
Let us take $t,\la,w$ as in the hypotheses of the Theorem~\ref{thm-main}. Let $(\amp,\posz)=\gs(\la,w)$, where $\gs$ is the function constructed in Section~\ref{sec-tfi}.
Let us define 
\eq{
	\pL\eqdef\frac{1}{\la} \pa{ \Phi_{\postO}\ampO+w-\Phi_{\post}\amp}
	\qandq
	\etaL\eqdef\Phi^* \pL.
}
By Proposition~\ref{etaL-nondegen} combined with Theorem~\ref{if-etaw-nondegen} where we take 
\begin{align}\label{def-RW}
	\RW\eqdef \sup\{\normLiVec{z};z\in\Bo\},
\end{align}
we have for $0<t< \min(\tW, t_1)$,  
\begin{align}\label{etaL-non-degenerate-cond}
  \forall \pos\in \Pos\setminus\bigcup_i \{\post_i\}, |\etaL(\pos)|<1 \qandq \forall 1\leq i \leq N, \ \etaL''(\post_i)<0,
\end{align}
while $\etaL(\post_i)=1=\sign(\amp_i)$ by definition.

We deduce that $\etaL$ is in the subdifferential of the total variation at $\meast$ because 
\begin{itemize}
\item $\etaL\in\ContX$,
\item $\normLi{\etaL}\leq 1$ thanks to Equation~\ref{etaL-non-degenerate-cond},
\item $\forall 1\leq i\leq N$, $\etaL(\post_i)=1=\sign(\amp_i)$ by definition of $\etaL$ (recall that $(\amp,\posz)=\gs(\la,w)$).
\end{itemize}
As a result $\meast$ is a solution to $\blasso$ and $\pL$ is the unique solution to the dual problem associated to $\blasso$ (see Section 2.4 of~\cite{duval-exact2013} for details on dual certificates and optimality conditions for $\blasso$).

Let $m$ be an other solution of $\blasso$. Then the support of $m$ is included in the saturation points of $\etaL=\Phi^*\pL$ \ie in $\{\post_1,\ldots,\post_N\}$. As a result $m=m_{\amp',\post}$ for some $\amp'\in\RR^N$ and $m$ satisfies the first order optimality equation $\Gatz^*\pa{ \Phitz\amp'-\PhitzO\ampO-w}+\la \begin{pmatrix}\bun_N\\ 0 \end{pmatrix}=0$. Hence $\Phitz^*\Phitz \amp'=\Phitz^*\Phitz \amp$ and since $\Phitz$ has full rank (by assumption $t$ is chosen sufficiently small, see Lemma~\ref{lem-dl-ftz} for the proof), $\Phitz^*\Phitz$ is invertible and $\amp'=\amp$. So $m=\meast$ and $\blasso$ admits a unique solution: $\meast$.

The bound on the error between $(\amp,\posz)$ and the amplitudes and positions of the initial measure $(\ampO,\poszO)$ is a direct consequence of Proposition~\ref{first-order-dev-sol}.
\end{proof}

\subsection{Necessary condition for the recovery in the limit $t\to0$}
\label{sec-necessary}

Our main contribution, Theorem~\ref{thm-main}, states that under a non-degeneracy property which involves $\etaW$, it is possible to perform the recovery of the support of a measure $\meastO$ in the limit $t\to0$ when the data are contaminated by some noise, provided that $\max(|\la|/t^{2N-1},\normObs{w}/t^{2N-1},\normObs{w}/\la)\leq C$ for some constant $C>0$ depending only on the filter $\phi$ and $(\ampO,\poszO)$. It is natural to ask whether the non-degeneracy condition on $\etaW$, in order to get the recovery of the support in some low noise regime, is sharp.

The following Theorem shows that the $(2N-1)$-non-degeneracy assumption on $\etaW$ is almost sharp in the sense that the recovery of the support in a low noise regime leads to $\normLi{\etaW}\leq 1$.

\begin{thm}\label{thm-necessary}
Suppose that $\injdn$ holds and $\phi\in\kernel{2N+1}$. 
Suppose also that there exists a sequence $(t_n)_{n\in\NN}$ such that $t_n \to 0$ and satisfying
\begin{align*}
\forall n\in\NN, \exists (\la_n,w_n), \exists (\amp_n,\posz_n)\in\RR^N\times\RR^N, \measn \mbox{ is solution of } \Pp_{\la_n}(y_{t_n}+w_n),
\end{align*}
where $(\la_n,w_n)\to0$ with $\frac{\normObs{w_n}}{\la_n}\to0$. Then
\begin{align}\label{etaW-normi}
	\normLi{\etaW}\leq 1.
\end{align}
\end{thm}
The proof of this result can be found in Appendix~\ref{sec-proof-necessary-condition-etaW}.

The remaining sections of the paper, namely Sections~\ref{sec-prelim}, \ref{sec-tfi} and \ref{sec-etaL} are devoted to the proof of Theorem~\ref{thm-main}.



\section{Preliminaries}
\label{sec-prelim}
\label{sec-prelim-dl}

Our study relies to a large extent on the asymptotic behavior of quantities built upon $\Phitz$ and $\Gatz$ for $t>0$ small, such as $(\Phitz^* \Phitz)^{-1}$ or  $(\Gatz^* \Gatz)^{-1}$. In this section, we gather several preliminary results that enable us to control that behavior.

\paragraph{Approximate Factorizations.}

Our asymptotic estimates are based on an approximate factorization of $\Phix$ and $\Gax$ using Vandermonde and Hermite matrices. It enables us to study the asymptotic behavior of the optimality conditions of~$\blasso$ when $t\to 0^+$. In the following, we consider $\posz\in\Bo$ (see~\eqref{vois-ini}) and $t\in(0,1]$, so that $\Htz$ is always invertible. Moreover, we shall always assume that $\phi\in\kernel{2N}$.

\begin{prop}\label{prop-facto}
The following expansion holds
\begin{align}
  \Gatz= \Fdn \Htz + \GLambda_{t,\posz} D_t,
\end{align}
where $\Fdn$ is defined in~\eqref{eq-defn-Fk}, 
$\Htz$ is defined in~\eqref{eq-defn-Hx}, 
and where 
\begin{align*}
  \GLambda_{t,\posz}&\eqdef 
	\bigg(
  		\Big(
      \int_0^1 (\posz_i)^{2N}\varphi^{(2N)}(s\post_i)\frac{(1-s)^{2N-1}}{(2N-1)!} \d s
		\Big)_{1\leq i\leq N},  \\
		& \qquad 
		\Big(
    \int_0^1 (\posz_i)^{2N-1}\varphi^{(2N)}(s\post_i)\frac{(1-s)^{2N-2}}{(2N-2)!} \d s
		\Big)_{1\leq i\leq N} 
	\bigg) \\
	D_t&\eqdef \diag(t^{2N},\ldots,t^{2N},t^{2N-1}, \ldots, t^{2N-1}).
\end{align*}
\end{prop}

\begin{proof}
This expansion is nothing but the Taylor expansions for $\varphi$ and $\varphi'$:
\begin{align}
  \phi( \post_i) &= \phiD{0} +(\post_i)\phiD{1}+\ldots +\frac{(\post_i)^{2N-1}}{(2N-1)!}\phiD{2N-1}\nonumber\\
  &\qquad \qquad \qquad  +(\post_i)^{2N}\int_0^1\varphi^{(2N)}(s\post_i)\frac{(1-s)^{2N-1}}{(2N-1)!} \d s ,\label{eq-taylor1}\\
  \phi'(\post_i)&= \phiD{1} + (\post_i)\phiD{2}+\ldots +\frac{(\post_i)^{2N-2}}{(2N-2)!}\phiD{2N-1}\nonumber\\
  &\qquad \qquad \qquad  +(\post_i)^{2N-1}\int_0^1\varphi^{(2N)}(s\post_i)\frac{(1-s)^{2N-2}}{(2N-2)!} \d s.\label{eq-taylor2}
\end{align}
\end{proof}

The above expansion yields a useful factorization for $\Gatz$,
\begin{align*}
   \Gatz = \GF_{\post}\Htz  
   \qwhereq \GF_{\post} &\eqdef  \Fdn + \GLambda_{t,\posz} D_t \Htz^{-1}.
\end{align*}
The rest of this section is devoted to the consequence of that factorization for the asymptotic behavior of $\Gatz$ and its related quantities.
The main ingredient of this analysis is the factorization of $\Htz$ as
  \begin{align}\label{facto-Htz}
    \Htz=\diag(1,t,\ldots, t^{2N-1})\Hz\diag\left(1,\ldots, 1,\frac{1}{t},\ldots, \frac{1}{t}\right).
  \end{align}

Let us emphasize that our Taylor expansions are uniform in $\posz\in \Bo$. More precisely, given two quantities $f(\posz,t)$, $g(\posz,t)$, we say that 
$f(\posz,t)=g(\posz,t) + \bigO{t^k}$ if 
 \eq{\limsup_{t\to 0^+} \sup_{z\in \Bo} \abs{\frac{f(\posz,t) - g(\posz,t)}{t^k}}<+\infty.
}

\begin{lem}\label{lem-dl-ftz}
The following expansion holds for $t\to 0^+$,
\eql{\label{eq-taylor-GF}\GF_{\post}= \Fdn + \bigO{t}.}
  Moreover, if $\injdn$ holds then $\GF_{\post}$ and $\Gatz$ have full column rank for $t>0$ small enough and 
\begin{align}
\label{eq-taylor-invgram}  (\GF_{\post}^*\GF_{\post})^{-1} &=(\Fdn^*\Fdn)^{-1}+\bigO{t}\\
    \Gatz^{+,*}\begin{pmatrix}\bun_N\\ 0\end{pmatrix} &= \Fdn^{+,*}\dirac{2N} + \bigO{t}.
\end{align}
  
\end{lem}

\begin{proof}
We begin by noticing that 
\begin{align*}
  \GLambda_{t,\posz}D_t\Htz^{-1}&= t^{2N}\GLambda_{t,\posz}\Hz^{-1}\diag(1,1/t,\ldots,1/t^{2N-1})\\
  &=\GLambda_{t,\posz}\Hz^{-1}\diag(t^{2N},t^{2N-1},\ldots,t).
\end{align*}
The function $z\mapsto \Hz^{-1}$ is $\Cder{\infty}$ and uniformly bounded on $\Bo$, and $(\posz,t)\mapsto \GLambda_{t,\posz}$ is $\Cder{0}$ on the compact set $\Bo\times [0,1]$ hence uniformly bounded too. As a result, we get~\eqref{eq-taylor-GF}.

Assume now that $\injdn$ holds. Since $\Fdn^*\Fdn$ is invertible, there is some $R>0$ such that for every $A$ in the closed ball $\bball{\Fdn^*\Fdn}{R}\subset \RR^{(2N-1)\times (2N-1)}$, $A$ is invertible. By the mean value inequality
  \begin{align*}
    & \norm{(\Fdn^*\Fdn)^{-1}-A^{-1}} \leq \!\!\! \sup_{B\in \bball{\Fdn^*\Fdn}{R}} \norm{B^{-1}(A-\Fdn^*\Fdn)B^{-1}}\\
    & \qquad\qquad \leq \left(\sup_{B\in \bball{\Fdn^*\Fdn}{R}}\norm{B^{-1}}\right)^2 \norm{A-\Fdn^*\Fdn}.
  \end{align*}
  Applying that to $A=\GF_{\post}^*\GF_{\post}= \Fdn^*\Fdn+\bigO{t}$ (since each term in the product is uniformly bounded), we get~\eqref{eq-taylor-invgram}.

Now, for the last point, we infer from $\Gatz=\GF_{\post}\Htz$ and the fact that $\Htz$ is invertible that $\Gatz^{+,*}=\GF_{\post}^{+,*}\Htz^{-1,*}$.
  Hence 
   \begin{align*}
     \Gatz^{+,*}
     \begin{pmatrix}\bun_N\\ 0\end{pmatrix}
     = \
     \GF_{\post}^{+,*} \dirac{2N}
     = 
     \GF_{\post}(\GF_{\post}^*\GF_{\post})^{-1} \dirac{2N}
   \end{align*}
   where $\dirac{2N}$ is defined in~\eqref{eq-dirac}.

   Each factor below being uniformly bounded in $\Bo\times [0,1]$, we get
   \begin{align*}
     \Gatz^{+,*}
     \begin{pmatrix}\bun_N\\ 0\end{pmatrix} &= \left(\Fdn +\bigO{t}\right)\left((\Fdn^*\Fdn)^{-1}+\bigO{t}\right) \dirac{2N}\\
     &= \Fdn(\Fdn^*\Fdn)^{-1}\dirac{2N}  +\bigO{t}.
   \end{align*}
\end{proof}

\paragraph{Projectors.}
In this paragraph, we shall always suppose that $\injdn$ holds.
Another important quantity in our study is the orthogonal projector $P_{(\Im \Gatz)^\perp}$ (resp. $P_{(\Im \Fdn)^\perp}$) onto $(\Im \Gatz)^\perp$ (resp. $(\Im \Fdn)^\perp$). We define 
\begin{align*}
  \GPi_{t\posz} & \eqdef P_{(\Im \Gatz)^\perp} = \Id_{\Obs}-\Gatz(\Gatz^*\Gatz)^{-1}\Gatz^*,\\
  \GPin &\eqdef P_{(\Im \Fdn)^\perp}= \Id_{\Obs}-\Fdn(\Fdn^*\Fdn)^{-1}\Fdn^*.
\end{align*}
Observing that $P_{(\Im \Gatz)^\perp}=P_{(\Im \GF_{\post})^\perp}$, we immediately obtain from the previous Lemma that $\GPi_{t\posz}=\GPin +\bigO{t}$.

By construction, $\GPi_{t\posz}\Phitz=\GPi_{t\posz}\Phitz'=0$, but the following proposition shows that this quantity is also small if we replace $\Phitz$ with $\Phitz''$.
\begin{lem}\label{lem-dl-gpiphiseconde}
  There exists a constant $\Lgpiphiseconde>0$ (which only depends on $\varphi$, $\ampO$ and $\poszO$) such that
  \begin{align*}
    \normObs{\GPi_{\post}\Phitz''}\leq \Lgpiphiseconde t^{2N-2}
  \end{align*}
  uniformly in $\posz\in \Bo$.
\end{lem}

\begin{proof}
  Applying a Taylor expansion to $\varphi^{(2)}$, we write 
  \eq{
    \Phitz''= \Fdn \tilde{\GV}_{\post} +t^{2N-2}\tilde{\GLambda}_{t,\posz}
  }
  where
  \eql{\label{eq-prelim-vandermonde}\begin{split}
    \tilde{\GV}_{\post} & = 
    \begin{pmatrix} 0 &\ldots & 0 \\
                                0&\ldots & 0\\
                              1&\ldots & 1\\
                            \vdots & & \vdots\\
                            \frac{(\post_1)^{2N-3}}{(2N-3)!}&\ldots & \frac{(\post_N)^{2N-3}}{(2N-3)!} 
	\end{pmatrix},  \\ 
  \tilde{\GLambda}_{t,\posz} & = 
    \begin{pmatrix}
      (z_i)^{2N-2}\int_{0}^1 \varphi^{(2N)}(s\post_i)\frac{(1-s)^{2N-3}}{(2N-3)!} \d s 
	\end{pmatrix}_{1\leq i\leq N}.
  \end{split}
}
 Hence
\begin{align*}
  \GPi_{\post}\Phitz'' &= \GPi_{\post}(\GF_{\post}\tilde{\GV}_{\post} + (\Fn-\GF_{\post})\tilde{\GV}_{\post} +t^{2N-2}\tilde{\GLambda}_{t,\posz})\\
  &=\GPi_{\post} (-\GLambda_{t,\posz}D_t\Htz ^{-1}\tilde{\GV}_{\post} +t^{2N-2}\tilde{\GLambda}_{t,\posz}) \mbox{ since }\GPi_{\post}\GF_{\post}=0.
  \intertext{Using~\eqref{facto-Htz}, we see that}
  D_t\Htz^{-1}\tilde{\GV}_{\post}&=\Hz^{-1}\diag(t^{2N},t^{2N-1},\ldots,t)\tilde{\GV}_{\post}\\
  &=t^{2N-2}\Hz^{-1}\tilde{\GV}_{\posz}, \quad \mbox{ hence}\\
 \GPi_{\post}\Phitz'' &= t^{2N-2}\GPi_{\post} (-\GLambda_{t,\posz}\Hz^{-1}\tilde{\GV}_{\posz}+\tilde{\GLambda}_{t,\posz}).
\end{align*}
Since $\norm{\GPi_{\post}}\leq 1$ and the continuous function $(z,t)\mapsto -\GLambda_{t,\posz}\Hz^{-1}\tilde{\GV}_{\posz}+\tilde{\GLambda}_{t,\posz}$ is uniformly bounded on the compact set $\Bo\times [0,1]$, we obtain
\eq{ \normObs{\GPi_{\post}\Phitz''}\leq \left(\sup_{(z',t')\in\Bo\times[0,1]} \normObs{\GLambda_{t',\posz'}H_{\posz'}^{-1}\tilde{\GV}_{\posz'}+\tilde{\GLambda}_{t',\posz'}}\right)  t^{2N-2}
}
\end{proof}

We study further the projector $\GPi_{\post}$  when it is not evaluated at the same $\posz$ as $\Gatz$.

\begin{lem}\label{lem-proj-gamma}
If $\varphi\in\kernel{2N+1}$, then there is a constant $\constPi>0$ (which only depends on $\varphi$, $\ampO$ and $\poszO$) such that for all $\posz\in \Bo$, all $t\in (0,1]$,
\begin{align*}
  \normObs{\GPi_{\post}\GatzO\AmpO}\leq \constPi t^{2N}\normLiVec{\posz-\poszO}.
\end{align*} 
\end{lem}

\begin{proof}
Let us observe that
\begin{align*}
    \GPi_{\post}\GatzO&= \GPi_{\post}(\Fdn \HtzO+\GLambda_{t,\poszO}D_t)\\
    &=\GPi_{\post}(\GF_{\post} \HtzO+ (\Fdn-\GF_{\post})\HtzO+\GLambda_{t,\poszO}D_t)\\
    &=\GPi_{\post}(-\GLambda_{t,\posz}D_t\Htz^{-1}\HtzO+\GLambda_{t,\poszO}D_t) \mbox{ since $\GPi_{\post}\GF_{\post}=0$.}
\end{align*}
Observing that 
\eq{
	D_t\Htz^{-1}\HtzO = t^{2N}\Hz^{-1}\HzO 
	\diag\pa{1, \ldots, 1, 1/t, \ldots, 1/t}
	=\Hz^{-1}\HzO D_t,
} 
we get
\eq{
   	\GPi_{\post}\GatzO = 
	\GPi_{\post}\left(\GLambda_{t,\poszO}(\Id_{2N}-\Hz^{-1}\HzO) + 
	(\GLambda_{t,\poszO}-\GLambda_{t,\posz})\Hz^{-1}\HzO\right)D_t.
}

For $k\in \{2N-1,2N\}$, the function $(u,s,t)\mapsto u^k\varphi^{(2N)}(stu)$ is defined and $\Cder{1}$ on the compact sets (since $\phi\in\kernel{2N+1}$)
\eq{
	\foralls 1\leq i\leq N, \quad
	\left[{\poszO}_i-\frac{\De_0}{4},{\poszO}_i+\frac{\De_0}{4}\right]\times [0,1]\times [0,1], 
}
where $\De_0$ is defined in~\eqref{min-sep-dist-ini}. Thus there is a constant $C>0$ (which does not depend on $t$ nor $\posz\in\Bo$) such that 
  \begin{align*}
    \left|\int_0^1 \left( (z_i)^{k}\varphi^{(2N)}(s\post_i)- ({\poszO}_i)^{k}\varphi^{(2N)}(st{\poszO}_i)\right)\frac{(1-s)^{k-1}}{(k-1)!} \d s\right| &\leq C\abs{\posz_i-{\poszO}_i},\\
    \mbox{hence }\quad\quad \norm{\GLambda_{t,\poszO}-\GLambda_{t,\posz}}&\leq C\normLiVec{\posz-\poszO}.
  \end{align*}
As a result, since $\norm{\GPi_{\post}}\leq 1$ and $z\mapsto \Hz^{-1}\HzO$ is bounded on $\Bo$,
  \eq{\normObs{\GPi_{\post} (\GLambda_{\postO}-\GLambda_{t,\posz})\Hz^{-1}\HzO}\leq C\sup_{z'\in\Bo}\norm{H_{z'}^{-1}\HzO} \normLiVec{\posz-\poszO}.}
  As for the left term, $\GLambda_{\postO}$ is bounded uniformly in $t\in [0,1]$, and the mapping $z\mapsto  \Hz^{-1}\HzO$ is $\Cder{1}$ on $\Bo$. As a result, there is a constant $\tilde{C}>0$ such that 
  \eq{\forall z\in\Bo,\quad  \norm{\Id_N-\Hz^{-1}\HzO}\leq \tilde{C}\normLiVec{\posz-\poszO}.}

To conclude, we observe that $D_t\AmpO=t^{2N}\AmpO$, and we combine the above inequalities to obtain
\eq{\normObs{\GPi_{\post}\GatzO\AmpO}\leq \left(C\sup_{z'\in\Bo}\norm{H_{z'}^{-1}\HzO} +\tilde{C}\sup_{t\in[0,1]}\norm{\GLambda_{\postO}}\right)t^{2N}\normLiVec{\posz-\poszO}.}
\end{proof}


\paragraph{Asymptotics of the vanishing derivatives precertificate.} 

We end this section devoted to the asymptotic behavior of quantities related to $\Gatz$ by studying the second derivative of the vanishing derivatives precertificate $\etaV$ (see Definition~\ref{propdef-vanishing}, and \cite{duval-exact2013} for more details).
Theorem~\ref{if-etaw-nondegen} ensures that the second derivatives of $\etaV$ do not vanish at $\poszOi$, $1\leq i\leq N$. However, it does not provide any estimation of those second derivatives. That is the purpose of the next proposition.

In view of Section~\ref{sec-tfi}, it will be useful to study those second derivatives not only for the precertificates that are defined by interpolating the sign at  $\postO$ but more generally for the precertificates that are defined to interpolate the sign at $tz$ for any $\posz\in \Bo$.

\begin{prop}\label{asympt-etaV2}
  Assume that $\varphi\in\kernel{2N+1}$ and that $\injdn$ holds.
  Then 
  \begin{align}
    {\Phitz''}^*\Gatz^{+,*}\begin{pmatrix}\bun_N \\ 0\end{pmatrix} &= t^{2N-2}\etaW^{(2N)}(0)\dz +\bigO{t^{2N-1}},\\
    \qwhereq \dz\in \RR^N, \quad \dzi \eqdef & \frac{2}{(2N)!} \prod_{j\neq i}(\posz_i-\posz_j)^2 \ \mbox{for}\ 1\leq i\leq N.
  \end{align}
\end{prop}

\begin{proof}
  We proceed as in the proof of Lemma~\ref{lem-dl-gpiphiseconde} by writing $\Phitz''= \Fdn \tilde{\GV}_{\post} +t^{2N-2}\tilde{\GLambda}_{t,\posz}$ (see~\eqref{eq-prelim-vandermonde}) and  $\GF_{\post}=\Gatz\Htz^{-1}$. We obtain
\begin{align*}
  \Phitz''&=\GF_{\post}\tilde{\GV}_{\post}+t^{2N-2}\tilde{\GLambda}_{t,\posz}-t^{2N-2}\GLambda_{t,\posz}\Hz^{-1}\tilde{\GV}_{\posz}.
\end{align*}
The first term yields
\eql{\label{eq-etaV2-un}
  \tilde{\GV}_{\post}{}^* \GF_{\post}^* \Gatz^{+,*}\begin{pmatrix}\bun_N \\ 0\end{pmatrix} 
  \!=  
  \tilde{\GV}_{\post}{}^*\GF_{\post}^*\GF_{\post}(\GF_{\post}^*\GF_{\post})^{-1}\Htz^{-1,*}\begin{pmatrix}\bun_N \\ 0\end{pmatrix}
  \!= 
  \tilde{\GV}_{\post}{}^* \dirac{2N} = 0.
}
As for the second term, we take the Taylor expansion a little further (using integration by parts),
\begin{align*}
  \int_0^1 \varphi^{(2N)}(s\post_i)\frac{(1-s)^{2N-3}}{(2N-3)!} \d s 
  &=  
  \frac{\phiD{2N}}{(2N-2)!} + \\
  & \quad \post_i\int_0^1\varphi^{(2N+1)}(s\post_i)\frac{(1-s)^{2N-2}}{(2N-2)!} \d s,
\end{align*}
so as to obtain 
\eq{
  \tilde{\GLambda}_{t,\posz} = 
  \pa{\phiD{2N}, \ldots, \phiD{2N}}
  \diag(\ez) + \bigO{t},
  \ \mbox{where} \ 
  \ez \eqdef \pa{\frac{(\posz_i)^{2N-2}}{(2N-2)!}}_{1\leq i\leq N} \in \RR^N
} 
and, as usual, $\bigO{t}$ is uniform in $\posz\in\Bo$.
From Lemma~\ref{lem-dl-ftz}, we also know that $\Gatz^{+,*}\begin{pmatrix}\bun_N \\ 0\end{pmatrix}=\pW+O(t)$, hence 
\eql{\label{eq-etaV2-deux}
  \tilde{\GLambda}_{t,\posz}^*\Gatz^{+,*}\begin{pmatrix}\bun_N \\ 0\end{pmatrix} =
    \diag(\ez) \begin{pmatrix}\dotObs{\phiD{2N}}{\pW}
      \\\vdots\\\dotObs{\phiD{2N}}{\pW}
    \end{pmatrix} + \bigO{t}= \etaW^{(2N)}(0) \ez +\bigO{t}.
}

Now, we proceed with the last term. Similarly, by integration by parts,
\begin{align*}
  \GLambda_{t,\posz} &= 
  \begin{pmatrix}\phiD{2N}&\ldots&\phiD{2N}\end{pmatrix}  \diag(\hz)  + O(t).
\end{align*}
where $\hz$ is defined by
\eql{\label{eq-defn-hz}
	\hz \eqdef \left(\frac{(\posz_1)^{2N}}{(2N)!},\ldots,\frac{(\posz_N)^{2N}}{(2N)!}, \frac{(\posz_1)^{2N-1}}{(2N-1)!},\ldots, \frac{(\posz_N)^{2N-1}}{(2N-1)!}\right) \in \RR^{2N}.
}
 Hence, 
\begin{align*}
  \GLambda_{t,\posz}^*\Gatz^{+,*}\begin{pmatrix}\bun_N \\ 0\end{pmatrix}
  &=
  \diag(\hz)  \begin{pmatrix}
    \dotObs{\phiD{2N}}{\pW}\\\vdots \\ \dotObs{\phiD{2N}}{\pW}
	\end{pmatrix}+O(t)
 =  \etaW^{(2N)}(0) \hz
  +O(t).
\end{align*} 

To conclude, we study $\tilde{\GV}_{\posz}{}^*\Hz^{-1,*}$ (which is uniformly bounded on $\Bo$). In $\RR_{2N-1}[X]$ endowed with the basis $\left(1,X,\ldots, \frac{X^{2N-1}}{(2N-1)!}\right)$, $\Hz^{*}$ is the matrix of the linear map which evaluates a polynomial and its derivatives at $\{\posz_1,\ldots,\posz_N\}$. On the other hand $\tilde{\GV}_{\posz}{}^*$ represents the evaluation of the second derivative at $\{\posz_1,\ldots,\posz_N\}$. Thus, 
\eq{
  \tilde{\GV}_{\posz}{}^*\Hz^{-1,*} \hz = (P''(\posz_i))_{1\leq i\leq N}, 
}
where $P$ is the unique polynomial in $\RR_{2N-1}[X]$ which satisfies 
\eq{
	\foralls  i = 1, \ldots, N, \quad
	P(\posz_i)=\frac{(\posz_i)^{2N}}{(2N)!}
	\qandq
	P'(\posz_i)=\frac{(\posz_i)^{2N-1}}{(2N-1)!}. 
}
One may check that 
  \begin{align*}
    P(X)&=\frac{X^{2N}}{(2N)!}- \frac{1}{(2N)!}\prod_{i=1}^N (X-\posz_i)^2\\
    \qandq  P''(\posz_i)&=\frac{\posz_i^{2N-2}}{(2N-2)!}-\frac{2}{(2N)!}\prod_{j\neq i}(\posz_i-\posz_j)^2.
  \end{align*}
As a result, 
\eql{\label{eq-etaV2-trois}
  -\tilde{\GV}_{\posz}{}^*\Hz^{-1,*}\GLambda_{t,\posz}^*\Gatz^{+,*}\begin{pmatrix}\bun_N \\ 0\end{pmatrix} = \etaW^{(2N)}(0)( \dz-\ez ) +\bigO{t},
}
where $\bigO{t}$ is uniform in $\posz\in\Bo$. We obtain the claimed result by summing \eqref{eq-etaV2-un}, \eqref{eq-etaV2-deux} and~\eqref{eq-etaV2-trois}. 
\end{proof}


\section{Building a Candidate Solution}
\label{sec-tfi}

Now that the technical issues regarding the asymptotic behavior of $\Gatz$ have been settled, we are ready to tackle the study of the BLASSO.
In this section, we build a candidate solution for $\blasso$ by relying on its optimality conditions.

\subsection{First Order Optimality Conditions}
\label{sec-first-order-opt-cond}

The optimality conditions for $\blasso$ (see~\cite{duval-exact2013}) state that any measure of the form $\meast=\sum_{i=1}^N \amp_i \delta_{\post_i}$ is a solution to $\blasso$ if and only if the function defined by $\etaL\eqdef\Phi^* \pL$ with $\pL\eqdef\frac{1}{\la} \pa{ \Phi_{\postO}\ampO+w-\Phi_{\post}\amp}$ satisfies $\normLi{\etaL}\leq 1$ and $\etaL(\post_i)=\sign(\amp_i)$  for all $1\leq i\leq N$.

Observe that we must have $\etaL'(\post_i)=0$  for all $1\leq i\leq N$. Moreover, in our case, since we assume that $\ampOi>0$ we have in fact  $\etaL(\post_i)=1$.

In order to build such a function $\etaL$, let us consider the function $\feq$ defined for some fixed $t>0$ on $\pa{\RR^N}^2\times\RR\times \Obs$ by
\begin{align}
  \feq(\apos,\param) \eqdef \Gatz^*\pa{ \Phitz\amp-\PhitzO\ampO-w}+\la \begin{pmatrix}\bun_N \\ 0\end{pmatrix} \label{eq-implicitfunc}\\
    \qwhereq \apos=(\amp,\posz) \qandq \param=(\la,w).
  \end{align}
Now, let us write $\aposO\eqdef(\ampO,\poszO)$.
Notice that $\meast$ is a solution to $\blasso$ if and only if $\feq(\apos,\param)=0$ and $\normLi{\etaL}\leq 1$. 
Our strategy is therefore to construct solutions of $\feq(\apos,\param)=0$ and to prove that $\normLi{\etaL}\leq1$ provided $(\la,w)$ and $\etaW$ satisfy certain properties. More precisely we start by parametrizing the solutions of $\feq(\apos,\param)=0$, in a neighborhood of $(\aposO,0)$, using the Implicit Function Theorem. 

The following Lemma~\ref{reg-f} (whose proof is omitted and corresponds to simple computations) shows that $\feq$ is smooth and gives its derivatives. 

\begin{lem}\label{reg-f}
If $\phi\in\kernel{k+1}$ for some $k\in\NN^*$ then $\feq$ is of class $\Cder{k}$ and for all $(\apos,\param)\in\pa{\RR^N}^2\times(\RR\times \Obs)$
\begin{align*}
	\partial_{\apos} \feq(\apos,\param) & = \Gatz^*\Gatz \diagIta+t\begin{pmatrix}
	0 & \diag(\Phitz'^*(\Phitz \amp-\PhitzO \ampO - w)) \\
	0 & \diag(\Phitz''^*(\Phitz \amp-\PhitzO \ampO - w)) 
	\end{pmatrix}\\	
	\partial_{\param} \feq(\apos,\param) &= \left(\begin{pmatrix}\bun_N \\ 0\end{pmatrix},-\Gatz^*\right)
\end{align*}
where $\diagIta \eqdef \begin{pmatrix}\Id_N & 0 \\ 0 & t\diag(\amp)\end{pmatrix}$. 
%
\end{lem}

\subsection{Implicit Function Theorem}
\label{sec-tfi-tfi}

Suppose that $\injdn$ holds and $\phi\in\kernel{2N+1}$. By the results of Section~\ref{sec-prelim}, there exists $0<t_0<1$ such that for $0<t<t_0$ and all $\posz\in \Bo$, $\Gatz^*\Gatz$ is invertible. In the following we shall consider a fixed value of such $t_0$  provided by Lemma~\ref{maj-df} below which also ensures additional properties.

Now, let $t\in (0,t_0)$ be fixed. By Lemma~\ref{reg-f}, $\feq$ is $\Cder{2N}$, $\partial_{\apos} \feq(\aposO,0)=\GatzO^*\GatzO \diagItaO$ is invertible and $\feq(\aposO,0)=0$. Hence by the Implicit Function Theorem, there exists $\Vt$ a neighborhood of $0$ in $\RR\times \Obs$, $\Ut$ a neighborhood of $\aposO$ in $\pa{\RR^N}^2$ and $\fimp:\Vt\to\Ut$ a $\Cder{2N}$ function such that
\begin{align*}
	\forall (\apos,\param)\in \Ut \times \Vt, \quad \feq(\apos,\param)=0 
	\quad\Longleftrightarrow\quad 
	\apos=\fimp(\param).
\end{align*}
Moreover, denoting $\d \fimp$ the differential of $\fimp$, we have
\eq{
	\foralls \param\in\Vt, \quad
	\d \fimp(\param) = -\pa{ \partial_{\apos} \feq(\fimp(\param),\param) }^{-1} \partial_{\param} \feq(\fimp(\param),\param).
}


\subsection{Extension of the Implicit Function $\fimp$}
\label{sec-extension-tfi}

Our goal is to prove that $\meast$ is the solution of the BLASSO, where $(\amp,\posz)=\apos=\fimp(\param)$. To this end, we shall exhibit additional constraints on $\param\in\Vt$, such as the scaling of the noise $\normObs{w}$ or $\la$ with respect to $t$, in order to ensure that $\normLi{\etaL}\leq 1$. 
However, the size of the neighborhood $\Vt$ provided by the Implicit Function Theorem is \textit{a priori} unknown, and it might implicitly impose even stronger conditions on $\la$ and $w$ as $t\to 0^+$.

Hence, before studying whether $\normLi{\etaL}\leq 1$, we show in this section that we may replace $V_t$ with some ball with radius of order $t^{2N-1}$ and still have a parametrization of the form $\apos=\fimp(\param)$ satisfying $\feq(\fimp(\param),\param)=0$ where $\feq$ is defined in~\eqref{eq-implicitfunc}.

Let $\Vs=\bigcup_{V\in \Vv} V$, where $\Vv$ is the collection of all open sets $V\subset \RR\times \Obs$ such that 
\begin{itemize}
  \item $0\in V$,
  \item $V$ is star-shaped with respect to $0$,
  \item $V\subset \ball{0}{\constTh t^{2N-2}}$, where $\constTh>0$ is a constant defined by Lemma~\ref{maj-df} below,
  \item there exists a $\Cder{2N}$ function $g:V\rightarrow (\RR^N)^2$ such that $g(0)=\aposO$ and $\feq(g(\param),\param)=0$ for all $\param\in V$,
  \item $g(V)\subset \ba\times \bo$.
\end{itemize}
Observe that $\Vv$ is nonempty (by the Implicit Function Theorem in Section~\ref{sec-tfi-tfi}) and stable by union, so that $\Vs\in \Vv$. Indeed, all the properties defining $\Vv$ are easy to check except possibly the last two. Let $V, \tilde{V}\in \Vv$ and $g,\tilde{g}$ be  corresponding functions. The set $\enscond{v\in V\cap \tilde{V}}{g(v)=\tilde{g}(v)}$ is nonempty (because $g(0)=\aposO=\tilde{g}(0)$) and closed in $V\cap \tilde{V}$. Moreover, it is open since for any $v\in V\cap \tilde{V}$, $v\in \ball{0}{\constTh t^{2N-2}}$ and, by Lemma~\ref{maj-df} below, the Implicit Function Theorem applies at $(g(v),v)$, yielding an open neighborhood in which $g$ and $\tilde{g}$ coincide. By connectedness of $V\cap \tilde{V}$, $g$ and $\tilde{g}$ coincide in the whole set $V\cap \tilde{V}$. As a result, the function $\gs: \Vs\rightarrow (\RR^N)^2$, defined by 
\begin{align}
  \gs(v)\eqdef g(v) \mbox{ if $v\in V$, $V\in\Vv$, and $g$ is a corresponding function},\label{eq-fimpext}
\end{align}
is well defined. Moreover, $\gs$ is $\Cder{2N}$ and $\gs(\Vs)\subset \ba\times \bo$.

Before proving that $\Vs$ contains a ball of radius of order $t^{2N-1}$ and studying the variations of $\gs$, we state Lemma~\ref{maj-df} mentioned above.

\newcommand{\OmPhi}{F_{t z}}

\begin{lem}\label{maj-df}
If $\injdn$ holds and $\phi\in\kernel{2N+1}$ then there exists $t_0\in (0,1)$ and $\constTh>0$ (which only depend on $\phi$ and $\aposO$) such that for all $t\in(0,t_0)$, if 
\begin{align}
  \apos=(\amp,\posz)\in \Ba\times\Bo \qandq \param=(\la,w)\in \ball{0}{\constTh t^{2N-2}},
\end{align}
then the matrix
\begin{align}\label{huge-matrix}
&\hugemat\eqdef\GF_{\post}^*\GF_{\post}+t
	\Htz^{*,-1} \OmPhi \diagIta^{-1}  \Htz^{-1}, \\
& \qwhereq
	\OmPhi \eqdef
	\begin{pmatrix} 
		0 & 0 \\ 
		0 & -\diag\pa{ \Phitz''^* q_{\post} }
  \end{pmatrix},\\
& \qandq
  q_{\post} \eqdef \la\Gatz^{*,+}\begin{pmatrix}\bun_N\\ 0\end{pmatrix}+\GPi_{\post} w+\GPi_{\post} \GatzO \AmpO,
\end{align}
is invertible and the norm of its inverse is less than $3\norm{(\Fdn^*\Fdn)^{-1}}$.

If, moreover, $\feq(\apos,\param)=0$, then 
\eq{
	\partial_{\apos} \feq(\apos,\param)=\Htz^* \hugemat \Htz\diagIta
}
and this is an invertible matrix.
\end{lem}
Let us precise that by $(\la,w)\in \ball{0}{\constTh t^{2N-2}}$, we mean that $\abs{\la}<\constTh t^{2N-2}$ and $\normObs{w}<\constTh t^{2N-2}$.

\begin{proof}
We consider $t_0\in (0,1)$ small enough so that for $0<t<t_0$ and all $\posz\in \Bo$, $\Gatz^*\Gatz$ is invertible and 
\begin{align}
	\label{construc-t0-GF}
  \norm{(\GF_{\post}^*\GF_{\post})^{-1}}&\leq 2 \norm{(\Fdn^*\Fdn)^{-1}} \mbox{ by Lemma~\ref{lem-dl-ftz}},\\
  \label{construc-t0-PhisGa}
  \normLiVec{\la\Phitz''^*\Gatz^{*,+}\begin{pmatrix}\bun_N\\ 0\end{pmatrix}} &\leq \frac{4(2\RW)^{2N}}{(2N)!} \abs{\la\etaW^{(2N)}(0)} t^{2N-2}  \mbox{ by Proposition~\ref{asympt-etaV2}}.
\end{align}
In the last equation, we have used the fact that $\abs{\prod_{j\neq i}(\posz_i-\posz_j)^2}\leq (2\RW)^{2N}$ (where $\RW=\sup\{\normLiVec{z};z\in\Bo\}$ is defined in Equation~\eqref{def-RW}).

We also know that for some constants $\Lgpiphiseconde,\constPi>0$  which only depend on $\phi$ and $\aposO$,
\begin{align*}  
  \normLiVec{\Phitz''^*\GPi_{\post} w} &\leq \normObs{w} \Lgpiphiseconde  t^{2N-2} \mbox{ by Lemma~\ref{lem-dl-gpiphiseconde}},\\
  \mbox{and }	\normLiVec{\Phitz''^*\GPi_{\post} \GatzO \begin{pmatrix}a_0\\0\end{pmatrix}} &\leq\Lgpiphiseconde\constPi \frac{\DeltO}{4}  t^{4N-2} \mbox{ by Lemma~\ref{lem-dl-gpiphiseconde} and Lemma~\ref{lem-proj-gamma}},
\end{align*}
for all $\posz\in\Bo=\bball{\poszO}{\frac{\DeltO}{4}}$, $t\in (0,t_0)$, where $\DeltO$ is defined in~\eqref{min-sep-dist-ini}.

Combining those inequalities with~\eqref{construc-t0-PhisGa}, we see that 
\begin{align*}
\normLiVec{\Phitz''^*q_{\post}} &= \normLiVec{\la\Phitz''^*\Gatz^{*,+}\begin{pmatrix}\bun_N\\ 0\end{pmatrix} + \Phitz''^*\GPi_{\post} w + \Phitz''^*\GPi_{\post} \GatzO \begin{pmatrix}a_0\\0\end{pmatrix} }\\
&\leq \abs{\la}\frac{4}{(2N)!} \abs{\etaW^{(2N)}(0)} t^{2N-2} +  \normObs{w}  \Lgpiphiseconde  t^{2N-2} + \Lgpiphiseconde\constPi \frac{\DeltO}{4}  t^{4N-2}
\end{align*}

On the other hand, since 
\begin{align*}
	\Htz^{*,-1} &= \diag\pa{1,\ldots, 1/t^{2N-1}} \Hz^{*,-1} \diag\pa{1,\ldots,1,t,\ldots,t},\\
\qandq \OmPhi &=
	\begin{pmatrix} 
		0 & 0 \\ 
		0 & -\diag\pa{ \Phitz''^* q_{\post} }
  \end{pmatrix},
\end{align*}
we get   
\begin{align*}
  \Htz^{*,-1} \OmPhi \diagIta^{-1}  \Htz^{-1} 
  &=t\diag\pas{1,\ldots,\frac{1}{t^{2N-1}}}\Hz^{*,-1} \OmPhi \\
  &{} \diagIa^{-1}\Hz^{-1}\diag\pas{1,\ldots,\frac{1}{t^{2N-1}}}
\end{align*}
so that
\begin{align*}
  \norm{t\Htz^{*,-1} \OmPhi \diagIta^{-1}  \Htz^{-1}}
  &\leq \frac{\normLiVec{a^{-1}}}{t^{4N-4}} \norm{\Hz^{*,-1}}\norm{\Hz^{-1}} \normLiVec{\Phitz''^*q_{\post}}\\
  &\leq C\left(\frac{\abs{\la}}{t^{2N-2}}\frac{4(2R)^{2N}}{(2N)!} \abs{\etaW^{(2N)}(0)}\right.\\ 
  &{} + \left. \frac{\normObs{w}}{t^{2N-2}}\Lgpiphiseconde + \constPi\Lgpiphiseconde \frac{\DeltO}{4}  t^{2}\right)
\end{align*}
with $C=\sup_{(\amp,\posz)\in\Ba\times\Bo}\norm{\Hz^{*,-1}}\norm{\Hz^{-1}}\normLiVec{a^{-1}}$.

Possibly choosing $t_0$ a bit smaller, we may assume that 
\eq{
	0\leq C\constPi\Lgpiphiseconde \frac{\DeltO}{4}  t_0^{2}<\frac{1}{8\norm{(\Fdn^*\Fdn)^{-1}}}.
}
As a consequence, there exists $\constTh>0$ such that for all $t\in(0,t_0)$, and all $(\amp,\posz)\in\Ba\times\Bo$,
\begin{align*}
 &\left(\max\left(\frac{\abs{\la}}{t^{2N-2}},\frac{\normObs{w}}{t^{2N-2}}\right)\leq \constTh\right) \\
&\Longrightarrow \norm{t\Htz^{*,-1} \OmPhi \diagIta^{-1}  \Htz^{-1}}\leq \frac{1}{4\norm{(\Fdn^*\Fdn)^{-1}}}. 
\end{align*}
Then, recalling~\eqref{construc-t0-GF} and setting
\begin{align*}
r&\eqdef \norm{t\pa{\GF_{\post}^*\GF_{\post}}^{-1}\Htz^{*,-1} \OmPhi \diagIta^{-1}  \Htz^{-1}}\\
&\leq 2\norm{(\Fdn^*\Fdn)^{-1}} \frac{1}{4\norm{(\Fdn^*\Fdn)^{-1}}}=\frac{1}{2},
\end{align*}
we see that the matrix~\eqref{huge-matrix} is invertible, and 
\begin{align*}
  \norm{\pa{\Id_{2N}+t\pa{\GF_{\post}^*\GF_{\post}}^{-1}\Htz^{*,-1} \OmPhi \diagIta^{-1}  \Htz^{-1}}^{-1}}
  \leq \sum_{k=0}^{+\infty}r^k
  = \frac{1}{1-r} \leq \frac{3}{2}.
\end{align*}

Eventually, using \eqref{construc-t0-GF} again, we obtain that the norm of the inverse of \eqref{huge-matrix} is less than $3\norm{(\Fdn^*\Fdn)^{-1}}$.

Now, if $\feq(\apos,\param)=0$, then $\Phitz'^*(\Phitz \amp-\PhitzO \ampO - w)=0$, so that thanks to Lemma~\ref{reg-f} we obtain
\begin{align*}
	\partial_{\apos} \feq(\apos,\param)&= 	\Gatz^*\Gatz \diagIta+t\begin{pmatrix}
	0 & 0 \\
	0 & \diag(\Phitz''^*(\Phitz \amp-\PhitzO \ampO - w)) 
	\end{pmatrix}.
\end{align*}
Moreover, 
\begin{align*}
	&\Phitz \amp-\PhitzO \ampO - w =\Gatz\Amp - \GatzO\AmpO-w \\
		& \qquad = \Gatz^{*,+}\Gatz^*\GatzO\AmpO+ \Gatz^{*,+}\Gatz^*w-\la \Gatz^{*,+}\begin{pmatrix}\bun_N \\ 0\end{pmatrix} - \GatzO\AmpO-w \\
    & \qquad= - \GPi_{\post} \GatzO \AmpO - \GPi_{\post} w - \la \Gatz^{*,+}\begin{pmatrix}\bun_N \\ 0\end{pmatrix} = -q_{\post}.
\end{align*}
As a result,
\begin{align*}
	\partial_{\apos} \feq(\apos,\param) 
	&= \Htz^*\GF_{\post}^*\GF_{\post}\Htz \diagIta + t  \begin{psmallmatrix} 
		0 & 0 \\ 0 & -\diag\pa{
			\Phitz''^*
      q_{\post}
		}
  \end{psmallmatrix}\\
		&= \Htz^* \hugemat \Htz\diagIta,
\end{align*}
and $\partial_{\apos} \feq(\apos,\param)$ is invertible.
\end{proof}

We may now study the variations of $\gs$.

\begin{cor}\label{maj-dg}
If $\injdn$ holds and $\phi\in\kernel{2N+1}$ then there exists $\constdg>0$ (which only depends on $\phi$ and $\aposO$), such that for $0<t<t_0$, for all $\param\in\Vs$ 
\eq{
	\norm{\d \gs(\param)}\leq \frac{\constdg}{t^{2N-1}}.
}
\end{cor}

\begin{proof}
  Let us recall that by construction, $\Vs\subset  \ball{0}{\constTh t^{2N-2}}$. Thus, from Lemma~\ref{maj-df}, we know that for all $\param\in\Vs$,	$\partial_{\apos} \feq(\fimp(\param),\param) = \Htz^* \hugemat \Htz\diagIta$, 
  where $(\amp,\posz)=\gs(v)$.
Since $\d \gs(\param)=-\pa{ \partial_{\apos} \feq(\gs(\param),\param) }^{-1} \partial_{\param} \feq(\gs(\param),\param)$, we get
\begin{align*}
	\d \gs (\param) &= -\diagIta^{-1}\Htz^{-1} \hugemat^{-1} \Htz^{*,-1} 
		\pa{ \begin{pmatrix}\bun_N \\ 0\end{pmatrix}, -\Htz^*\GF_{\post}^* } \\
		&= \diagIa^{-1}\Hz^{-1} \diag\pas{1,\ldots,\frac{1}{t^{2N-1}}} \hugemat^{-1}  
		\left(
			\dirac{2N}, \GF_{\post}^*
		\right),
\end{align*}
Using Lemma~\ref{maj-df} and the fact that $\diagIa^{-1},\Hz^{-1},\GF_{\post}$ are uniformly bounded on $\Ba\times\Bo$, we obtain the claimed upper bound of $\norm{\d \gs(\param)}$ for all $\param\in\Vs$.
\end{proof}

We are now in position to prove that $\Vs$ contains a ball of radius of order $t^{2N-1}$.

\begin{prop}\label{prop-tfi-ext}
  If $\injdn$ holds and $\phi\in\kernel{2N+1}$, there exists $\constR>0$ such that for all $t\in(0,t_0)$,
  \eq{
    	\ball{0}{\constR t^{2N-1}}\subset \Vs \qwithq \constR \geq \min\left(\frac{\DeltO}{4\constdg},\frac{\min_i(\amp_{0,i})}{4\constdg},\frac{\constTh}{t_0}\right).
  }
\end{prop}

\begin{proof}
  Let $\param\in \RR\times \Obs$ with unit norm (i.e. $\max(\la,\normObs{w})=1$), and define 
  \begin{align*}
    R_\param\eqdef\sup \enscond{r\geq 0}{rv\in \Vs}.
  \end{align*}
  Clearly $0<R_\param\leq \constTh t^{2N-2}$. Assume that $R_\param< \constTh t^{2N-2}$.
  Then by Corollary~\ref{maj-dg}, $\gs$ is uniformly continuous on $\Vs$, so that the value of $\gs(R_\param\param)$ can be defined as a limit, and $\feq(\gs(R_\param\param),R_\param\param)=0$.

  By contradiction, if $\gs(R_\param\param)\in \ba\times\bo$, then by Lemma~\ref{maj-df}, we may apply the Implicit Function Theorem to obtain a neighborhood of $(\gs(R_\param\param),R_\param)$ in which $\gs$ may be extended. This enables us to construct an open set $V\in \Vv$ (in particular we may ensure that $V$ is star-shaped with respect to $0$) such that $\Vs\subsetneq V$, which contradicts the maximality of $\Vs$.

Hence, $\gs(R_\param\param)\in \partial(\ba\times\bo)=\left(\partial(\ba)\times\Bo\right)\cup \left(\Ba\times\partial(\bo)\right)$. Assume for instance that $\gs(R_\param\param)\in \Ba\times\partial(\bo)$ (the other case being similar). Then, for $(\amp,\posz)=\gs(R_\param\param)$,
\begin{align*}
  \frac{\DeltO}{4}=\normLiVec{\posz-\poszO}\leq \int_0^1 \normLiVec{\d \gs(sR_\param\param)\cdot R_\param\param}\d s\leq \frac{\constdg}{t^{2N-1}}R_\param,
\end{align*}
which yields $R_\param\geq \frac{\DeltO}{4\constdg}t^{2N-1}$. Similarly, if $\gs(R_\param\param)\in\partial(\ba)\times\Bo$, we may prove that $R_\param\geq \frac{\min_i(\amp_{0,i})}{4\constdg}t^{2N-1}$.

Eventually, we have proved that for all $\param\in \RR\times \Obs$ with unit norm, 
\eq{
	R_\param\geq \min\left(\frac{\DeltO}{4\constdg}t^{2N-1},\frac{\min_i(\amp_{0,i})}{4\constdg}t^{2N-1},\constTh t^{2N-2}\right), 
}
and the claimed result follows.
\end{proof}

\subsection{Continuity of $\gs$ at $0$}

Before moving to Section~\ref{sec-etaL} and to the proof of $\normLi{\etaL}\leq1$ (which ensures that $\meast$ is a solution to the BLASSO), we give a first order expansion of our candidate solution $\apos=\gs(\param)$ for all $v\in B(0,\constR t^{2N-1})$.

\begin{prop}\label{first-order-dev-sol}
If $\injdn$ holds and $\phi\in\kernel{2N+1}$ then for all $t\in(0,t_0)$, $\param\in B(0,\constR t^{2N-1})$,
\begin{align}\label{lip-gt}
	\norm{\gs(\param)-\gs(0)}\leq M\pa{\frac{|\la|}{t^{2N-1}}+\frac{\normObs{w}}{t^{2N-1}}}.
\end{align}
\end{prop}

\begin{proof}
To show Equation~\eqref{lip-gt}, it suffices to write
\eq{
	\gs(\param)=\gs(0)+\int_0^1 \d\gs(s\param)\cdot \param \d s,
}
and use Corollary~\ref{maj-dg} to conclude.

\end{proof}


\section{Convergence of $\etaL$ to $\etaW$}
\label{sec-etaL}

It remains to prove that $\meast$ where $(\amp,\posz)=\gs(\param)$ is indeed a solution to $\blasso$. To show this statement, we prove that $\etaL$ converges towards $\etaW$ when $(t,\la,w)\to 0$ in a well chosen domain. This section is devoted to this result. The proof of our main contribution, Theorem~\ref{thm-main}, which can be found in Section~\ref{sec-main-contrib}, uses this convergence result and the assumption that $\etaW$ is $(2N-1)$-non-degenerate to conclude.




\begin{prop}\label{etaL-nondegen}
  Assume that $\varphi\in\kernel{2N+1}$ and that $\injdn$ holds, and let $\constW>0$ be the constant defined in Theorem~\ref{if-etaw-nondegen}, $\gs$, $t_0>0$ and $\constR>0$ be the function and constants defined in Section~\ref{sec-tfi}.

Then there exist constants $t_1\in (0,t_0)$ and $C>0$ (which depend only on $\phi$ and $\aposO$) such that for all $t\in\pa{0,t_1}$ and for all $(\la,w)\in \ball{0}{\constR t^{2N-1}}$ with $\normObs{\frac{w}{\la}}\leq C$, the following inequalities hold
\eq{
\forall \ell\in \{0,\ldots, 2N\},\quad  \normLi{\etaL^{(\ell)}-\etaW^{(\ell)}}\leq \constW,
}
with $\etaL=\Phi^*\left(\frac{1}{\la} \pa{ \Phi_{\postO}\ampO+w-\Phi_{\post}\amp}\right)$ and $(\amp,\posz)=\gs(\la,w)$.
\end{prop}

\begin{proof}
  Let $t\in(0,t_0)$, $\param\in \ball{0}{\constR t^{2N-1}}$, and $(\amp,\posz)=\apos=\gs(\param)$.
Then, using $\feq(\apos,\param)=0$ (see~\eqref{eq-implicitfunc}), we get
\begin{align*}
	\pL&\eqdef \frac{1}{\la} \pa{ \Phi_{\postO}\ampO+w-\Phi_{\post}\amp}\\
  &=\frac{1}{\la}\pa{\GatzO\AmpO+w-\Gatz\Amp}\\
  &=\Gatz^{*,+}\begin{pmatrix}\bun_N \\ 0\end{pmatrix}+\GPi_{\post} \frac{w}{\la} + \frac{1}{\la}\GPi_{\post}\GatzO\AmpO.
\end{align*}
Hence,
\begin{align*}
	\normObs{\pL-\pW} \leq \normObs{\Gatz^{*,+}\begin{pmatrix}\bun_N \\ 0\end{pmatrix}-\pW} 
	+ 
	\normObs{\GPi_{\post} \frac{w}{\la}} 
	+ 
	\normObs{\frac{1}{\la}\GPi_{\post}\GatzO\AmpO}
\end{align*}
From Lemma~\ref{lem-dl-ftz}, there exists $\constV>0$ and $t_V>0$ (which only depends on $\phi,\aposO$) such that for all $t\in(0,t_V)$ and all $\posz\in \Bo$,
\eq{
	\normObs{\Gatz^{*,+}\begin{pmatrix}\bun_N \\ 0\end{pmatrix}-\pW}\leq \constV t.
}
Moreover, since $\GPi_{\post}$ is an orthogonal projector,
\eq{
	\normObs{\GPi_{\post} \frac{w}{\la}}\leq \normObs{\frac{w}{\la}},
}
and by Lemma~\ref{lem-proj-gamma}, there exists $\constPi>0$ (which only depends on $\phi,\aposO$) such that, for all $t\in(0,t_0)$
\begin{align*}
	\normObs{\frac{1}{\la}\GPi_{\post}\GatzO\AmpO}&\leq \frac{\constPi}{|\la|}t^{2N} \abs{\posz-\poszO}_\infty \\
			&\leq \frac{\constPi}{|\la|}t^{2N} \frac{\constdg}{t^{2N-1}} (|\la|+\normObs{w}) \mbox{ by Proposition~\ref{first-order-dev-sol}} \\
			&\leq \constPi \constdg t \pa{1+\normObs{\frac{w}{\la}}}.
\end{align*}
Gathering all these upper-bounds, one obtains
\begin{align*}
	\normObs{\pL-\pW} &\leq (\constV+\constPi \constdg)t + (1 + \constPi \constdg t) \normObs{\frac{w}{\la}}\\
			&\leq (\constV+\constPi \constdg)t + (1+ \constPi \constdg) \normObs{\frac{w}{\la}}
\end{align*}
Now, denoting by $\Phi^{(\ell)}: \radon\rightarrow \Obs$ the operator $m\mapsto \int_\Pos \phi^{(\ell)}(\pos)\d m(\pos)$ and by ${\Phi^{(\ell)}}^*:\Obs\rightarrow \ContX$ its adjoint (so that $\etaL^{(\ell)}={\Phi^{(\ell)}}^*\pL$ and $\etaW^{(\ell)}={\Phi^{(\ell)}}^*\pW$), we let
\eq{
  K \eqdef \underset{0\leq l\leq 2N}{\max}\underset{\pos\in\Pos}{\sup}\normObs{\phi^{(\ell)}(\pos)},
}
which satisfies $K<+\infty$ because $\phi\in\kernel{2N+1}$.
Then, for all $\ell \in \{0,\ldots, 2N\}$, 
\begin{align*}
\normLi{\etaL^{(\ell)}-\etaW^{(\ell)}} &\leq K \normObs{\pL-\pW}\\
&\leq K\left( (\constV+\constPi \constdg)t + (1+ \constPi \constdg) \normObs{\frac{w}{\la}}\right).
\end{align*}
As a consequence, by taking $t$ smaller than $\min(t_0,t_V,\frac{\constW}{2K(\constV+\constPi \constdg)})$ and for all $(\la,w)\in B(0,\constR t^{2N-1})$ such that
\eq{
	(1 + \constPi \constdg) \normObs{\frac{w}{\la}}\leq\frac{\constW}{2K},
}
we get
\eq{
  \normLi{\etaL^{(\ell)}-\etaW^{(\ell)}} \leq \constW.
}
\end{proof}

\begin{rem}
  	The constants involved in Proposition~\ref{etaL-nondegen} are 
  	\eql{\label{eq-cst-nondeg}
  	t_1\eqdef \min(t_0,t_V,\frac{\constW}{2K(\constV+\constPi \constdg)})
	\qandq
	C \eqdef \frac{\constW}{2K(1+\constPi \constdg)}. 
	}
	They only depend on $\phi$ and $\aposO$.
\end{rem}

\section{Conclusion}

In this paper, we have proposed a detailed analysis of the recovery performance of the BLASSO for positive measures. 
We have shown that if the signal-to-noise ratio is of the order of $1/t^{2N-1}$, then the BLASSO achieves perfect support estimation. 
This results nicely matches both Cramer-Rao lower-bounds~\cite{vetterli-sparse2008}, bounds for combinatorial approaches~\cite{demanet-recoverability2014} and practical 
performances of the MUSIC algorithm~\cite{condat-cadzow2013}. It is also close to the upper-bound obtained by~\cite{candes-stable2014} for stable recovery on a grid by the LASSO. We have 
hence shown that these bounds do not only imply stability, they imply exact support recovery for the BLASSO, if $\etaW$ is $(2N-1)$-non-degenerate.
We have observed numerically that this condition holds for the ideal low-pass filter and that this is the case for a large class of low-pass filters. We have showed that this is true for the Gaussian kernel by providing the expression of $\etaW$.

\section*{Acknowledgements}

We would like to thank Laurent Demanet for his initial questions that motivated us to prove the main result of this paper. 
We would also like to thank Jean-Marie Mirebeau for stimulating discussions about the intriguing structure of the inverse of checkerboard matrices. 
This work has been partly supported by the European Research Council (ERC project SIGMA-Vision).

\appendix

\section{Proof of the Results of Section~\ref{sec-exact-support-recov}}

\subsection{Proof of Proposition~\ref{hyp-fourier}}
\label{sec-proof-hyp-fourier}

The functions $(\phiD{0},\ldots,\phiD{k})$ are linearly independent in $\Ldeux(\TT)$ if and only if their respective Fourier coefficients are linearly independent in $\ell^2(\ZZ)$. If $(c_n[\phiD{0}])_{n\in\ZZ}$ denotes the Fourier coefficients of $\phiD{0}$, the Fourier coefficients of $\phiD{j}$ are given by $\left((2\imath\pi n )^j c_n[\phiD{0}]\right)_{n\in\ZZ}$  (with the convention that $0^0=1$).
 
If $\phiD{0}$ has $k+1$ nonzeros Fourier coefficients corresponding to pairwise distinct frequencies $(n_0,\ldots n_k)$, those Fourier coefficients are given by the matrix product
 \eq{
   \begin{pmatrix}
     c_{n_0}[\phiD{0}] & &0 \\
                   &\ddots&\\
    0               &      &c_{n_k}[\phiD{0}]
   \end{pmatrix}
 \begin{pmatrix}
  1 & (2\imath\pi n_0) & (2\imath\pi n_0)^2 & \ldots & (2\imath\pi n_0)^{k} \\
  1 & (2\imath\pi n_1) & (2\imath\pi n_1)^2 & \ldots & (2\imath\pi n_1))^{k} \\
  \vdots & \vdots & \vdots & \ddots & \vdots \\
  1 & (2\imath\pi n_k) & (2\imath\pi n_k)^2 & \ldots & (2\imath\pi n_k)^{k} \\
 \end{pmatrix}
 }
Both the diagonal and the Vandermonde matrices are invertible, hence the family of Fourier coefficients of $(\phiD{0}, \ldots, \phiD{k})$ is linearly independent.

Conversely, if $\injk$ holds, one can find $k+1$ Fourier coefficients, corresponding to some frequencies $n_0,\ldots n_k$, such that the matrix $(c_{n_\ell}[\phiD{j}])_{0\leq \ell,j\leq k}$ is invertible. From the above factorization, we deduce that each $c_{n_\ell}[\phiD{0}]$ must be nonzero for $0\leq \ell\leq k$.

\subsection{Proof of Theorem~\ref{if-etaw-nondegen}}
\label{proof-if-etaw-nondegen}

The proof proceeds in two steps. First we show the result locally around $0$ in $\Pos$ thanks to $\etaW^{(2N)}(0)<0$ (because $\etaW^{(2N)}(0)\neq 0$ and $|\etaW|<1$ on $\Pos\setminus\{0\}$) and then we extend the result to all $\Pos$ thanks to $|\etaW|<1$ on $\Pos\setminus\{0\}$.

\paragraph{Locally.}

Let us prove that there exist $\constW>0$, $\tW>0$ such that for all $t\in (0,\tW)$, $\posz\in\RR^N$ with pairwise disctinct coordinates and $\normLiVec{\posz}\leq\RW$, there exist $r^+>0$ with $r^+>\underset{1\leq i\leq N}{\max} \tW\posz_i$ and $r^-<0$ with $r^-<\underset{1\leq i\leq N}{\min} \tW\posz_i$ such that for all $\eta \in \precert{2N}$ satisfying for all $1\leq i\leq N$, $\eta(\post_i)=1$ and $\eta'(\post_i)=0$, the following implication holds,
\eq{\begin{split}
&\left(  \forall \ell\in\{0,\ldots, 2N\}, \normLi{\eta^{(\ell)}-\etaW^{(\ell)}}\leq \constW\right)\\
&\Longrightarrow \left( 	\forall \pos\in (r^-,r^+)\setminus\bigcup_i \{\post_i\}, \quad
  |\eta(\pos)|<1 \qandq \forall 1\leq i \leq N, \ \eta''(\post_i)<0\right).
\end{split}
}

First, we prove that, provided $\constW>0$, $\tW>0$ are small enough, $\eta'$ has exactly $2N-1$ zeros in $(r^-,r^+)$.

Let $t>0$ and $\eta \in \precert{2N}$ and $\posz \in \RR^N$ with pairwise distinct coordinates and $\normLiVec{\posz}\leq\RW$ such that for all $1\leq i\leq N$, $\eta(\post_i)=1$ and $\eta'(\post_i)=0$. We suppose that $\posz_1<\posz_2<\ldots<\posz_N$. By Rolle's Theorem, for all $1\leq i\leq N-1$, there exists $c_i(t)\in (\post_i,\post_{i+1})$ such that $\eta'(c_i(t))=0$. As a result $\eta'$ has at least $2N-1$ zeros in all $(r^-,r^+)$ satisfying the requirements.

Now, let us assume by contradiction that $\eta'$ has strictly more than $2N-1$ zeros for arbitrarily small values of $\constW$, $\tW$ and in all $(r^-,r^+)$ satisfying the requirements.
As a result, there are sequences $(t_k)_{k\in\NN}$ where $t_k\to 0$, $(\posz_k)_{k\in\NN}$ (where each $z_k\in\RR^N$ has pairwise distinct coordinates and $\normLiVec{\posz_k}\leq\RW$), $(r^+_k)_{k\in\NN}$ and $(r^-_k)_{k\in\NN}$ where for all $k\in \NN$, $r^+_k>0$, $r^+_k>\underset{1\leq i\leq N}{\max} t_k (\posz_k)_i$, $r^+_k\to0$ and $r^-_k<0$, $r^-_k<\underset{1\leq i\leq N}{\min} t_k(\posz_k)_i$ and $r^-_k\to0$. And there exists $(\eta_k)_{k\in\NN}\in \pa{\precert{2N}}^{\NN}$ such that for all $k\in\NN$ 
\begin{align*}
  &\forall i\in\{1,\ldots, N-1\}, \ \eta_k(t_k(\posz_k)_i)=1 \qandq \eta_k'(t_k(\posz_k)_i)=0, \\
  &\forall \ell\in\{0,\ldots, 2N\}, \ \normLi{\eta_k^{(\ell)}-\etaW^{(\ell)}}\leq \frac{1}{k},
\end{align*}
and $\eta_k'$ has at least $2N$ zeros in $(r^-_k,r^+_k)$ (we already know that $\eta_k'$ has at least $2N-1$ zeros in $(r^-_k,r^+_k)$).
Thus, for all $k\in \NN$, by applying successively Rolle's Theorem, we obtain that there exists $\pos_k\in (r^-_k,r^+_k)$ such that $\eta_k^{(2N)}(\pos_k)=0$. Since $\pos_k\to0$ (because $r^-_k,r^+_k\to0$) and $\normLi{\eta_k^{(2N)}-\etaW^{(2N)}}\leq \frac{1}{k}$, we deduce that $\eta_W^{(2N)}(0)=0$, which is a contradiction. Hence  $\eta'$ has exactly $2N-1$ zeros in some $(r^-,r^+)$.

Using the same argument, we may also prove that for all $i\in\{1,\ldots, N\}$, $\eta''(\post_i)\neq0$. Let us now observe that, either for all $1\leq i \leq N$, $\eta''(\post_i)>0$ or, for all $1\leq i \leq N$, $\eta''(\post_{i+1})<0$. Indeed, assume for instance by contradiction that for some $1\leq i\leq N-1$,  $\eta''(\post_i)>0$ and $\eta''(\post_{i+1})<0$. Then, there exists $c_i(t)\in (\post_i,\post_{i+1})$ such that $\eta(c_i(t))=1$. Applying Rolle's Theorem on respectively $(\post_i,c_i(t))$ and $(c_i(t),\post_{i+1})$, we obtain that $\eta'$ vanishes at least twice in $(\post_i,\post_{i+1})$. It is a contradiction with the fact that $\eta'$ has exactly $2N-1$ zeros in $(r^-,r^+)$ for all $0<t<\tW$.

As a result, there are only two possibilities: either $\eta''(\post_i)<0$ for all $i$ (and then for all $\pos\in (r^-,r^+)\setminus\bigcup_i \{\post_i\}$  $\eta(\pos)<1$), or $\eta''(\post_i)>0$ for all $i$ (and then for all $\pos\in (r^-,r^+)\setminus\bigcup_i \{\post_i\}$ $\eta(\pos)>1$). But since $\etaW^{(2N)}(0)<0$, there is some $\tilde{\pos}\in (r^-,r^+)$ and some $\varepsilon_0>0$ such that $\etaW(\tilde{\pos}) < 1 - \varepsilon_0$. Choosing $\constW$ small enough so that $\constW<\varepsilon_0/2$, we obtain that $\eta(\tilde{\pos})<1-\varepsilon_0/2$. As a consequence, for all $\pos\in (r^-,r^+)\setminus\bigcup_i \{\post_i\}$, $\eta(\pos)<1$. Finally we can suppose that $\etaW(\pos)>-1$ on $(r^-,r^+)$ by imposing $0<\constW<\inf_{\Pos} \etaW+1$.

To sum up, we have proved the following :
\begin{align}\label{eta-vois1}
	\forall \pos\in (r^-,r^+)\setminus\bigcup_i \{\post_i\}, |\eta(\pos)|<1 \qandq \forall 1\leq i \leq N, \ \eta''(\post_i)<0.
\end{align}

\paragraph{Globally.}

As $\etaW$ is non-degenerate, we have $\sup_{\Pos\setminus (r^-,r^+)}\abs{\etaW}<1$. We can assume that $0<\constW<(1-\sup_{\Pos\setminus (r^-,r^+)}\abs{\etaW})/2$ and use $\normLi{\eta-\etaW}\leq \constW$ so as to obtain 
\begin{align}\label{eta-vois2}
	\forall \pos \in \Pos\setminus(r^-,r^+), \quad |\eta(\pos)|<1.
\end{align}
Gathering Equations \eqref{eta-vois1} and \eqref{eta-vois2}, we obtain the claimed result.

\subsection{Proof of Proposition~\ref{etaw-locnondegen}}
\label{sec-proof-etaw-locnondegen}

The proof of Proposition~\ref{etaw-locnondegen} relies on the study of the structure of the matrices $\Fk^*\Fk$. We start by introducing the notion of checkerboard matrices.

\begin{defn}\label{def-checkerboard}
 Let $n\in\NN^*$ and $A=(a_{i,j})_{i,j} \in \RR^{n\times n}$. We say that $A$ is a \emph{checkerboard} matrix if for all $(i,j)$ such that $i+j$ is odd, $a_{i,j}=0$.
\end{defn}

For an odd $n$, such a matrix looks like 
\eq{
	\begin{pmatrix}
		a_{1,1} & 0 & a_{1,3} & 0 & \ldots & 0 & a_{1,n-2} & 0 & a_{1,n} \\
		0 & a_{2,2} & 0 & a_{2,4} & \ldots & a_{2,n-3} & 0 & a_{2,n-1} & 0 \\
		\vdots & \vdots & \vdots & \vdots & \ddots & \vdots & \vdots & \vdots & \vdots \\
		0 & a_{n-1,2} & 0 & a_{n-1,4} & \ldots & a_{n-1,n-3} & 0 & a_{n-1,n-1} & 0 \\
		a_{n,1} & 0 & a_{n,3} & 0 & \ldots & 0 & a_{n,n-2} & 0 & a_{n,n}
	\end{pmatrix}.
}

\begin{lem}\label{checkerboard_struct}
 The set of \emph{checkerboard} matrices of size $n \in \NN^*$ is an algebra over $\RR$. As a result the inverse of a \emph{checkerboard} matrix is also a \emph{checkerboard} matrix.
\end{lem}
\begin{proof}
 The only difficulty is to show that the product of two checkerboard matrices is also a checkerboard matrix. Let $A=(a_{i,j})_{i,j}$ and $B=(b_{i,j})_{i,j}$ be two such matrices. Let $(i,j)$ such 
 that $i+j$ is odd. Then $\sum_{k=1}^n a_{i,k}b_{k,j} = 0$ because if $i+k$ is even then $(i+j)-(i+k)=j-k$ is odd and $j-k+2k=j+k$ is odd, hence $b_{j,k}=0$. On the contrary, if $i+k$ is odd then $a_{i,k}=0$. So $AB$ is a checkerboard matrix. 
 The last statement holds because the inverse of any matrix is a polynomial in that matrix.
\end{proof}

Now we give a more precise result on the structure of some particular checkerboard matrices.

\begin{lem}\label{pfaffien}
Let $A=(a_{i,j})_{1\leq i,j\leq 2n+1}$ be a symmetric positive-definite checkerboard matrix of size $2n+1$ such that 
\eq{
  a_{i,j}=(-1)^{\frac{i-j}{2}} a_{\frac{i+j}{2},\frac{i+j}{2}}
}
for all $i,j\in\{1,\ldots, 2n+1\}$ such that $i+j$ is even, and let  $(b_{i,j})_{1\leq i,j\leq 2n+1}$ denote the entries of $A^{-1}$. 
Then $b_{2n+1,2n+1}$ and $b_{1,2n+1}$ are positive.
\end{lem}

\begin{proof}
 The fact that $b_{2n+1,2n+1}$ is positive, is a direct consequence of the fact that $A^{-1}$ is symmetric positive-definite because so is $A$.
 
From the expression of the inverse of a matrix using cofactors, we see that $b_{1,2n+1}=\det(\hat{A}_{1,2n+1})/\det(A)$, where $\hat{A}_{1,2n+1}$ is the matrix 
 obtained from $A$ by removing the first row and the last column. Since $A$ is symmetric positive-definite, $\det(A)>0$, and we need only show that $\det(\hat{A}_{1,2n+1})>0$. 
 But one can see that $\hat{A}_{1,2n+1}$ is a skew symmetric matrix of size $2n$, hence 
 \eq{
	\det(\hat{A}_{1,2n+1})=\mbox{pf}(\hat{A}_{1,2n+1})^2,
 }
 where $\mbox{pf}(\hat{A}_{1,2n+1})$ is the Pfaffian of the matrix $\hat{A}_{1,2n+1}$. For more details on the Pfaffian (definition and proof 
 of the result used here), see \cite{lang-algebra2002}.
As a result, $\det(\hat{A}_{1,2n+1}) \geq 0$.  
 
 Moreover, $\hat{A}_{1,2n+1}$ is invertible. Indeed if we denote the colums of $A$ (resp. the columns $\hat{A}_{1,2n+1}$) by $C_1,\ldots,C_{2n+1}$ (resp. 
 $\hat{C}_1, \ldots,\hat{C}_{2n}$), we observe that for $0\leq i \leq n-1$, $\hat{C}_{2i+1}\in E_1$ and $\hat{C}_{2(i+1)}\in E_2$ where
 \begin{align*}
  E_1 &= \Span\{ e_{2j+1} ; \ 0 \leq j \leq n-1 \}, \\
  E_2 &= \Span\{ e_{2(j+1)} ; \ 0 \leq j \leq n-1 \},
 \end{align*}
 with $(e_i)_{1\leq i\leq 2n}$ is the canonical base of $\RR^{2n}$. Since $E_1$ and $E_2$ are in direct sum, if $\hat{A}_{1,2n+1}$ is not invertible it means that either 
 $(\hat{C}_{2i+1})_{0\leq i \leq n-1}$ or $(\hat{C}_{2(i+1)})_{0\leq i \leq n-1}$ is linearly dependent. But it cannot be $(\hat{C}_{2(i+1)})_{0\leq i \leq n-1}$ because it would 
 imply that $(C_{2(i+1)})_{0\leq i \leq n-1}$ is also linearly dependent, because
 \eq{
	C_{2(i+1)}=\begin{pmatrix}
	                0 \\
	                \hat{C}_{2(i+1)}
	             \end{pmatrix},
 }
 which would contradict the invertibility of $A$.
 
 So it means that $(\hat{C}_{2i+1})_{0\leq i \leq n-1}$ must be linearly dependent. However from the structure of the matrix $A$, one can see that for all $0\leq i \leq n-1$
 \eq{
  \hat{C}_{2i+1}=
  \begin{pmatrix}
   0 \\
   (-1)^{i+1} a_{2i+1,1} \\
   0 \\
   -(-1)^{i+1} a_{2i+3,1} \\
   \vdots \\
   0 \\
   (-1)^{i+1} a_{2(i+n)-1,1} \\
   0 \\
   -(-1)^{i+1} a_{2(i+n)+1,1} 
  \end{pmatrix}
  \qandq
  \hat{C}_{2(i+1)}=
  \begin{pmatrix}
   (-1)^{i+1} a_{2i+1,1} \\
   0 \\
   -(-1)^{i+1} a_{2i+3,1} \\
   0 \\
   \vdots \\
   (-1)^{i+1} a_{2(i+n)-1,1} \\
   0 \\
   -(-1)^{i+1} a_{2(i+n)+1,1} \\
   0
  \end{pmatrix}.
 }
 Thus, a linear combination between the elements of $(\hat{C}_{2i+1})_{0\leq i \leq n-1}$ gives the same linear combination between the elements of 
 $(\hat{C}_{2(i+1)})_{0\leq i \leq n-1}$, which contradicts for the same reason as before the invertibility of $A$.
 
 Hence $\hat{A}_{1,2n+1}$ is invertible and $b_{1,2n+1}>0$.
\end{proof}

The next result describes the structure of the matrices $\Fk^*\Fk$ when $\Phi$ is a convolution operator (\ie for all $\pos\in\RR$, $\phi(\pos)=\tilde{\phi}(\cdot-\pos)$) and $k\in\NN^*$ is odd.

\begin{lem}\label{damier_gramFk}
 Suppose that $\Phi$ is a convolution operator. Then for all $k\in \NN^*$, $\Fk^* \Fk$ is a checkerboard matrix. Moreover if $\injdnu$ holds, then the entries indexed by $(1,2N+1)$ and $(2N+1,2N+1)$ of $(\Fdnu^*\Fdnu)^{-1}$ are positive.
\end{lem}

\begin{proof}
 Let $(i,j)$ such that $i+j$ is odd and for example $i>j$ ($\Fk^* \Fk$ is symmetric so it does not matter). The entry $(i,j)$ of $\Fk^* \Fk$ is equal to
 \eq{
	\dotObs{ \phiD{i-1} }{ \phiD{j-1} }.
 }
Recall that in the setting of Proposition~\ref{etaw-locnondegen}, $\Obs=L^2(\Pos)$ and $\dotObs{f}{g}=\int_X fg$. Successive integrations by parts (we integrate the left term and derive the right term) yield
 \begin{align*}
	\dotObs{ \phiD{i-1} }{ \phiD{j-1} } = (-1)^{\frac{i-j+1}{2}} \dotObs{ \phiD{\frac{i+j-1}{2}-1} }{ \phiD{\frac{i+j+1}{2}-1} }, \\
	\dotObs{ \phiD{i-1} }{ \phiD{j-1} } = (-1)^{\frac{i-j-1}{2}} \dotObs{ \phiD{\frac{i+j+1}{2}-1} }{ \phiD{\frac{i+j-1}{2}-1} }.
 \end{align*}
Hence, by the symmetry of the scalar product, we obtain that 
\eq{
	\dotObs{ \phiD{i-1} }{ \phiD{j-1} } = -\dotObs{ \phiD{i-1} }{ \phiD{j-1} }. 
}
This implies that $\dotObs{ \phiD{i-1} }{ \phiD{j-1} }=0$ for odd values of $i+j$, so that $\Fk^*\Fk$ is a checkerboard matrix (see Definition~\ref{def-checkerboard}).
 
Moreover, similar computations show that 	when $i+j$ is even
\eq{
  \dotObs{ \phiD{i-1} }{ \phiD{j-1} }=(-1)^{\frac{i-j}{2}}\dotObs{ \phiD{\frac{i+j}{2}-1} }{ \phiD{\frac{i+j}{2}-1} }.
}
If $\injdnu$ holds, $\Fdnu^*\Fdnu$ is symmetric positive-definite. As a result $\Fdnu^*\Fdnu$ satisfies the assumptions of Lemma \ref{pfaffien}, so the entries $(1,2N+1)$ and $(2N+1,2N+1)$ of $(\Fdnu^*\Fdnu)^{-1}$ are positive.
\end{proof}

We can now prove that $\etaW^{(2N)}(0)<0$ when $\Phi$ is a convolution operator and $\injdnu$ holds.

\begin{proof}[Proof of Proposition~\ref{etaw-locnondegen}]
Since $\etaW=\Phi^*\pW$, we deduce that 
  \begin{align*}
   \etaW^{(2N)}(0) &= \dotObs{\phiD{2N}}{\pW}
   = \dotObs{\phiD{2N}}{ \Fdn(\Fdn^*\Fdn)^{-1}\dirac{2N} }.
  \end{align*}
  Consider the symmetric positive definite matrix
  \begin{align*}
    \Fdnu^*\Fdnu= \begin{pmatrix} \Fdn^*\Fdn & \Fdn^* \phiD{2N}\\
    [ \Fdn^* \phiD{2N} ]^* & \normObs{\phiD{2N}}^2
    \end{pmatrix} \in \RR^{2N \times 2N}.
  \end{align*}
  Observe that $\Fdn^*\Fdn$ is invertible, and that 
  \eq{
  	S \eqdef \normObs{\phiD{2N}}^2 - 
	\dotObs{ \phiD{2N} }{
		\Fdn(\Fdn^*\Fdn)^{-1}\Fdn^*\phiD{2N}
		} \neq 0. 
  }
Indeed, 
\begin{align*}
  S  = 
  \normObs{(\Id_{\Obs} - \GPin)\phiD{2N}}^2>0
\end{align*}
since $(\phiD{0}, \ldots \phiD{2N})$ has full rank.

Thus, we may apply the block inversion formula, and $(\Fdnu^*\Fdnu)^{-1}$ is of the form
\begin{align*}
&  (\Fdnu^*\Fdnu)^{-1}= \frac{1}{\normObs{(\Id_{\Obs} - \GPin)\phiD{2N}}^2} \times \\
  & \qquad
  \begin{pmatrix}
    \ast\ast\ast   & -(\Fdn^*\Fdn)^{-1} \Fdn^*\phiD{2N} \\
    -[(\Fdn^*\Fdn)^{-1} \Fdn^*\phiD{2N}]^* & 1
  \end{pmatrix}
\end{align*}
Lemma~\ref{damier_gramFk} ensures that the entry $(2N+1,1)$ of $(\Fdnu^*\Fdnu)^{-1}$ is (strictly) positive. This precisely means that  
\begin{align*}
  -\frac{
    \dotObs{ \phiD{2N} }{ \Fdn(\Fdn^*\Fdn)^{-1}\dirac{2N} }
  }{
  	\normObs{(\Id_{\Obs}-\GPin)\phiD{2N}}^2
	}
	>0,
\end{align*} and as a result $\etaW^{(2N)}(0)=\dotObs{\phiD{2N}}{\pW} <0$.
\end{proof}

\subsection{Proof of Proposition~\ref{etaW-gaussian-nondegen}}
\label{sec-proof-gaussian-nondegen}

Let denote $(\igk)_{0\leq k\leq 2N-1}$ the coefficients of the first column of the matrix $(\Fdn^*\Fdn)^{-1}$. Then we know that for all $x\in\RR$,
\eq{
	\etaW(x)=\sumkN \igdk \phi \star \phi^{(2k)}(x).
}
Since $\phi\star\phi=x\to \sqrt{\pi}e^{-x^2/4}$, we get that for all $x\in\RR$,
\begin{align}\label{etaW-formu-inter}
	\etaW(x)=\sqrt{\pi}e^{-x^2/4}\sumkN \igdk\hHrdk(x),
\end{align}
where $\hHrdk$ is the Hermite polynomial of order $2k$ associated to $x\to e^{-x^2/4}$.

Now let us show~\eqref{etaW-closedform-gaussian} recursively on $N\in\NN^*$. If $N=1$, we have that $\etaW(x)=e^{-x^2/4}$ because $\Fdn^*\Fdn=(\sqrt{\pi})$. Suppose that the property is true for some $N\in\NN^*$, \ie,
\begin{align}\label{rec-hyp}
	\etaWN(x)=e^{-x^2/4}\sumkN \frac{x^{2k}}{2^{2k}k!},
\end{align}
where we use the notation $\etaWN$ to recall that this is the function $\etaW$ for $N$ spikes. Then form for all $x\in\RR$,
\eq{
	\kappa(x)\eqdef e^{x^2/4}\pa{\etaWNu(x)-\etaWN(x)}.
}
$\kappa$ is a polynomial of degree $2N$ (thanks to~\eqref{etaW-formu-inter}), it satisfies $\kappa(0)=0$ and for all $1\leq i\leq 2N-1$, $\kappa^{(i)}(0)=0$. As a result for all $x\in\RR$,
\eq{
	\kappa(x)=\la x^{2N},
}
for some $\la\in\RR$.

It remains to show that $\la=\frac{1}{2^{2N}N!}$. Remark that for all $x\in\RR$, $\kappa^{(2N)}(x)=\la(2N)!$ and on the other hand $\kappa^{(2N)}(0)=-\etaWN^{(2N)}(0)$ because $\etaWNu^{(2N)}(0)=0$ by definition of $\etaWNu$. So it is enough to show that,
\begin{align}\label{etaW2N-formu}
	\etaWN^{(2N)}(0)=-\frac{(2N)!}{2^{2N}N!}.
\end{align}
Thanks to the Leibniz formula applied to~\eqref{rec-hyp}, we have the following formula,
\begin{align*}
	\etaWN^{(2N)}(0)=\sumkN \binom{2N}{2k}\hHr_{2N-2k}(0)\frac{(2k)!}{2^{2k}k!}.
\end{align*}
Since $\hHr_{2N-2k}(0)=(-1)^{N-k}\frac{(2N-2k-1)!}{2^{2N-2k-1}(N-1-k)!}$, thanks to Lemma~\ref{lem-coeffconstant-herm} below, we obtain,
\begin{align*}
	\etaWN^{(2N)}(0)&=\sumkN \frac{(2N)!}{(2N-2k)!(2k)!}\cdot(-1)^{N-k}\frac{(2N-2k-1)!}{2^{2N-2k-1}(N-1-k)!}\cdot\frac{(2k)!}{2^{2k}k!}\\
		&= -\frac{(2N)!}{2^{2N}N!}\sumkN \frac{N(N-1)!}{(N-1-k)!k!}\cdot\frac{2\cdot(-1)^{N-1-k}}{(2N-2k)}\\
		&=-\frac{(2N)!}{2^{2N}N!} N \sumkN \binom{N-1}{k}(-1)^{N-1-k} \int_0^1 x^{N-1-k} \d x \\
		&=-\frac{(2N)!}{2^{2N}N!} N \underbrace{\int_0^1 (1-x)^{N-1} \d x}_{=1/N}.\\
\end{align*}
This ends the recursive proof and thus we have~\eqref{etaW-closedform-gaussian}.

To conclude that $\etaW$ is $(2N-1)$-non-degenerate, it remains to show that for all $x\in\RR^*$, $|\etaW(x)|<1$, since we already know that $\etaW^{(2N)}(0)<0$ thanks to~\eqref{etaW2N-formu}. But we have, thanks to~\eqref{etaW-closedform-gaussian}, that for all $x\in\RR^*$, $0<\etaW(x)<1$ since $\sumkN \frac{x^{2k}}{2^{2k}k!}$ is the truncated power series of $e^{x^2/4}$.

\begin{lem}\label{lem-coeffconstant-herm}
One has for all $k\in\NN$,
\begin{align}\label{coeffconstant-her}
	\hHr_{2k+2}(0)=(-1)^{k+1}\frac{(2k+1)!}{2^{2k+1}k!}.
\end{align}
\end{lem}
\begin{proof}
We know that $\Hr_{k+2}(x)=x\Hr_{k+1}(x)-(k+1)\Hrk(x)$ where $\Hrk$ is the Hermite polynomial of order $k$ associated to $x\mapsto e^{-x^2/2}$. Thus $\Hr_{2k+2}(0)=-(2k+1)\Hr_{2k}(0)$ and then,
\begin{align}\label{inter}
	\Hr_{2k+2}(0)=(-1)^{k+1}\frac{(2k+1)!}{2^k k!}.
\end{align}
Now, remark that $\hHrk(x)=2^{-k/2}\Hrk(x/\sqrt{2})$, so that together with \eqref{inter} we get the expected result~\eqref{coeffconstant-her}.
\end{proof}

\subsection{Proof of Theorem~\ref{thm-necessary}}
\label{sec-proof-necessary-condition-etaW}

For all $n\in\NN$, the solution $\mn\eqdef\measn$ of $\blasson$ satisfies the following first order optimality conditions:
\begin{align}\label{first-order-cond}
\Gatzn^*(\Phitzn\amp_n -\PhitzO\ampO-w_n)+\la_n \begin{pmatrix}\bun_N\\0\end{pmatrix}=0.
\end{align}
Since $\etaLn=\Phi^*(\frac{1}{\la_n}(\PhitzO \ampO +w_n-\Phitzn\amp_n))$ satisfies $\normLi{\etaLn}\leq1$, $\etaLn\in\Im\Phi^*$ and $\etaLn(t\posz_n)=\sign(\amp_n)=1$.

As a result, by taking $n$ large enough if necessary, we know that $(\amp_n,\posz_n)=\fimp^*(\la_n,w_n)$ with $\fimp^*\in\Cder{2N}(\Vt^*)$ so that 
$(\amp_n,\posz_n)\to(\ampO,\poszO)$. Therefore all the asymptotic results that we established are true when applied for $(\amp_n,\posz_n)$ for $n$ large enough.
Observe from~\eqref{first-order-cond}, that we get for all $n\in\NN$:
\begin{align*}
p_{\la_n,t_n}=\Gatzn^{*,+}\begin{pmatrix}\bun_N \\ 0\end{pmatrix}+\Pi_{t_n\posz_n} \frac{w_n}{\la_n} + \frac{1}{\la_n}\Pi_{t_n\posz_n}\GatzO\AmpO.
\end{align*}
As a consequence:
\eq{
	\normObs{p_{\la_n,t_n}-\pW}\to0 \quad \mbox{when} \quad n\to+\infty,
} 
so that for all $0\leq l\leq 2N$,
\begin{align}
	\normLi{\etaLn^{(l)}-\etaW^{(l)}}\to0 \quad \mbox{when} \quad n\to+\infty.
\end{align}
In particular, since $\normLi{\etaL}\leq1$, we deduce~\eqref{etaW-normi}:
\eq{
	\normLi{\etaW}\leq1.
}

\bibliographystyle{plain}
\bibliography{biblio}

\begin{thebibliography}{10}

\bibitem{bhaskar-atomic2011}
B.N. Bhaskar and B.~Recht.
\newblock Atomic norm denoising with applications to line spectral estimation.
\newblock {\em In 2011 49th Annual Allerton Conference on Communication,
  Control, and Computing}, pages 261--268, September 2011.

\bibitem{vetterli-sparse2008}
T.~Blu, P.~L. Dragotti, M.~Vetterli, P.~Marziliano, and L.~Coulot.
\newblock Sparse sampling of signal innovations: Theory, algorithms and
  performance bounds.
\newblock {\em IEEE Signal Processing Magazine}, 25(2):31--40, 2008.

\bibitem{bredies-inverse2013}
K.~Bredies and H.K. Pikkarainen.
\newblock Inverse problems in spaces of measures.
\newblock {\em ESAIM: Control, Optimisation and Calculus of Variations},
  19(1):190--218, 2013.

\bibitem{candes-superresolution2013}
E.J. Cand{\`e}s and C.~Fernandez-Granda.
\newblock Super-resolution from noisy data.
\newblock {\em Journal of Fourier Analysis and Applications}, 19(6):1229--1254,
  2013.

\bibitem{candes-towards2013}
E.J. Cand{\`e}s and C.~Fernandez-Granda.
\newblock Towards a mathematical theory of super-resolution.
\newblock {\em Communications on Pure and Applied Mathematics}, 67(6):906--956,
  2014.

\bibitem{claerbout-robust1973}
J.F. Claerbout and F.~Muir.
\newblock {Robust modeling with erratic data}.
\newblock {\em Geophysics}, 38(5):826--844, 1973.

\bibitem{condat-cadzow2013}
L.~Condat and A.~Hirabayashi.
\newblock Cadzow denoising upgraded: A new projection method for the recovery
  of dirac pulses from noisy linear measurements.
\newblock {\em Preprint hal-00759253}, 2013.

\bibitem{deCastro-exact2012}
Y.~de~Castro and F.~Gamboa.
\newblock Exact reconstruction using beurling minimal extrapolation.
\newblock {\em Journal of Mathematical Analysis and Applications}, 395(1):336
  -- 354, 2012.

\bibitem{den-resolution1997}
A.J.~Den Dekker, A.~Van den Bos, et~al.
\newblock Resolution: a survey.
\newblock {\em JOSA A}, 14(3):547--557, 1997.

\bibitem{demanet-recoverability2014}
L.~Demanet and N.~Nguyen.
\newblock The recoverability limit for superresolution via sparsity.
\newblock {\em preprint arXiv:1502.01385}, 2015.

\bibitem{donoho-superresolution1992}
D.L. Donoho.
\newblock Super-resolution via sparsity constraints.
\newblock {\em SIAM J. Math. Anal., 23(5):1309-1331}, 9 1992.

\bibitem{duval-exact2013}
V.~Duval and G.~Peyr{\'e}.
\newblock Exact support recovery for sparse spikes deconvolution.
\newblock {\em To appear in Foundation and Computational Mathematics}, 2015.

\bibitem{fernandez-support2013}
C.~Fernandez{-}Granda.
\newblock Support detection in super-resolution.
\newblock {\em Proc. Proceedings of the 10th International Conference on
  Sampling Theory and Applications}, pages 145--148, 2013.

\bibitem{azais-spike2014}
Y.~de~Castro J.-M.~Aza{\"i}s and F.~Gamboa.
\newblock Spike detection from inaccurate samplings.
\newblock {\em Applied and Computational Harmonic Analysis}, 38(2):177--195,
  2015.

\bibitem{krim-two1996}
H.~Krim and M.~Viberg.
\newblock Two decades of array signal processing research: the parametric
  approach.
\newblock {\em Signal Processing Magazine, IEEE}, 13(4):67--94, Jul 1996.

\bibitem{lang-algebra2002}
S.~Lang.
\newblock {\em Algebra}.
\newblock Graduate Texts in Mathematics. Springer New York, 2002.

\bibitem{levy-reconstruction1981}
S.~Levy and P.~Fullagar.
\newblock Reconstruction of a sparse spike train from a portion of its spectrum
  and application to high resolution deconvolution.
\newblock {\em Geophysics}, 46(9):1235--1243, 1981.

\bibitem{liao-music2014}
W.~Liao and A.~Fannjiang.
\newblock {MUSIC} for single-snapshot spectral estimation: Stability and
  super-resolution.
\newblock {\em CoRR}, abs/1404.1484, 2014.

\bibitem{candes-stable2014}
V.~I. Morgenshtern and E.~J. Cand{\`e}s.
\newblock Super-resolution of positive sources: the discrete setup.
\newblock {\em preprint arXiv:1504.00717}, 2014.

\bibitem{santosa-linear1986}
F.~Santosa and W.W. Symes.
\newblock {Linear Inversion of Band-Limited Reflection Seismograms}.
\newblock {\em SIAM Journal on Scientific and Statistical Computing},
  7(4):1307--1330, 1986.

\bibitem{schmidt-multiple1986}
R.O. Schmidt.
\newblock Multiple emitter location and signal parameter estimation.
\newblock {\em Antennas and Propagation, IEEE Transactions on}, 34(3):276--280,
  Mar 1986.

\bibitem{chen-atomic1998}
D.L.~Donoho S.S.~Chen and M.A. Saunders.
\newblock Atomic decomposition by basis pursuit.
\newblock {\em SIAM Journal on Scientific Computing}, 20(1):33--61, 1998.

\bibitem{tibshirani-regression1994}
R.~Tibshirani.
\newblock {Regression Shrinkage and Selection Via the Lasso}.
\newblock In {\em Journal of the Royal Statistical Society, Series B},
  volume~58, pages 267--288, 1994.

\end{thebibliography}

\end{document}